\documentclass[a4paper,UKenglish,cleveref, thm-restate]{lipics-v2021}

\hideLIPIcs  

\title{Towards exact structural thresholds for parameterized complexity} 

\author{Falko Hegerfeld}{Humboldt-Universit\"at zu Berlin, Germany}{hegerfeld@informatik.hu-berlin.de}{https://orcid.org/0000-0003-2125-5048}{Partially supported by DFG Emmy Noether-grant (KR 4286/1).}
\author{Stefan Kratsch}{Humboldt-Universit\"at zu Berlin, Germany}{kratsch@informatik.hu-berlin.de}{https://orcid.org/0000-0002-0193-7239}{}

\authorrunning{F. Hegerfeld and S. Kratsch} 
\Copyright{Falko Hegerfeld and Stefan Kratsch} 

\ccsdesc[500]{Theory of computation~Parameterized complexity and exact algorithms}
\ccsdesc[300]{Mathematics of computing~Graph algorithms} 

\keywords{Parameterized complexity, lower bound, vertex cover, odd cycle transversal, SETH, modulator, treedepth, cliquewidth} 

\category{} 

\nolinenumbers 

\EventEditors{Holger Dell and Jesper Nederlof}
\EventNoEds{2}
\EventLongTitle{17th International Symposium on Parameterized and Exact Computation (IPEC 2022)}
\EventShortTitle{IPEC 2022}
\EventAcronym{IPEC}
\EventYear{2022}
\EventDate{September 7--9, 2022}
\EventLocation{Potsdam, Germany}
\EventLogo{}
\SeriesVolume{249}
\ArticleNo{9}

\usepackage{multirow}
\usepackage{tabularx, environ}
\usepackage{xspace}

\theoremstyle{definition}
\newtheorem{dfn}{Definition}[section]

\theoremstyle{plain}
\newtheorem{mythm}[dfn]{Theorem}
\newtheorem{cor}[dfn]{Corollary}
\newtheorem{lem}[dfn]{Lemma}

\crefname{dfn}{Definition}{Definitions}
\crefname{thm}{Theorem}{Theorems}
\crefname{cor}{Corollary}{Corollaries}
\crefname{lem}{Lemma}{Lemmata}
\crefname{prop}{Proposition}{Propositions}
\crefname{rem}{Remark}{Remarks}
\crefname{algorithm}{Algorithm}{Algorithms}
\crefname{algocf}{Algorithm}{Algorithms}

\crefformat{equation}{(#2#1#3)}
\crefformat{figure}{#2Figure~#1#3}
\crefformat{section}{#2Section~#1#3}

\newcommand{\Oh}{\mathcal{O}} 
\newcommand{\NP}{\ensuremath{\textsf{NP}}\xspace}
\newcommand{\SETH}{SETH\xspace}

\newcommand{\NN}{\mathbb{N}}
\newcommand{\EF}{\mathcal{T}} 
\newcommand{\TT}{\mathcal{T}} 
\newcommand{\bag}{\mathbb{B}}
\newcommand{\powerset}{\mathcal{P}}
\newcommand{\tfa}{\text{ for all }}
\newcommand{\sep}{:}

\newcommand{\DS}{\textsc{Dominating Set}\xspace}
\newcommand{\TDS}{\textsc{Total Dominating Set}\xspace}
\newcommand{\VC}{\textsc{Vertex Cover}\xspace}
\newcommand{\OCT}{\textsc{Odd Cycle Transversal}\xspace}
\newcommand{\COL}{\textsc{Coloring}\xspace}
\newcommand{\SAT}{\textsc{Satisfiability}\xspace}
\newcommand{\HS}{\textsc{Hitting Set}\xspace}
\newcommand{\MC}{\textsc{Maximum Cut}\xspace}
\newcommand{\LC}{\textsc{List Coloring}\xspace}
\newcommand{\LHom}{\textsc{List Homomorphism}\xspace}
\newcommand{\CVC}{\textsc{Connected Vertex Cover}\xspace}
\newcommand{\ncols}{r}
\newcommand{\DTC}{\textsc{Deletion to $\ncols$-Colorable}\xspace}
\newcommand{\KFD}{\textsc{$K_\ncols$-free Deletion}\xspace}

\DeclareMathOperator{\lcw}{lin-cw} 
\DeclareMathOperator{\cw}{cw} 
\DeclareMathOperator{\tw}{tw} 
\DeclareMathOperator{\pw}{pw} 
\DeclareMathOperator{\td}{td} 
\DeclareMathOperator{\ctw}{ctw} 
\DeclareMathOperator{\tctw}{tc-tw} 
\DeclareMathOperator{\tcpw}{tc-pw} 
\DeclareMathOperator{\tcctw}{tc-ctw} 
\DeclareMathOperator{\tctd}{tc-td} 
\DeclareMathOperator{\nd}{nd} 
\newcommand{\modulator}{X}
\newcommand{\tcm}{\mathcal{X}}
\newcommand{\tcpartition}{\Pi_{tc}}

\newcommand{\intro}[1]{{#1}}
\newcommand{\union}{\cup}
\newcommand{\relab}[2]{\rho_{#1 \rightarrow #2}}
\newcommand{\join}[2]{\eta_{#1,#2}}

\newcommand{\lfct}{\mathtt{lab}}
\newcommand{\tree}{\mathcal{T}}
\newcommand{\cwexp}{\mu}

\newcommand{\DP}{A}
\newcommand{\psols}{\mathcal{Q}}

\newcommand{\sol}{X}
\newcommand{\cover}{Y}
\newcommand{\budget}{b}
\newcommand{\coloring}{\varphi}
\newcommand{\lclist}{\Lambda}
\newcommand{\del}{\bot}
\newcommand{\delset}{{\coloring^{-1}(\del)}}
\newcommand{\algo}{\mathcal{A}}

\newcommand{\formula}{\sigma}
\newcommand{\tassign}{\tau}
\newcommand{\grpsize}{p}
\newcommand{\vgrpsize}{{p_0}}
\newcommand{\ngrps}{t}
\newcommand{\nvars}{n}
\newcommand{\nclss}{m}
\newcommand{\clss}{q}
\newcommand{\embedding}{\kappa}
\newcommand{\eps}{\varepsilon}

\newcommand{\guards}{Q}
\newcommand{\goodsets}{\mathcal{S}}
\newcommand{\states}{\mathbf{states}}
\newcommand{\one}{\mathbf{1}}
\newcommand{\two}{\mathbf{2}}
\newcommand{\three}{\mathbf{3}}
\newcommand{\four}{\mathbf{4}}

\newcommand{\crit}{\mathcal{H}}
\newcommand{\arrow}{A}
\newcommand{\decode}{B}
\newcommand{\ttcset}{\mathcal{U}}
\newcommand{\setfam}{\mathcal{S}}
\newcommand{\scolorings}{\Phi}
\newcommand{\packing}{\mathcal{P}}
\newcommand{\cost}{\mathrm{cost}}

\newcommand{\UU}{U}
\newcommand{\family}{\mathcal{F}}
\newcommand{\hsbudget}{t}
\newcommand{\hssol}{H}

\usepackage{tikzit}

\tikzstyle{filled}=[fill=black, draw=black, shape=circle, inner sep=1pt, minimum size=5]
\tikzstyle{big empty}=[fill=white, draw=black, shape=rectangle, inner sep=2pt, minimum size=20pt]
\tikzstyle{small empty}=[fill=white, draw=black, shape=circle, inner sep=0pt, minimum size=20pt, line width=0.75pt]
\tikzstyle{rect}=[fill=black, draw=black, shape=rectangle, inner sep=1pt, minimum size=5pt]
\tikzstyle{big empty red}=[fill=white, draw={rgb,255: red,171; green,0; blue,60}, shape=rectangle, inner sep=2pt, minimum size=20pt]
\tikzstyle{mid empty}=[fill=white, draw=black, shape=rectangle, inner sep=2pt, minimum size=20pt]
\tikzstyle{mid empty red}=[fill=white, draw={rgb,255: red,171; green,0; blue,60}, shape=rectangle, inner sep=2pt, minimum size=20pt]
\tikzstyle{small empty red}=[fill={rgb,255: red,227; green,159; blue,178}, draw={rgb,255: red,171; green,0; blue,60}, shape=circle, inner sep=0pt, minimum size=20pt, line width=0.75pt]
\tikzstyle{red rect}=[fill={rgb,255: red,171; green,0; blue,60}, draw={rgb,255: red,171; green,0; blue,60}, shape=rectangle, minimum size=5pt]
\tikzstyle{big circle}=[fill=white, draw=black, shape=circle, line width=1pt, minimum size=25pt]
\tikzstyle{filled red}=[fill={rgb,255: red,171; green,0; blue,60}, draw={rgb,255: red,171; green,0; blue,60}, shape=circle, inner sep=1pt, minimum size=5pt]
\tikzstyle{state 0}=[fill=white, draw=black, shape=circle, line width=0.75pt, minimum size=10pt]
\tikzstyle{state L}=[fill={rgb,255: red,25; green,61; blue,182}, draw=black, shape=circle, line width=0.75pt, minimum size=10pt]
\tikzstyle{state R}=[fill={rgb,255: red,191; green,0; blue,64}, draw=black, shape=circle, line width=0.75pt, minimum size=10pt]
\tikzstyle{state unknown}=[fill=black, draw=black, shape=circle, line width=0.75pt, minimum size=10pt]
\tikzstyle{tiny empty}=[fill=white, draw=black, shape=circle, inner sep=0pt, minimum size=12pt, line width=0.75pt]
\tikzstyle{tiny filled}=[fill=black, draw=black, shape=circle, inner sep=0pt, minimum size=12pt, line width=0.75pt]
\tikzstyle{wide rectangle}=[fill=white, draw=black, shape=rectangle, inner sep=2 pt, minimum height=20pt, minimum width=30pt, line width=0.75pt]
\tikzstyle{empty 30pt}=[fill=white, draw=black, shape=circle, inner sep=0pt, minimum size=30pt, line width=0.75pt]
\tikzstyle{empty 20pt}=[fill=white, draw=black, shape=circle, inner sep=0pt, minimum size=20pt, line width=0.75pt]
\tikzstyle{filled 10pt}=[fill=black, draw=black, shape=circle, minimum size=10pt]
\tikzstyle{tiny rectangle}=[fill=white, draw=black, shape=rectangle, minimum width=5pt, minimum height=5pt, inner sep=0pt]
\tikzstyle{dash circle}=[fill=white, draw=black, shape=circle, dashed, inner sep=0.5pt]
\tikzstyle{tall}=[fill=white, draw=black, shape=rectangle, minimum width=1cm, minimum height=1.25cm, inner sep=0pt, line width=0.5pt]
\tikzstyle{filled R}=[fill={rgb,255: red,191; green,0; blue,64}, draw=black, shape=circle, minimum size=5pt, inner sep=1pt]
\tikzstyle{filled B}=[fill={rgb,255: red,0; green,0; blue,208}, draw=black, shape=circle, minimum size=5pt, inner sep=1pt]
\tikzstyle{filled G}=[fill={rgb,255: red,0; green,186; blue,0}, draw=black, shape=circle, minimum size=5pt, inner sep=1pt]
\tikzstyle{filled W}=[fill=white, draw=black, shape=circle, minimum size=5pt, inner sep=1pt]

\tikzstyle{thick}=[-, line width=1pt]
\tikzstyle{thick arrow}=[->, line width=1pt]
\tikzstyle{packing}=[-, line width=1pt, draw={rgb,255: red,171; green,0; blue,60}]
\tikzstyle{dotted}=[-, line width=1pt, dashed]
\tikzstyle{dashed cyan thick}=[-, dashed, line width=1pt, draw={rgb,255: red,2; green,200; blue,200}]
\tikzstyle{white}=[-, draw=white]
\tikzstyle{del}=[-, double, line width=1pt]
\tikzstyle{thin}=[-, line width=0.5pt]
\tikzstyle{red dashed}=[-, line width=1pt, draw={rgb,255: red,191; green,0; blue,64}, dashed]
\tikzstyle{dark green}=[-, line width=1pt, draw={rgb,255: red,44; green,127; blue,35}]
\tikzstyle{thin green}=[-, line width=0.5pt, draw={rgb,255: red,27; green,125; blue,52}]
\tikzstyle{thin red}=[-, line width=0.5pt, draw={rgb,255: red,181; green,0; blue,3}]
\tikzstyle{thin dashed}=[-, line width=0.5pt, dashed]

\begin{document}

\maketitle

\thispagestyle{empty}

\begin{abstract}
Parameterized complexity seeks to optimally use input structure to obtain faster algorithms for NP-hard problems. This has been most successful for graphs of low treewidth, i.e., graphs decomposable by small separators: Many problems admit fast algorithms relative to treewidth and many of them are optimal under the Strong Exponential-Time Hypothesis (SETH). Fewer such results are known for more general structure such as low clique-width (decomposition by  large and dense but structured separators) and more restrictive structure such as low deletion distance to some sparse graph class.

Despite these successes, such results remain ``islands'' within the realm of possible structure. Rather than adding more islands, we seek to determine the transitions between them, that is, we aim for structural thresholds where the complexity increases as input structure becomes more general. Going from deletion distance to treewidth, is a single deletion set to a graph with simple components enough to yield the same lower bound as for treewidth or does it take many disjoint separators? Going from treewidth to clique-width, how much more density entails the same complexity as clique-width? Conversely, what is the most restrictive structure that yields the same lower bound?

For treewidth, we obtain both refined and new lower bounds that apply already to graphs with a single separator $X$ such that $G-X$ has treewidth at most $r=\Oh(1)$, while $G$ has treewidth $|X|+\Oh(1)$. We rule out algorithms running in time $\Oh^*((r+1-\varepsilon)^{k})$ for \textsc{Deletion to $\ncols$-Colorable} parameterized by $k=|X|$; this implies the same lower bound relative to treedepth and (hence) also to treewidth. It specializes to $\Oh^*((3-\varepsilon)^{k})$ for \textsc{Odd Cycle Transversal} where $\tw(G-X)\leq r=2$ is best possible. For clique-width, an extended version of the above reduction rules out time $\Oh^*((4-\varepsilon)^k)$, where $X$ is allowed to be a possibly large separator consisting of $k$ (true) twinclasses, while the treewidth of $G - X$ remains $\ncols$; this is proved also for the more general \textsc{Deletion to $\ncols$-Colorable} and it implies the same lower bound relative to clique-width. Further results complement what is known for \textsc{Vertex Cover}, \textsc{Dominating Set} and \textsc{Maximum Cut}. All lower bounds are matched by existing and newly designed algorithms.
\end{abstract}

\newpage
\setcounter{page}{1}

\section{Introduction}

The goal of parameterized complexity is to leverage input structure to obtain faster algorithms than in the worst case and to identify algorithmically useful structure. The most prominent structural graph parameter \emph{treewidth} measures the size of \emph{separators} decomposing the graph. Many problems admit fast algorithms relative to treewidth and we can often certify their optimality assuming the \emph{Strong Exponential-Time Hypothesis} (\SETH) \cite{CurticapeanLN18, CyganKN18, CyganNPPRW22, LokshtanovMS18}. Such (conditional) \emph{optimality results} allow us to conduct a precise study of the impact of structure on the running time, whereas otherwise the currently best running time might be an artifact due to the momentary lack of algorithmic tools and not inherent to the structure. 

The structure captured by treewidth can be varied in several ways: In the \emph{sparse setting}, we may restrict the interplay of separators and/or allow additional connected components from some graph class $\mathcal{H}$; this yields notions such as treedepth as well as deletion resp.\ elimination distance to $\mathcal{H}$. In the \emph{dense setting}, we may allow large and dense but structured separators; this yields e.g.\ clique-width and rank-width. Conceptually, the difference between parameters may be quite large: if the complexity of a problem changes between two parameters, then it is difficult to pinpoint which structural feature has lead to the change in complexity.

We seek to delineate more \emph{exact structural thresholds} between these parameters. This can be done by designing algorithms relative to more permissible parameters or by establishing the same lower bounds relative to more restrictive parameters. We focus on the latter approach in a fine-grained setting, i.e., all considered problems can be solved in time $\Oh^*(c^k)$\footnote{The $\Oh^*$-notation suppresses factors that are polynomial in the input size.} for some constant $c$ and parameter $k$ and we determine the precise value of the base $c$. 

For parameters other than treewidth far fewer optimality results are known. In particular, to the best of our knowledge, the only known fine-grained optimality results for \NP-hard problems relative to a deletion distance are for $\ncols$-\COL~\cite{JaffkeJ23, LokshtanovMS18}, its generalization \LHom~\cite{PiecykR21}, and isolated results on \VC~\cite{JacobPRS22} and \CVC~\cite{CyganDLMNOPSW16}. The crux is that other lower bound proofs deal with more complex problems (e.g., deletion of vertices, packing of subgraphs, etc.) by copying the same (type of) partial solution over many \emph{noncrossing} separators; this addresses several obstacles but makes the approach unsuitable for deletion distance parameters (or even for treedepth). We show that a much broader range of problems may admit such improved lower bounds by giving the new tight lower bounds for \emph{vertex deletion problems} such as \VC and \OCT relative to deletion distance parameters, in both sparse and dense settings. 

\subparagraph*{Sparse Setting.} 
Our main problem of study is \DTC, i.e., delete as few vertices as possible so that an $\ncols$-colorable graph remains, which specializes to \VC for $\ncols = 1$ and to \OCT for $\ncols = 2$. The first parameterization which we study is the size $|\modulator|$ of a \emph{modulator} $\modulator \subseteq V(G)$, or deletion distance, to treewidth $\ncols$, i.e., $\tw(G - \modulator) \leq \ncols$. Our main result in the sparse setting is the following.

\begin{mythm}
  \label{thm:intro_lb_sparse}
  If there are $\ncols \geq 2$, $\eps > 0$ such that \DTC can be solved in time $\Oh^*((\ncols + 1 - \eps)^{|\modulator|})$, where $\modulator$ is a modulator to treewidth $\ncols$, then \SETH is false.\footnote{We assume that an appropriate decomposition is given, thus strengthening the lower bounds.}
\end{mythm}
The general construction for \DTC, $\ncols \geq 2$, does not work for the case $\ncols = 1$, i.e., \VC, and we fill this gap by providing a simple ad-hoc construction for \VC parameterized by a modulator to pathwidth 2.
\begin{mythm}
  \label{thm:intro_lb_vc}
  If there is an $\eps > 0$ such that \VC can be solved in time $\Oh^*((2 - \eps)^{|\modulator|})$, where $\modulator$ is a modulator to pathwidth 2, then \SETH is false.
\end{mythm} 
These results improve the known lower bounds for \VC and \OCT parameterized by pathwidth and provide new tight lower bounds for $\ncols \geq 3$ as a matching upper bound follows from generalizing the known algorithm for \OCT parameterized by treewidth. Note that in \cref{thm:intro_lb_sparse} the treewidth bound $\ncols$ is the same as the bound $\ncols$ on the number of colors. This treewidth bound, at least for $\ncols = 2$, and the pathwidth bound in \cref{thm:intro_lb_vc} cannot be improved due to upper bounds obtained by Lokshtanov et al.~\cite{LokshtanovNRRS14} for \VC and \OCT parameterized by an odd cycle transversal or a feedback vertex set. Lokshtanov et al.~\cite{LokshtanovMS18} asked if the complexity of problems, other than $\ncols$-\COL (where a modulator to a single path is already sufficient~\cite{JaffkeJ23}), relative to treewidth could already be explained with parameterization by feedback vertex set. As argued, this cannot be true for \VC and \OCT, so our results are essentially the next best explanation. 

Furthermore, the previous two theorems also imply the same lower bound for parameterization by \emph{treedepth}\footnote{\label{footnote:mod_tw_td}If $\tw(G - \modulator) \leq t$, then $\td(G) \leq |\modulator| + (t + 1) \log_2 |V|$, cf.\ Nešetřil and Ossona de Mendez~\cite{NesetrilO12}, and $\Oh^*(c^{\td(G)}) = \Oh^*(c^{|\modulator|}|V|^{(t+1)\log_2 c}) = \Oh^*(c^{|\modulator|})$.}, thus yielding the first tight lower bounds relative to treedepth for vertex selection problems and partially resolving a question of Jaffke and Jansen~\cite{JaffkeJ23} regarding the complexity relative to treedepth for problems studied by Lokshtanov et al.~\cite{LokshtanovMS18}.
\begin{cor}
  \label{thm:intro_lb_td}
  If there is an $\ncols \geq 1$ and an $\eps > 0$ such that \DTC can be solved in time $\Oh^*((\ncols + 1 - \eps)^{\td(G)})$, then \SETH is false.
\end{cor} 

\subparagraph*{Dense Setting.}

Our results on deletion distances can actually be lifted to the \emph{dense setting}. We do so by considering \emph{twinclasses}, which are arguably the simplest form of dense structure. A twinclass is an equivalence class of the \emph{twin}-relation, which says that two vertices $u$ and $v$ are twins if $N(u) \setminus \{v\} = N(v) \setminus \{u\}$, i.e., $u$ and $v$ have the same neighborhood outside of $\{u,v\}$. Given two distinct twinclasses, either all edges between them exist or none of them do. Contracting each twinclass yields the \emph{quotient graph} $G^q$ and we obtain \emph{twinclass-variants} of the usual graph parameters treedepth, cutwidth, pathwidth, and treewidth by measuring these parameters on the quotient graph $G^q$, e.g., the \emph{twinclass-pathwidth} of $G$ is $\tcpw(G) = \pw(G^q)$. The parameters twinclass-pathwidth and twinclass-treewidth have been studied before under the name \emph{modular pathwidth} and \emph{modular treewidth}~\cite{Lampis20, Mengel16, PaulusmaSS16}. Furthermore, we remark that the previously studied parameter \emph{neighborhood diversity} satisfies $\nd(G) = |V(G^q)|$ \cite{Lampis12}. Relationships between parameters transfer to their twinclass-variants and twinclass-pathwidth is more restrictive than \emph{linear-clique-width}. Similarly, we obtain \emph{twinclass-modulators}, but we measure the complexity of the remaining components on the level of the original graph, i.e., a twinclass-modulator (TCM) $\tcm$ to treewidth $\ncols$ is a family $\tcm$ of twinclasses such that $\tw(G - \bigcup_{\modulator \in \tcm} \modulator) \leq \ncols$. We can now state our second main result, which, similarly to the sparse setting, also carries over to twinclass-treedepth.%
\begin{mythm}
 \label{thm:intro_lb_dense}
 If there are $\ncols \geq 2$, $\eps > 0$ such that \DTC can be solved in time $\Oh^*((2^\ncols - \eps)^{|\tcm|})$, where $\tcm$ is a TCM to treewidth $\ncols$, then \SETH is false.

  Additionally, it follows that if there are $\ncols \geq 2$, $\eps > 0$ such that \DTC can be solved in time $\Oh^*((2^\ncols - \eps)^{\tctd(G)})$, then \SETH is false.
\end{mythm}

Due to the inequalities $\cw(G) \leq \tcpw(G) + 3$ and $\pw(G) \leq \td(G)$, cf.\ Lampis~\cite{Lampis20} and Nešetřil and Ossona de Mendez~\cite{NesetrilO12}, we see that $\cw(G) \leq \tctd(G) + 3$. Hence any $\Oh^*(c^{\cw(G)})$-time algorithm also implies a $\Oh^*(c^{\tctd(G)})$-time algorithm. Thus, the following result, relying on standard techniques for dynamic programming on graph decompositions such as the $(\min,+)$-\emph{cover product}, yields a tight upper bound complementing the previous lower bounds.
\begin{mythm}
 \label{thm:intro_algo_dense}
 Given a $k$-clique-expression $\cwexp$ for $G$, \DTC on $G$ can be solved in time $\Oh^*((2^\ncols)^k)$.\footnote{Jacob et al.~\cite{JacobBDP21} have simultaneously proven this upper and lower bound for the special case of \OCT, $\ncols = 2$, parameterized by clique-width. Their construction also proves the lower bound for linear-clique-width, but not for the more restrictive twinclass-treedepth or twinclass-modulator like our construction.}
\end{mythm}
There is no further lower bound result for \VC, since $\ncols + 1 = 2^\ncols$ for $\ncols = 1$ and hence \cref{thm:intro_lb_vc} already yields a tight lower bound for the clique-width-parameterization. 

Going into more detail, the twinclasses of the modulator in the construction for \cref{thm:intro_lb_dense} are \emph{true twinclasses}, i.e., each twinclass induces a clique, and moreover they are of size $\ncols$ (with a small exception). Intuitively, allowing for deletions, there are $2^\ncols$ possible sets of at most $\ncols$ colors that can be assigned to a clique of size $\ncols$, e.g., the empty set $\emptyset$ corresponds to deleting the clique completely. Hence, our results essentially show that it is necessary and optimal to go through all of these color sets for each twinclass in the modulator.

In contrast, consider the situation for $\ncols$-\COL where Lampis~\cite{Lampis20} has obtained tight running times of $\Oh^*\left(\binom{\ncols}{\lfloor \ncols / 2 \rfloor}^{\tctw(G)}\right)$ when parameterized by twinclass-treewidth and of time $\Oh^*((2^\ncols - 2)^{\cw(G)})$ when parameterized by clique-width. Whereas the complexities for $\ncols$-\COL vary between the twinclass-setting and clique-width, this is not the case for \DTC. The base $\binom{\ncols}{\lfloor \ncols / 2 \rfloor}$ is due to the fact that without deletions only color sets of the same size as the considered (true) twinclass can be attained and the most sets are possible when the size is $\lfloor \ncols / 2 \rfloor$. For clique-width, a \emph{label class} may induce more complicated graphs than cliques or independent sets and the interaction between two label classes may also be more intricate. Lampis~\cite{Lampis20} shows that the extremal cases of color sets $\emptyset$ and $[\ncols] = \{1, \ldots, \ncols\}$ can be handled separately, thus yielding the base $2^\ncols - 2$ for clique-width.

\subparagraph*{Additional results.}
As separate results, we obtain the following four results:
\begin{mythm}
  Assuming the \SETH, the following lower bounds hold:
  \begin{itemize}
    \item \DS cannot be solved in time $\Oh^*((4 - \eps)^{\tcctw(G)})$ for any $\eps > 0$.
    \item \TDS cannot be solved in time $\Oh^*((4 - \eps)^{\ctw(G)})$ for any $\eps > 0$.
    \item \MC cannot be solved in time $\Oh^*((2- \eps)^{|\modulator|})$ for any $\eps > 0$, where $\modulator$ is a modulator to treewidth at most $2$.
    \item \KFD cannot be solved in time $\Oh^*((2 - \eps)^{|\modulator|})$ for any $\eps > 0$ and $r \geq 3$, where $\modulator$ is a modulator to treewidth at most $r - 1$.
  \end{itemize}
\end{mythm}
The first result improves the parameterization of the tight lower bound for \DS obtained by Katsikarelis et al.~\cite{KatsikarelisLP19} from linear-clique-width to twinclass-cutwidth. We prove this by reducing \TDS parameterized by cutwidth to \DS parameterized by twinclass-cutwidth and providing a lower bound construction for \TDS parameterized by cutwidth.

Lastly, the lower bound for \VC, \cref{thm:intro_lb_vc}, also implies tight lower bounds for \MC and \KFD, which again imply the same lower bounds parameterized by treedepth. The former also partially answers a question of Jaffke and Jansen~\cite{JaffkeJ23}, by being another problem considered by Lokshtanov et al.~\cite{LokshtanovMS18} whose running time cannot be improved when parameterizing by treedepth instead of treewidth.

\subparagraph*{Technical contribution.}
We start by recalling the standard approach of Lokshtanov et al.~\cite{LokshtanovMS18} to proving tight lower bounds for problems parameterized by pathwidth at a high level. Given a \SAT instance $\formula$, the variables are partitioned into $\ngrps$ groups of constant size. For each variable group, a \emph{group gadget} is constructed that can encode all assignments of this variable group into partial solutions of the considered target problem. The group gadget usually consists of a bundle of long path-like gadgets inducing a sequence of \emph{disjoint} separators. Further gadgets attached to these separators decode the partial solutions and check whether the corresponding assignment satisfies some clause. Ideally, the path gadgets are designed so that a partial solution transitions through a well-defined sequence of states when viewed at consecutive separators. For most problems, the gadgets do not behave this nicely though. For example, in \OCT it is locally always preferable to delete a vertex instead of not deleting it. Such behavior leads to undesired state changes called \emph{cheats}, but for appropriate path gadgets there can only be a constant number of cheats on each path. By making the path gadgets long enough, one can then find a region containing no cheats where we can safely decode the partial solutions.

For problems such as $\ncols$-\COL, all states are equally constraining and such cheats do not occur, hence enabling us to prove the same lower bounds under more restrictive parameters such as feedback vertex set. But for \emph{vertex deletion problems}, like \OCT, these cheats do occur and pose a big issue when trying to compress the path gadgets into a single separator $\modulator$, since deletions in $\modulator$ are highly favorable. On a single separator $\modulator$ such behavior means that one partial solution is \emph{dominating} another and if we cannot control this behavior, then we lose the dominated partial solution for the purpose of encoding group assignments. Concretely, for \OCT we obtain dominating partial solutions by deleting further vertices in the single separator $\modulator$. The number of deletions is bounded from above by the \emph{budget constraint}, but if we limit the number of deletions in $\modulator$, then we do not have $3^{|\modulator|}$ partial solutions anymore and the construction may not be able to attain the desired base in the running time.

To resolve this issue we expand upon a technique of Cygan et al.~\cite{CyganDLMNOPSW16} and construct an instance with a slightly large parameter value, i.e., a slightly larger single separator $\modulator$. Thus, we can limit the number of deletions and are still able to encode sufficiently many group assignments. More precisely, we consider only partial solutions with the same number of deletions in $\modulator$, hence only pairwise non-dominating partial solutions remain. We construct a \emph{structure gadget} to enforce a lower bound on the number of deletions in $\modulator$. A positive side effect is that the remaining gadgets can also leverage the structure of the partial solutions.

In the dense setting and especially for a higher number $\ncols$ of colors, this issue is amplified. Here, we consider the states of twinclasses, instead of single vertices, in a partial solution. For a twinclass, there is a hierarchy of dominating states: any state that does not delete all vertices in the twinclass is dominated by a state that deletes further vertices in the twinclass. For \DTC, the maximum number of states is achieved on a true twinclass of size $\ncols$ and we can partition the states into \emph{levels} based on the number of deletions they induce. Within each level, the states are pairwise non-dominating. Consequently, we restrict the family of partial solutions so that for every level the number of twinclasses with that level is fixed. This requires a considerably more involved construction of the structure gadget which now has to distinguish states based on their level.

\subparagraph*{Related work.}
There is a long line of work relative to treewidth~\cite{BorradaileL16, CurticapeanLN18, CyganKN18, CyganNPPRW22, DubloisLP21a, DubloisLP22, EgriMR18, FockeMR22, KatsikarelisLP22, LokshtanovMS18, MarxSS21, MarxSS22, OkrasaPR20, OkrasaR21} and all of these lower bounds, except for the result by Egri et al.~\cite{EgriMR18}, already apply to pathwidth.
In the sparse setting, there is further work on the parameterization by \emph{cutwidth}~\cite{CurticapeanM16, GroenlandMNS22, JansenN19, MarxSS21, PiecykR21, GeffenJKM20} and by feedback vertex set~\cite{LokshtanovMS18, PiecykR21}. We remark that the works of van Geffen et al.~\cite{GeffenJKM20} and Piecyk and Rz{\k a}\.{z}ewski~\cite{PiecykR21} show that previous lower bounds relative to pathwidth already hold for more restrictive parameterizations. 
In the dense setting, there are some results~\cite{IwataY15, JacobBDP21, KatsikarelisLP19, Lampis20} on parameterization by clique-width and these lower bounds already apply to linear-clique-width, but not to the more restrictive parameters considered here. The work by Iwata and Yoshida~\cite{IwataY15} also provides equivalences between different lower bounds and works under a weaker assumption than \SETH, unfortunately their techniques blow up the modulator too much and are not applicable in our case. 
Finally, the complexity of $\ncols$-\COL and the more general homomorphism problems has been extensively studied~\cite{EgriMR18, FockeMR22, GanianHKOS22, JaffkeJ23, Lampis20, OkrasaPR20, OkrasaR21, PiecykR21}, with only two of these articles~\cite{GanianHKOS22, Lampis20} consider the dense setting. Jaffke and Jansen~\cite{JaffkeJ23} closely study the complexity of $\ncols$-\COL parameterized by the deletion distance to various graph classes $\family$; in particular, the base for treewidth can already be explained by deletion distance to a single path.

On the algorithmic side, the study of heterogeneous parameterizations has been gaining traction~\cite{BulianD16, BulianD17, EibenGHK21, EibenGS18b, EibenGS18a, HolsKP20, JansenKW21}, yielding the notions of $\mathcal{H}$-treewidth and $\mathcal{H}$-\emph{elimination distance}, which is a generalization of treedepth. Currently, only few of these works~\cite{EibenGHK21, JansenKW21} contain algorithmic results that are sufficiently optimized to apply to our fine-grained setting. Jansen et al.~\cite{JansenKW21} show that \VC can be solved in time $\Oh^*(2^k)$ and \OCT in time $\Oh^*(3^k)$ when parameterized by bipartite-treewidth. Eiben et al.~\cite{EibenGHK21} show that \MC can be solved in time $\Oh^*(2^k)$ when parameterized by $\mathcal{R}_w$-treewidth, where $\mathcal{R}_w$ denotes the graphs of \emph{rank-width} at most $w$.

Another line of work is on depth-parameters in the dense setting~\cite{ChenF20, DeVosKO20, GajarskyK20, GajarskyLO13, GanianHNOM19, HlinenyKOO16, KwonMOW21} such as \emph{shrub-depth} and \emph{sc-depth}. The algorithmic results relative to these parameters are largely concerned with meta-results so far~\cite{ChenF20, GanianHNOM19} and their relation to clique-width is not strong enough to preserve the complexity in our fine-grained setting.


\subparagraph*{Organization.}
We discuss the preliminaries and basic notation in \cref{sec:prelims}. In \cref{sec:outline}, we give an outline of our two main results: the lower bound for \DTC in the sparse setting and the dense setting. The relationships between the considered parameters are discussed in \cref{sec:parameters}. The algorithm for \DTC parameterized by clique-width is given in \cref{sec:cw_algo}. \cref{sec:dtc_lb} contains the detailed proofs of the main results. In \cref{sec:vc} and \cref{sec:ds_lb} we present the additional results, i.e., the lower bounds for \VC and \DS. We conclude in \cref{sec:conclusion}. \cref{sec:problems} contains the formal definitions of the considered problems.

\section{Preliminaries}
\label{sec:prelims}

For a positive integer $n$, we define $[n] = \{1, \ldots, n\}$. For a set $S$, we define the \emph{power set} $\powerset(S) = \{T \subseteq S\}$ and for $0 \leq k \leq |S|$, we define $\binom{S}{k} = \{T \subseteq S \sep |T| = k\}$ and $\binom{S}{\leq k} = \{T \subseteq S \sep |T| \leq k\}$ and $\binom{S}{\geq k}$ analogously. For $0 \leq k \leq n$, we define $\binom{n}{\leq k} = |\binom{[n]}{\leq k}|$ and similarly $\binom{n}{\geq k}$. For a set family $\setfam$, we define $\bigcup(\setfam) = \bigcup_{S \in \setfam} S$. For a function $f \colon A \rightarrow C$ and subset $B \subseteq A$, we denote by $f\big|_B$ the \emph{restriction} of $f$ to $B$. For two functions $f,g\colon A \rightarrow B$, we write $f \equiv g$ if $f(a) = g(a)$ for all $a \in A$. For a boolean predicate $p$, we let $[p]$ denote the \emph{Iverson bracket} of $p$, which is $1$ if $p$ is true and $0$ if $p$ is false.

We use common graph-theoretic notation and assume that the reader knows the essentials of parameterized complexity. Let $G = (V, E)$ be an undirected graph. For a vertex set $X \subseteq V$, we denote by $G[X]$ the subgraph of $G$ that is induced by $X$. The \emph{open neighborhood} of a vertex $v$ is given by $N(v) = \{u \in V \sep \{u,v\} \in E\}$, whereas the \emph{closed neighborhood} is given by $N[v] = N(v) \cup \{v\}$. For sets $X \subseteq V$ we define $N[X] = \bigcup_{v \in X} N[v]$ and $N(X) = N[X] \setminus X$. For two disjoint vertex subsets $A, B \subseteq V$, adding a \emph{join} between $A$ and $B$ means adding all edges between $A$ and $B$. For a vertex set $X \subseteq V$, we define $\delta(X) = \{ \{x, y\} \in E \sep x \in X, y \notin X\}$.

An \emph{$r$-coloring} of a graph $G = (V,E)$ is a function $\varphi \colon V \rightarrow [r]$ such that $\varphi(u) \neq \varphi(v)$ for all $\{u,v\} \in E$. We say that $G$ is \emph{$r$-colorable} if there is an $r$-coloring of $G$. The \emph{chromatic number} of $G$, denoted by $\chi(G)$, is the minimum $r$ such that $G$ is $r$-colorable. 

\subparagraph*{Quotients and twins.} Let $\Pi$ be a partition of $V(G)$. The \emph{quotient graph} $G / \Pi$ is given by $V(G / \Pi) = \Pi$ and $E(G / \Pi) = \{\{B_1, B_2\} \sep \exists u \in B_1, v \in B_2 \colon \{u,v\} \in E(G)\}$. We say that two vertices $u, v$ are \emph{twins} if $N(u) \setminus \{v\} = N(v) \setminus \{u\}$. The equivalence classes of this relation are called \emph{twinclasses}. More specifically, if $N(u) = N(v)$, then $u$ and $v$ are \emph{false twins} and if $N[u] = N[v]$, then $u$ and $v$ are \emph{true twins}. Every twinclass of size at least 2 consists of only false twins or only true twins. A false twinclass induces an independent set and a true twinclass induces a clique. Let $\tcpartition(G)$ be the partition of $V(G)$ into twinclasses. 

\subsection{Graph Parameters}

\begin{dfn}
 A \emph{tree/path decomposition} of a graph $G$ consists of a tree/path $\TT$ and bags $\bag_t \subseteq V(G)$ for all $t \in V(\TT)$, such that:
 \begin{itemize}
  \item $V(G) = \bigcup_{t \in V(\TT)} \bag_t$.
  \item For every edge $\{u, v\} \in E(G)$, there is a node $t \in V(\TT)$ such that $\{u, v\} \subseteq \bag_t$.
  \item For every $v \in V(G)$, the set $\{t \in V(\TT) \sep v \in \bag_t\}$ induces a connected subgraph of $\TT$.
 \end{itemize}
 The \emph{width} of a tree/path decomposition is equal to $\max_{t \in V(\TT)} |\bag_t| - 1$. The \emph{treewidth} of $G$, denoted $\tw(G)$, is the minimum possible width of a tree decomposition of $G$. The \emph{pathwidth} of $G$, denoted $\pw(G)$, is the minimum possible width of a path decomposition of $G$.
\end{dfn}

\paragraph*{Clique-Expressions and Clique-Width.}

A \emph{labeled graph} is a graph $G = (V,E)$ together with a \emph{label function} $\lfct \colon V \rightarrow \NN = \{1, 2, 3, \ldots\}$. We say that a labeled graph is \emph{$k$-labeled} if $\lfct(v) \leq k$ for all $v \in V$. We consider the following four operations on labeled graphs: 
\begin{itemize}
  \item the \emph{introduce}-operation $\intro{\ell}(v)$ which constructs a single-vertex graph whose unique vertex $v$ has label $\ell$,
  \item the \emph{union}-operation $G_1 \union G_2$ which constructs the disjoint union of two labeled graphs $G_1$ and $G_2$,
  \item the \emph{relabel}-operation $\relab{i}{j}(G)$ changes the label of all vertices in $G$ with label $i$ to label $j$, 
  \item the \emph{join}-operation $\join{i}{j}(G)$, $i \neq j$, which adds an edge between every vertex in $G$ with label $i$ and every vertex in $G$ with label $j$. 
\end{itemize}
A valid expression that only consists of introduce-, union-, relabel-, and join-operations is called a \emph{clique-expression}. The graph constructed by a clique-expression $\cwexp$ is denoted $G_\cwexp$ and the constructed label function is denoted $\lfct_\cwexp \colon V(G_\cwexp) \rightarrow \NN$.

We associate to a clique-expression $\cwexp$ the syntax tree $\tree_\cwexp$ in the natural way and to each node $t \in V(\tree_\cwexp)$ the corresponding operation. For any node $t \in V(\tree_\cwexp)$ the subtree rooted at $t$ induces a \emph{subexpression} $\cwexp_t$. When a clique-expression $\cwexp$ is fixed, we define $G_t = G_{\cwexp_t}$, $V_t = V(G_t)$, and $\lfct_t = \lfct_{\cwexp_t}$ for any $v \in V(\tree_\cwexp)$. Furthermore, we write $V_t^\ell = \lfct_t^{-1}(\ell)$ for the set of all vertices with label $\ell$ at node $t$. 

We say that a clique-expression $\cwexp$ is a \emph{$k$-clique-expression} or just \emph{$k$-expression} if $G_t$ is $k$-labeled for all $t \in V(\tree_\cwexp)$. The \emph{clique-width} of a graph $G$, denoted by $\cw(G)$, is the minimum $k$ such that there exists a $k$-expression $\cwexp$ with $G = G_\cwexp$. A clique-expression $\cwexp$ is \emph{linear} if in every union-operation the second graph consists only of a single vertex. Accordingly, we also define the \emph{linear-clique-width} of a graph $G$, denoted $\lcw(G)$, by only considering linear clique-expressions.

\paragraph*{Treedepth and Modulators.}

 We will construct graphs that have small treewidth except for one central part. This structure is captured by the concept of a \emph{(vertex) modulator}. We say that $\modulator \subseteq V(G)$ is a \emph{modulator to treewidth/pathwidth $\ncols$} for $G$ if $\tw(G - \modulator) \leq \ncols$ or $\pw(G - \modulator) \leq \ncols$, respectively. 

\begin{dfn}
 An \emph{elimination forest} of an undirected graph $G = (V, E)$ is a rooted forest $\EF = (V, E_\EF)$ so that for every edge $\{u, v\} \in E$ either $u$ is an ancestor or descendant of $v$ in $\EF$. The \emph{depth} of a rooted forest is the largest number of nodes on a path from a root to a leaf. The \emph{treedepth} of $G$, denoted $\td(G)$, is the minimum depth over all elimination forests of $G$.
\end{dfn}

\paragraph*{Cutwidth.}

A \emph{linear layout} of a graph $G = (V, E)$ with $n$ vertices is a linear ordering of its vertices given by a bijection $\pi \colon V \rightarrow [n]$. The \emph{cutwidth} of $G$ with respect to $\pi$ is 
\begin{equation*}
  \ctw_\pi(G) = \max_{i \in [n]} |\{\{u, v\} \in E \sep \pi(u) \leq i \wedge \pi(v) > i\}|.
\end{equation*}
The \emph{cutwidth} $\ctw(G)$ of $G$ is the minimum cutwidth over all linear layouts of $G$. 

\paragraph*{Lifting to Twinclasses.}

We define the \emph{twinclass-treewidth}, \emph{twinclass-pathwidth}, \emph{twinclass-treedepth}, and \emph{twinclass-cutwidth} of $G$ by $\tctw(G) = \tw(G/\tcpartition(G))$, $\tcpw(G) = \pw(G/\tcpartition(G))$, $\tctd(G) = \td(G/\tcpartition(G))$, and $\tcctw(G) = \ctw(G/\tcpartition(G))$, respectively. The parameters twinclass-treewidth and twinclass-pathwidth have been considered before under the name modular treewidth and modular pathwidth~\cite{Lampis20, Mengel16, PaulusmaSS16}. We prefer to use the prefix twinclass instead of modular to distinguish from the case where one works with the quotient graph arising from the \emph{modular partition} of $G$. 

\begin{dfn}
  Let $G = (V,E)$ be a graph. A \emph{twinclass-modulator (TCM)} $\tcm \subseteq \tcpartition(G)$ of $G$ to treewidth $\ncols$ is a set of twinclasses of $G$ such that $\tw(G - \bigcup(\tcm)) \leq \ncols$. The \emph{size} of a twinclass-modulator $\tcm$ is $|\tcm|$, i.e., the number of twinclasses $\tcm$ contains.
\end{dfn}

\subsection{Strong Exponential-Time Hypothesis}

For our lower bounds, we assume the \emph{Strong Exponential-Time Hypothesis} (\SETH) \cite{ImpagliazzoPZ01} which concerns the complexity of $\clss$-\SAT, i.e., \SAT where all clauses contain at most $\clss$ literals. Let $c_\clss = \inf \{\delta \sep \clss\text{-\SAT can be solved in time } \Oh(2^{\delta \nvars}) \}$ for all $\clss \geq 3$. The weaker \emph{Exponential-Time Hypothesis} (ETH) of Impagliazzo and Paturi~\cite{ImpagliazzoP01} posits that $c_3 > 0$, whereas the Strong Exponential-Time Hypothesis states that $\lim_{q \rightarrow \infty} c_q = 1$, i.e., for large enough clause size there is essentially no better algorithm than brute force.

When proving lower bounds based on \SETH, we make use of the following equivalent formulations.

\begin{mythm}[\cite{CyganDLMNOPSW16}]
 \label{thm:seth_hitting_set}
 The following statements are equivalent to \SETH:
 \begin{enumerate}
  \item For all $\delta < 1$, there is a clause size $\clss$ such that $\clss$-\SAT cannot be solved in time $\Oh(2^{\delta \nvars})$, where $\nvars$ is the number of variables.
  \item For all $\delta < 1$, there is a set size $\clss$ such that $\clss$-\textsc{Hitting Set}, i.e., all sets contain at most $\clss$ elements, cannot be solved in time $\Oh(2^{\delta \nvars})$, where $\nvars$ is the universe size.
 \end{enumerate}
\end{mythm}

\section{Outline of Main Result}
\label{sec:outline}

We outline our two main results, i.e., tight lower bounds for \DTC parameterized by a (twinclass-)modulator to treewidth $\ncols$. Conceptually, the constructions for the sparse setting and for the dense setting are similar. The most significant change is in the \emph{structure gadget}, since we have to enforce a considerably more involved structure in the dense setting. We give an overview of both settings and go into more detail for the dense case.

We fix the number of colors $\ncols \geq 2$. \emph{Solutions} are functions $\coloring \colon V(G) \rightarrow [\ncols] \cup \{\del\}$ so that for every edge $\{u,v\} \in E(G)$ either $\coloring(u) = \coloring(v) = \del$ or $\coloring(u) \neq \coloring(v)$. Hence, $\delset$ is the set of deleted vertices, whereas $\coloring\big|_{V(G) \setminus \delset}$ is an $\ncols$-coloring of the remaining graph.

In both settings we want to simulate a \emph{logical OR} constraint. For \OCT, i.e.\ $\ncols = 2$, we can use \emph{odd cycles}. For $\ncols \geq 3$, \cref{thm:critical_family_informal} provides an analogue, where a graph $H$ is $(\ncols + 1)$-\emph{critical} if $\chi(H) = \ncols + 1$ and $\chi(H - v) = \ncols$ for all $v \in V(H)$.
\begin{mythm}[see \cref{sec:critical_graphs}]
  \label{thm:critical_family_informal}
  There exists a family $\crit^\ncols$ of $(\ncols + 1)$-critical graphs with treewidth $\ncols$ such that for every $s \in \NN$, there exists a graph $H \in \crit^\ncols$ with $s \leq |V(H)| \leq s + \ncols$.
\end{mythm}

\subparagraph*{Setup.} 
Given a $\clss$-\SAT instance $\formula$ with $\nvars$ variables and $\nclss$ clauses, we start with the following standard step~\cite{LokshtanovMS18}: we partition the variables into $t = \lceil \nvars / \vgrpsize \rceil$ groups of size $\vgrpsize$, where $\vgrpsize$ only depends on the running time base that we want to rule out. Furthermore, we pick an integer $\grpsize$ depending on $\vgrpsize$ that represents the size of the groups in the graph.

\subsection{Sparse Setting}

\subparagraph*{Central vertices and solution structure.} We construct a graph $G$ that has a solution $\coloring$ for \DTC with cost $|\delset| \leq \budget$ if and only if $\formula$ is satisfiable. Converting from base $\ncols + 1$ to base $2$ implies that $G$ should admit a modulator $\modulator$ to treewidth $\ncols$ of size roughly $\nvars \log_{\ncols + 1}(2)$. Like Cygan et al.~\cite{CyganDLMNOPSW16}, we make the modulator slightly larger, thus picking a larger $\grpsize$. The modulator $\modulator $ consists of $\ngrps + 1$ vertex groups: the first $\ngrps$ groups $U_i$, $i \in [\ngrps]$, are independent sets of size $\grpsize$ each and correspond to the variable groups; the last group $F$ is a clique of size $\ncols$ which simulates \LC constraints.

On each group $U_i$, we consider the set of partial solutions $\scolorings_i = \{\coloring \colon U_i \rightarrow [\ncols] \cup \{\del\} \sep |\delset| = \grpsize / (\ncols + 1) \}$. By picking $\grpsize$ large enough, $\scolorings_i$ is sufficiently large to encode all assignments of the $i$-th variable group. Defining $\scolorings_i$ in this way achieves two things: first, the solutions in $\scolorings_i$ are pairwise non-dominating; secondly, this fixes the budget used on the modulator. The second point is important, because by also fixing the budget on the remaining graph via a vertex-disjoint packing $\packing$ of $(\ncols + 1)$-critical graphs, no vertex of $F$ can be deleted, which allows us to simulate \LC constraints with the clique $F$.

\subparagraph*{Structure gadgets.} The next step is to enforce that only the solutions in $\scolorings_i$ can be attained on group $U_i$. By choosing the budget $\budget$ appropriately, we obtain an upper bound on the number of deletions in $U_i$. To obtain a lower bound, we construct the \emph{structure gadgets}. These are built by combining $(\ncols + 1)$-critical graphs with the \emph{arrow} gadget of Lokshtanov et al.~\cite{LokshtanovMS18}. A (thin) arrow simply propagates a deletion from a vertex $u$ to another vertex $v$; else if $u$ is not deleted, then $v$ is not deleted and the arrow does not affect the remaining graph.

The structure gadget works as follows: if $\coloring$ deletes less than $\grpsize / (\ncols + 1)$ vertices in group $U_i$, then there is a subset $S \subseteq U_i$ of size $|S| = (|U_i| - \grpsize / (\ncols + 1)) + 1$ that avoids all deletions in $U_i$. For every subset of this size, $G$ contains a $(\ncols + 1)$-critical graph $L_{i,S}$ with an arrow from every $u \in S$ to a private vertex $v$ in $L_{i,S}$, hence simulating an OR on the vertices in $S$. Since $S$ avoids all deletions of $\coloring$, no deletion is propagated to $L_{i,S}$ and $\coloring$ must pay extra to resolve $L_{i,S}$. By copying each $L_{i,S}$ sufficiently often, we can ensure that the existence of a deletion-avoiding $S$ implies that $\coloring$ must exceed our budget constraint.

\subparagraph*{Decode and verify.} The remaining construction decodes the partial solution on the modulator $\modulator$ and verifies if the corresponding truth assignment satisfies all clauses of $\formula$. One could generalize the gadgets of Lokshtanov et al.~\cite{LokshtanovMS18} to higher $\ncols$, but this leads to an involved construction with a worse bound on the treewidth of the remainder: for \OCT the construction of Lokshtanov et al.\ has treewidth 4, whereas the simpler construction we use has only treewidth 2. More details will be presented in the dense case.

\subsection{Dense Setting}

We now have a twinclass-modulator $\tcm$ to treewidth $\ncols$ instead of a basic modulator and this changes the possible states as follows. Whereas $\coloring$ could assume $\ncols + 1$ different states on a single vertex $u$, i.e., one of the $\ncols$ colors or deleting the vertex, there are $2^{\ncols}$ possible states on a true twinclass $U$ of size $\ncols$; each corresponds to a possible value of $\coloring(U) \setminus \{\del\} \subseteq [\ncols]$. Since $U$ is a true twinclass, no color is used multiple times and the exact mapping $\coloring\big|_U$ is irrelevant.

\subparagraph*{Central twinclasses and setup.} The twinclass-modulator $\tcm$ of the constructed graph $G$ consists of $\ngrps + 1$ groups and each group is a family of twinclasses. The first $\ngrps$ groups $\ttcset_i$, $i \in [\ngrps]$, correspond to the variable groups and each consists of $\grpsize$ true twinclasses of size $\ncols$ that are pairwise non-adjacent. The last group contains the clique $F$.

\subparagraph*{Solution structure.} Our family $\scolorings_i$ of considered partial solutions on group $\ttcset_i$ should achieve the same two things as before. First, consider the structure of states of $\coloring$ on a twinclass $U \in \ttcset_i$ precisely: fix a state $C = \coloring(U) \setminus \{\del\}$ and note that all states $C' \subsetneq C$ \emph{dominate} $C$ if we disregard the budget constraint, i.e., $\coloring$ remains a solution if we replace $C$ by $C'$. After arranging the states into \emph{levels} according to the number $\ell$ of deleted vertices, there is no domination between states on the same level. This motivates the following definition.
\begin{dfn}[informal]
 Given rationals $0 < c_\ell < 1$, $\ell \in \{0\} \cup [\ncols]$, with $\sum_{\ell = 0}^\ncols c_\ell = 1$, the set $\scolorings_i$ consists of solutions $\coloring$ on the family of twinclasses $\ttcset_i$ such that for every $\ell \in \{0\} \cup [\ncols]$ there are exactly $c_\ell \cdot |\ttcset_i|$ twinclasses $U \in \ttcset_i$ where $\coloring$ deletes exactly $\ell$ vertices in $U$.
\end{dfn}
Essentially, we are only restricting how the deletions can be distributed inside the modulator; there are no restrictions on the used colors. This again fixes the budget used on the modulator, allowing us to simulate \LC constraints with the clique $F$. By picking $c_\ell = \binom{\ncols}{\ell}2^{-\ncols}$, $\ell \in \{0\} \cup [\ncols]$, we ensure that $\scolorings_i$ contains the solutions on $\ttcset_i$ where all $2^\ncols$ states appear the same number of times. This enables us to choose $\grpsize$ small enough so that the time calculations work out and simultaneously large enough so that an injective mapping $\embedding_i \colon \{0,1\}^\vgrpsize \rightarrow \scolorings_i$, mapping truth assignments of the $i$-th variable group to solutions in $\scolorings_i$, exists.

\subparagraph*{Thick arrows and structure gadgets.} To enforce the structure of $\scolorings_i$, we need a gadget to distinguish different number of deletions inside a twinclass. We can construct such a gadget $\arrow_\ell(U,v)$, $\ell \in [\ncols]$, also called \emph{thick $\ell$-arrow}. See \cref{thm:arrow_behavior_informal} for the gadget's behavior.
\begin{lem}[informal]
	\label{thm:arrow_behavior_informal}
  Let $U$ be a set of $\ncols$ true twins and $v$ be a vertex that is not adjacent to $U$ and $\ell \in [\ncols]$. There is a gadget $\arrow = \arrow_\ell(U,v)$ of treewidth $\ncols$ with the following properties:
  \begin{itemize}
   \item Any solution $\coloring$ must delete at least $\ell$ vertices in $\arrow - U$.
   \item If a solution $\coloring$ deletes exactly $\ell$ vertices in $\arrow - U$, then $\coloring$ can only delete $v$ if $\coloring$ deletes at least $\ell$ vertices in $U$.
  \end{itemize}
\end{lem}
We proceed by constructing the \emph{structure gadgets} which enforce that the partial solution on $\ttcset_i$ belongs to $\scolorings_i$. Let $c_{< \ell} = c_0 + \cdots + c_{\ell - 1}$ for all $\ell \in \{0\} \cup [\ncols]$. For every group $i \in [\ngrps]$, number of deletions $\ell \in [\ncols]$, set of twinclasses $\setfam \subseteq \ttcset_i$ with $|\setfam| = c_{< \ell} \cdot \grpsize + 1$, we add an $(\ncols + 1)$-critical graph $L_{i, \ell, \setfam} \in \crit^\ncols$ consisting of at least $|\setfam|$ vertices. For every $U \in \setfam$, we pick a private vertex $v$ in $L_{i, \ell, \setfam}$ and add the thick $\ell$-arrow $\arrow_\ell(U,v)$. We create a large number of copies of each $L_{i, \ell, \setfam}$ and the incident thick arrows. This concludes the construction of the structure gadget.

The number of deletions in the central vertices is already bounded from above by the budget constraint. If too few deletions occur in the twinclasses of $\ttcset_i$, then we can find an $\ell$ and an $\setfam \subseteq \ttcset_i$ with $|\setfam| = c_{< \ell} \cdot \grpsize + 1$ such that less than $\ell$ vertices are deleted in each $U \in \setfam$. Hence, all thick $\ell$-arrows leading to $L_{i, \ell, \setfam}$ and its copies cannot propagate deletions. To resolve all these $(\ncols + 1)$-critical graphs, one extra vertex per copy must be deleted. Due to the large number of copies, this implies that we must violate our budget constraint. 

Hence, for any $\setfam \subseteq \ttcset_i$ with $|\setfam| = c_{< \ell} \cdot \grpsize + 1$ and any solution $\coloring$ obeying the budget constraint there is at least one twinclass $U \in \setfam$ in which $\coloring$ deletes at least $\ell$ vertices. Therefore, there are at least $(1 - c_{< \ell}) \grpsize$ twinclasses in $\ttcset_i$ where $\coloring$ deletes at least $\ell$ vertices. Since this holds for all $\ell \in \{0\} \cup [\ncols]$ and the budget $\budget$ is chosen appropriately, all inequalities have to be tight and the deletions inside $\ttcset_i$ follow the distribution imposed by $\scolorings_i$. The upcoming decoding gadgets will make use of this structure.

\subparagraph*{Color-set-gadgets and decoding gadgets.} Next, we discuss the decoding part of the construction. Since gadgets cannot read the color of single vertices but only of a whole twinclass, we need \emph{color-set-gadgets} to detect the colors used on a twinclass, cf.\ \cref{thm:decode_behavior_informal}.
\begin{lem}[informal]
  \label{thm:decode_behavior_informal}
  Let $U$ be a set consisting of $\ncols$ true twins and $v$ be a vertex that is not adjacent to $U$ and let $C \subsetneq [\ncols]$. There is a gadget $\decode = \decode_C(U,v)$ of treewidth $\ncols$ such that:
  \begin{itemize}
   \item Any solution $\coloring$ deletes at least $(\ncols - |C|) + 1$ vertices in $\decode - U$.
   \item If $\coloring$ deletes exactly $(\ncols - |C|) + 1$ vertices in $\decode - U$, then $\coloring(v) = \del$ only if $\coloring(U) \setminus \{\del\} \subseteq C$.
  \end{itemize}
\end{lem}
To construct the color-set-gadgets we rely on the \LC constraints that are simulated with the central clique $F$. Note that the color-set-gadgets only check for set inclusion and not set equality. Using the structure of solutions in $\scolorings_i$ however, the color-set-gadgets will still be sufficient to distinguish the solutions in $\scolorings_i$ from each other. 

By using a complete $(\ncols + 1)$-partite graph with all sets of the partition being singletons except for one large independent set, we can simulate a logical AND, see \cref{thm:and_gadget_informal}.
\begin{lem}[informal]
 \label{thm:and_gadget_informal}
 Let $n_Y$ be a positive integer. There is a gadget $Y$ of treewidth $\ncols$ with a set of input vertices $V' \subseteq V(Y)$, $|V'| = n_Y$, and a vertex $\hat{y} \in V(Y) \setminus V'$ such that:
 \begin{itemize}
  \item Any solution $\coloring$ has to delete at least one vertex in $Y - V'$.
  \item If $\coloring$ deletes exactly one vertex in $Y - V'$, then $\coloring(\hat{y}) = \del$ only if $\coloring(V') = \{\del\}$.
 \end{itemize}
\end{lem}
For the $j$-th clause, variable group $i \in [\ngrps]$, solution $\coloring_i \in \scolorings_i$, we invoke \cref{thm:and_gadget_informal} to create a gadget $Y^j_{i, \coloring_i}$ for $n_Y = (1 - c_r) \grpsize = (1 - 2^{-\ncols})\grpsize$ input vertices and with distinguished vertex $\hat{y}^j_{i, \coloring_i}$. For every twinclass $U \in \ttcset_i$ with $\coloring_i(U) \neq [\ncols]$, we pick a private input vertex $v$ of $Y^j_{i, \coloring_i}$ and add the color-set-gadget $B_{\coloring_i(U) \setminus \{\del\}}(U,v)$. By \cref{thm:and_gadget_informal}, the vertex $\hat{y}^j_{i, \coloring_i}$ can only be deleted if all input vertices of $Y^j_{i, \coloring_i}$ are deleted. Due to \cref{thm:decode_behavior_informal} and the structure of $\scolorings_i$, this will only be the case if $\coloring_i$ is the partial solution on $\ttcset_i$. 

\subparagraph*{Clause gadgets.} For the $j$-th clause, we add an $(\ncols + 1)$-critical graph $Z^j \in \crit^\ncols$ consisting of at least $\clss 2^{\vgrpsize}$ vertices. For every group $i \in [\ngrps]$ and solution $\coloring_i \in \scolorings_i$ such that $\embedding_i^{-1}(\coloring_i)$ is a partial truth assignment satisfying the $j$-th clause, we pick a private vertex $v$ in $Z^j$ and add a thin arrow from $\hat{y}^j_{i, \coloring_i}$ to $v$. The budget constraint will ensure that the only way to delete a vertex in $Z^j$ is by propagating a deletion via a thin arrow from some $\hat{y}^j_{i, \coloring_i}$. By construction of the decoding and clause gadgets this is only possible if the partial solution on $\ttcset_i$ corresponds to a satisfying assignment of the $j$-th clause. This concludes the construction, cf.\ \cref{fig:dtc_overview_outline}.

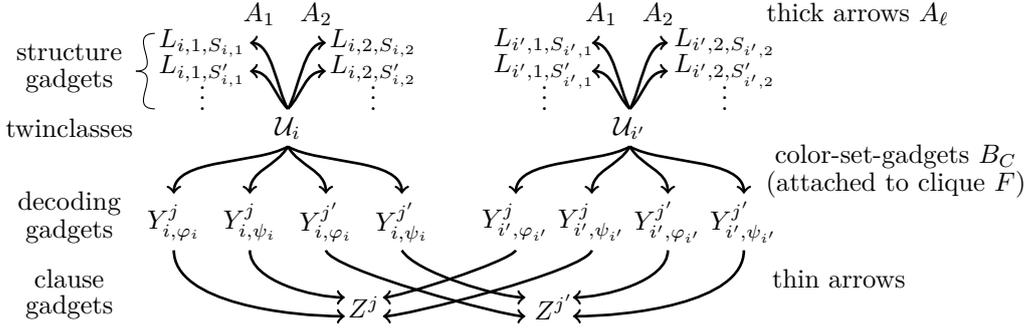
\begin{figure}
  \centering
  \begin{tikzpicture}
	\begin{pgfonlayer}{nodelayer}
		\node [style=none] (0) at (8.5, 5) {$\mathcal{U}_i$};
		\node [style=none] (1) at (17.5, 5) {$\mathcal{U}_{i'}$};
		\node [style=none] (2) at (6.25, 7.25) {$L_{i,1,S_{i,1}}$};
		\node [style=none] (3) at (6.25, 6.5) {$L_{i,1,S'_{i,1}}$};
		\node [style=none] (4) at (10.75, 7.25) {$L_{i,2,S_{i,2}}$};
		\node [style=none] (5) at (10.75, 6.5) {$L_{i,2,S'_{i,2}}$};
		\node [style=none] (6) at (6.25, 6) {$\vdots$};
		\node [style=none] (7) at (10.75, 6) {$\vdots$};
		\node [style=none] (8) at (15.25, 7.25) {$L_{i',1,S_{i',1}}$};
		\node [style=none] (9) at (15.25, 6.5) {$L_{i',1,S'_{i',1}}$};
		\node [style=none] (10) at (15.25, 6) {$\vdots$};
		\node [style=none] (11) at (20, 7.25) {$L_{i',2,S_{i',2}}$};
		\node [style=none] (12) at (20, 6.5) {$L_{i',2,S'_{i',2}}$};
		\node [style=none] (13) at (20, 6) {$\vdots$};
		\node [style=none] (14) at (7.5, 7.25) {};
		\node [style=none] (15) at (7.5, 6.5) {};
		\node [style=none] (16) at (8.5, 5.5) {};
		\node [style=none] (17) at (9.5, 7.25) {};
		\node [style=none] (18) at (9.5, 6.5) {};
		\node [style=none] (19) at (16.5, 7.25) {};
		\node [style=none] (20) at (16.5, 6.5) {};
		\node [style=none] (21) at (17.5, 5.5) {};
		\node [style=none] (22) at (18.5, 7.25) {};
		\node [style=none] (23) at (18.5, 6.5) {};
		\node [style=none] (24) at (7.75, 8) {$A_1$};
		\node [style=none] (25) at (9.25, 8) {$A_2$};
		\node [style=none] (26) at (5.5, 2.5) {$Y_{i,\varphi_i}^j$};
		\node [style=none] (27) at (7.5, 2.5) {$Y_{i,\psi_i}^j$};
		\node [style=none] (28) at (9.5, 2.5) {$Y_{i,\varphi_i}^{j'}$};
		\node [style=none] (29) at (11.5, 2.5) {$Y_{i,\psi_i}^{j'}$};
		\node [style=none] (30) at (5.5, 3.25) {};
		\node [style=none] (31) at (7.5, 3.25) {};
		\node [style=none] (32) at (9.5, 3.25) {};
		\node [style=none] (33) at (11.5, 3.25) {};
		\node [style=none] (34) at (8.5, 4.5) {};
		\node [style=none] (35) at (14.5, 2.5) {$Y_{i',\varphi_{i'}}^j$};
		\node [style=none] (36) at (16.5, 2.5) {$Y_{i',\psi_{i'}}^j$};
		\node [style=none] (37) at (18.5, 2.5) {$Y_{i',\varphi_{i'}}^{j'}$};
		\node [style=none] (38) at (20.5, 2.5) {$Y_{i',\psi_{i'}}^{j'}$};
		\node [style=none] (39) at (14.5, 3.25) {};
		\node [style=none] (40) at (16.5, 3.25) {};
		\node [style=none] (41) at (18.5, 3.25) {};
		\node [style=none] (42) at (20.5, 3.25) {};
		\node [style=none] (43) at (17.5, 4.5) {};
		\node [style=none] (44) at (24.5, 4.25) {color-set-gadgets $B_C$};
		\node [style=none] (45) at (23.5, 8) {thick arrows $A_\ell$};
		\node [style=none] (46) at (16.75, 8) {$A_1$};
		\node [style=none] (47) at (18.25, 8) {$A_2$};
		\node [style=none] (48) at (10.5, 0.25) {$Z^j$};
		\node [style=none] (49) at (15.5, 0.25) {$Z^{j'}$};
		\node [style=none] (50) at (5.5, 1.75) {};
		\node [style=none] (51) at (7.5, 1.75) {};
		\node [style=none] (52) at (9.5, 1.75) {};
		\node [style=none] (53) at (11.5, 1.75) {};
		\node [style=none] (54) at (14.5, 1.75) {};
		\node [style=none] (55) at (16.5, 1.75) {};
		\node [style=none] (56) at (18.5, 1.75) {};
		\node [style=none] (57) at (20.5, 1.75) {};
		\node [style=none] (58) at (10, 0.5) {};
		\node [style=none] (59) at (10, 0) {};
		\node [style=none] (60) at (11, 0.5) {};
		\node [style=none] (61) at (11, 0) {};
		\node [style=none] (62) at (14.75, 0.5) {};
		\node [style=none] (63) at (14.75, 0) {};
		\node [style=none] (64) at (16, 0) {};
		\node [style=none] (65) at (16, 0.5) {};
		\node [style=none] (66) at (23, 1) {thin arrows};
		\node [style=none] (67) at (2.75, 6.25) {gadgets};
		\node [style=none] (68) at (2.75, 0.25) {gadgets};
		\node [style=none] (69) at (24.5, 3.5) {(attached to clique $F$)};
		\node [style=none] (70) at (2.75, 5) {twinclasses};
		\node [style=none] (71) at (2.75, 2.25) {gadgets};
		\node [style=none] (72) at (4.5, 6.5) {};
		\node [style=none] (73) at (5, 7.5) {};
		\node [style=none] (74) at (5, 5.5) {};
		\node [style=none] (75) at (2.75, 7) {structure};
		\node [style=none] (77) at (2.75, 3) {decoding};
		\node [style=none] (78) at (2.75, 1) {clause};
	\end{pgfonlayer}
	\begin{pgfonlayer}{edgelayer}
		\draw [style=thick arrow, in=180, out=60, looseness=0.75] (16.center) to (17.center);
		\draw [style=thick arrow, in=180, out=45] (16.center) to (18.center);
		\draw [style=thick arrow, in=0, out=120, looseness=0.75] (16.center) to (14.center);
		\draw [style=thick arrow, in=0, out=135] (16.center) to (15.center);
		\draw [style=thick arrow, in=180, out=60, looseness=0.75] (21.center) to (22.center);
		\draw [style=thick arrow, in=-180, out=45] (21.center) to (23.center);
		\draw [style=thick arrow, in=0, out=120, looseness=0.75] (21.center) to (19.center);
		\draw [style=thick arrow, in=0, out=135] (21.center) to (20.center);
		\draw [style=thick arrow, in=90, out=-150, looseness=0.75] (34.center) to (30.center);
		\draw [style=thick arrow, in=90, out=-120] (34.center) to (31.center);
		\draw [style=thick arrow, in=90, out=-60] (34.center) to (32.center);
		\draw [style=thick arrow, in=90, out=-30, looseness=0.75] (34.center) to (33.center);
		\draw [style=thick arrow, in=90, out=-150, looseness=0.75] (43.center) to (39.center);
		\draw [style=thick arrow, in=90, out=-120] (43.center) to (40.center);
		\draw [style=thick arrow, in=90, out=-60] (43.center) to (41.center);
		\draw [style=thick arrow, in=90, out=-30, looseness=0.75] (43.center) to (42.center);
		\draw [style=thick arrow, in=180, out=-90, looseness=0.75] (50.center) to (59.center);
		\draw [style=thick arrow, in=180, out=-90, looseness=0.75] (51.center) to (58.center);
		\draw [style=thick arrow, in=0, out=-120, looseness=0.25] (54.center) to (60.center);
		\draw [style=thick arrow, in=0, out=-105, looseness=0.25] (55.center) to (61.center);
		\draw [style=thick arrow, in=180, out=-75, looseness=0.25] (52.center) to (63.center);
		\draw [style=thick arrow, in=180, out=-60, looseness=0.50] (53.center) to (62.center);
		\draw [style=thick arrow, in=0, out=-90, looseness=0.75] (56.center) to (65.center);
		\draw [style=thick arrow, in=0, out=-90, looseness=0.75] (57.center) to (64.center);
		\draw [in=180, out=0, looseness=1.75] (72.center) to (73.center);
		\draw [in=180, out=0, looseness=1.75] (72.center) to (74.center);
	\end{pgfonlayer}
\end{tikzpicture}
  \caption{An overview of the construction for the dense setting in case of $\ncols = 2$. The arrows point in the direction that deletions are propagated by the corresponding gadget.}
  \label{fig:dtc_overview_outline}
\end{figure}

\subparagraph*{Budget and packing.} The budget $\budget = \budget_0 + \cost_\packing$ of the constructed instance $(G, \budget)$ consists of two parts; $\budget_0 = \ngrps \ncols \grpsize / 2$ is allocated to the central twinclasses and matches the number of deletions incurred by picking a partial solution $\coloring_i \in \scolorings_i$ on $\ttcset_i$ for each group $i \in [\ngrps]$; the second part $\cost_\packing$ is due to a vertex-disjoint packing $\packing$ which we describe next. A part of each thin arrow in $G$ is added to $\packing$ and for every thick arrow, color-set-gadget, or decoding gadget, we add the appropriate parts to $\packing$ given by Lemmas 3.3, 3.4, 3.5, respectively. Summing up the implied costs yields $\cost_\packing$. Hence, we know how the deletions are distributed throughout the various gadgets. In particular, this ensures that no vertex of the central clique $F$ is deleted.

\cref{thm:intro_lb_dense} follows by using these ideas and working out the remaining technical details.

\section{Relations between Parameters}
\label{sec:parameters}

We discuss the relationships between the considered parameters and provide an equivalent definition for treewidth via a \emph{search game} in this section.

\begin{lem}[\cite{Bodlaender98}, Chapter 6 of \cite{NesetrilO12}]
  \label{thm:td_tw}
  For any graph $G$, we have that $\tw(G) \leq \pw(G) \leq \td(G) - 1$, $\td(G) \leq (\tw(G) + 1) \log_2 |V(G)|$, $\tw(G) \leq \pw(G) \leq \ctw(G)$, and $\td(G) \leq \td(G - v) + 1$ for any vertex $v \in V(G)$. These inequalities come with algorithms that can transform the appropriate decomposition in polynomial time.
\end{lem}

\begin{cor}
  \label{thm:mod_to_tw_implies_td}
  For any graph $G = (V, E)$ and $c, \ncols \in \NN$, if there is a modulator $\modulator \subseteq V$ to treewidth $\ncols$, i.e., $\tw(G - \modulator) \leq \ncols$, then we have that $\td(G) \leq |\modulator| + (\ncols + 1) \log_2 |V|$. In particular, we have that $\Oh^*(c^{\td(G)}) \leq \Oh^*(c^{|\modulator|})$ for all $c \geq 1$. The decompositions can be transformed in polynomial time.
\end{cor}

\begin{proof}
  Let $\modulator$ be a modulator to treewidth $\ncols$ for $G$. By \cref{thm:td_tw}, we see that $\td(G - \modulator) \leq (\ncols + 1) \log_2 |V|$ for $G - \modulator$. By repeatedly invoking the inequality $\td(G) \leq \td(G - v) + 1$ for $v \in X$, we obtain $\td(G) \leq |\modulator| + (\ncols + 1) \log_2 |V|$.
  To see the claim regarding the $\Oh^*$-notation, we compute $\Oh^*(c^{\td(G)}) = \Oh^*(c^{|\modulator|}|V|^{(\ncols+1)\log_2 c}) = \Oh^*(c^{|\modulator|})$.
\end{proof}

\begin{mythm}
  \label{thm:hierarchy_cliquewidth}
  Let $G = (V, E)$ be a graph. We have the following two chains of inequalities:
  $$\arraycolsep=1.4pt
  \begin{array}{lllll}
    \cw(G) & \leq \lcw(G) & \leq \tcpw(G) + 3 & \leq \tctd(G)  + 2 & \leq \td(G)  + 2, \\
    \cw(G) & \leq \lcw(G) & \leq \tcpw(G) + 3 & \leq \tcctw(G) + 3 & \leq \ctw(G) + 3.
  \end{array}
  $$
\end{mythm}

\begin{proof} 
  Follows from \cite[Lemma 2.1]{Lampis20}, \cref{thm:td_tw} and the last inequalities in both rows follow from the fact that $G / \tcpartition(G)$ is a subgraph of $G$ and that treedepth and cutwidth are subgraph-monotone.
\end{proof}

\begin{lem}
  \label{thm:tcm_tctd}
  Suppose that $G$ admits a TCM $\tcm$ to treewidth $\ncols$, then $\tctd(G) \leq |\tcm| + (\ncols + 1) \log_2 |V(G/\tcpartition(G))|$. In particular, we have for any $c \geq 1$ that $\Oh^*(c^{\tctd(G)}) \leq \Oh^*(c^{|\tcm|})$. The decompositions can be transformed in polynomial time.
\end{lem}

\begin{proof}
  Since $G/\tcpartition(G) - \tcm$ is an induced subgraph of $G - \bigcup(\tcm)$, we see that $\tw(G/\tcpartition(G) - \tcm) \leq \tw(G - \bigcup(\tcm)) \leq \ncols$. The remainder of the proof is analogous to \cref{thm:mod_to_tw_implies_td} by working on the quotient graph $G/\tcpartition(G)$.
\end{proof}

\subsection*{Search Games}

To prove that the graphs in our lower bound constructions have small treewidth, it is easier to use a \emph{search game} characterization instead of directly constructing a tree decomposition.

\subparagraph*{Cops and Robbers.} The search game corresponding to treewidth is the \emph{omniscient-cops-and-robber-game}. In this game $k + 1$ cops try to capture a robber on a graph $G$. The cops can occupy vertices of $G$ and the robber can run at infinite speed along edges through vertices that are not occupied by cops. The cops and the robber are omniscient in that they know each others position at all times. At the start, the cops are placed on some vertices and the robber chooses a starting vertex for the escape. In every round some cops get into their helicopters and declare where they will land next. Before these cops land again, the robber may move to any other vertex that is reachable by a path that does not contain any non-airborne cop. Afterwards the cops land and the next round begins. The cops \emph{win} if they have a \emph{strategy} to capture the robber by landing on the robber's location after a finite number of rounds; otherwise the robber wins. 

\begin{mythm}[\cite{SeymourT93}]
  \label{thm:cops_robber}
  A graph $G$ has treewidth at most $k$ if and only if $k + 1$ cops can win the omniscient-cops-and-robber-game on $G$.
\end{mythm}

\section{Algorithm for Deletion to $\ncols$-Colorable}
\label{sec:cw_algo}

In this section we describe how to solve \DTC in time $\Oh^*((2^\ncols)^k)$ given a $k$-expression $\cwexp$ for $G$. We perform bottom-up dynamic programming along the syntax tree $\tree_\cwexp$. We again view \emph{solutions} to \DTC as functions $\coloring \colon V(G) \rightarrow [\ncols] \cup \{\del\}$ with the property discussed in the outline, cf.~\cref{sec:outline}.

\begin{mythm}
 \label{thm:dtc_algo}
 Given a $k$-expression $\cwexp$ for $G$, \DTC on $G$ can be solved in time $\Oh^*((2^\ncols)^k)$.
\end{mythm}

\begin{proof}
 Let $(G,\budget)$ be a \DTC instance and $\cwexp$ a $k$-expression for $G$. We can without of loss of generality assume that $\cwexp$ consists of $\Oh(|V(G)|)$ union-operations and $\Oh(|V(G)|k^2)$ unary operations~\cite{CourcelleO00}. For every node $t \in V(\tree_\cwexp)$ and label $\ell$, we store the set of colors used on $V^\ell_t$. After deleting the appropriate vertices, the remaining graph should be $\ncols$-colorable, hence the possible color sets are precisely the subsets of $[\ncols]$, where $\emptyset$ indicates that all vertices are deleted. Since we use at most $k$ labels at every node, this results in $(2^\ncols)^k$ possible types of partial solutions at each node. If the work for each type is only polynomial, then the claimed running time immediately follows, since there are only a polynomial number of nodes in $V(\tree_\cwexp)$. 

 For every $t \in V(\tree_\cwexp)$, a function $f \colon [k] \rightarrow \powerset([\ncols])$ is called a \emph{signature} and the set of \emph{partial solutions at $t$ respecting $f$} is defined by
\begin{align*}
	\psols_t(f) = \{ \coloring \colon V_t \rightarrow [\ncols] \cup \{\del\} \sep \,\,& \text{$\coloring$ induces an $\ncols$-coloring of } G_t - \delset \text{ and} \\ & \coloring(V_t^\ell) \setminus \{\del\} = f(\ell) \tfa \ell \in [k] \}.
\end{align*}
We want to compute the quantity $\DP_t(f) = \min \{|\delset| \sep \coloring \in \psols_t(f)\}$, where $\DP_t(f) = \infty$ if $\psols_t(f) = \emptyset$. Let $t_0$ be the root node of the $k$-expression $\cwexp$. The algorithm returns true if there is a $t_0$-signature $f$ such that $\DP_{t_0}(f) \leq \budget$ and otherwise the algorithm returns false.

Note that $f(\ell) = \emptyset$ implies $\coloring(V_t^\ell) = \{\bot\}$ for all $\coloring \in \psols_t(f)$, i.e., all vertices with label $\ell$ are deleted. Furthermore, the definition of $\psols_t(f)$ implies that $\psols_t(f) = \emptyset$ and $\DP_t(f) = \infty$ whenever $|f(\ell)| > |V^\ell_t|$ for some $\ell \in [k]$; we will not explicitly mention this edge case again in what follows and assume that the considered $f$ satisfy $|f(\ell)| \leq |V^\ell_t|$ for all $\ell \in [k]$. We proceed by presenting the recurrences to compute $\DP_t(f)$ for all nodes $t$ and signatures $f$ and afterwards show the correctness of these recurrences.

\paragraph*{Base case.} 
If $t = \intro{\ell}(v)$ for some $\ell \in [k]$, then $\DP_t(f) = [f(\ell) = \emptyset]$, because the solution cost is 1 if $v$ is deleted and $0$ otherwise.

\paragraph*{Relabel case.} 
If $t = \relab{i}{j}(G_{t'})$ for some $i \neq j \in [k]$ and where $t'$ is the child of $t$, then
\begin{equation*} 
	\DP_t(f) = \min \{ \DP_{t'}(f') \sep f'(\ell) = f(\ell) \tfa \ell \in [k] \setminus \{i,j\} \text{ and } f'(i) \cup f'(j) = f(j) \}. 
\end{equation*}
By assumption, $f$ will always satisfy $f(i) = \emptyset$ here, since there are no vertices with label $i$ in $G_{t}$. This recurrence goes over all ways how the colors $f(j)$ used for vertices with label $j$ in $G_{t}$ can be split among the vertices with label $i$ and $j$ in the previous graph $G_{t'}$.
Observe that we are taking the minimum over at most $(2^\ncols)^2 = \Oh(1)$ numbers on the right-hand side, hence this recurrence can be computed in polynomial time.

\paragraph*{Join case.} 
If $t = \join{i}{j}(G_{t'})$ for some $i \neq j \in [k]$, where $t'$ is the child of $t$, and assuming without loss of generality that $V^i_{t'} \neq \emptyset$ and $V^j_{t'} \neq \emptyset$, then 
\begin{equation*}
	\DP_t(f) = 
	\begin{cases} \DP_{t'}(f) & \text{if } f(i) \cap f(j) = \emptyset, \\
									  \infty & \text{else.} 
  \end{cases}
\end{equation*}
This recurrence filters out all partial solutions where the coloring properties are not satisfied at some newly added edge. This happens precisely when $f(i) \cap f(j) \neq \emptyset$, because then there exists an edge in the join between label $i$ and $j$ whose endpoints get the same color.

\paragraph*{Union case.} 
If $t = G_{t_1} \union G_{t_2}$ where $t_1$ and $t_2$ are the children of $t$, then
\begin{equation*}
	\DP_t(f) = \min \{\DP_{t_1}(f_1) + \DP_{t_2}(f_2) \sep f_1(\ell) \cup f_2(\ell) = f(\ell) \tfa \ell \in [k]\}.
\end{equation*}
Here, we assume that $\infty + x = x + \infty = \infty + \infty = \infty$ for all $x \in \NN$. 
This recurrence goes for each label $\ell \in [k]$ over all ways how the color set $f(\ell)$ can be split among the vertices with label $\ell$ in the first graph $G_{t_1}$ and in the second graph $G_{t_2}$. 

To compute this recurrence, we make use of the following fast cover product algorithm.
\begin{lem}[\cite{CyganFKLMPPS15}]\label{thm:fast_cover_product_small}
	For two functions $f,g \colon \powerset(U) \rightarrow \{-M, \ldots, M\}$, given all $2^{|U|}$ values of $f$ and $g$ in the input, all the $2^{|U|}$ values of the \emph{cover product} $f *_c g$ defined by
	\begin{equation*}
		(f *_c g)(X) = \min_{Y \cup Z = X} f(Y) + g(Z),
	\end{equation*}
	can be computed in time $2^{|U|}|U|^{\Oh(1)} \cdot \Oh(|M| \log |M| \log \log |M|)$.
\end{lem}
\cref{thm:fast_cover_product_small} allows us to compute the recurrence for all $f$ simultaneously in time $\Oh^*((2^\ncols)^k)$: We interpret the signatures $f \colon [k] \rightarrow \powerset([\ncols])$ as subsets of $[k] \times [\ncols]$ in the following way: $S(f) = \{(i,c) \sep i \in [k], c \in f(i)\}$. Observe that $f_1(\ell) \cup f_2(\ell) = f(\ell) \tfa \ell \in [k]$ is equivalent to $S(f_1) \cup S(f_2) = S(f)$. Now, $\DP_t$ can be considered as a function $\powerset([k] \times [\ncols]) \rightarrow [n+1]$ by replacing $\infty$ with $n + 1$, since we have $\DP_t(f) \leq n$ whenever $\psols_t(f) \neq \emptyset$. The recurrence of the union case can then be considered as the $(\min, +)$-cover product of $\DP_{t_1}$ and $\DP_{t_2}$. By \cref{thm:fast_cover_product_small} with $U = [k] \times [\ncols]$ and $M = n+1$, we can compute all values of $\DP_t$ in time $2^{k\ncols} (k\ncols)^{\Oh(1)} \cdot \Oh(n \log n \log \log n) = \Oh^*((2^\ncols)^k)$. 

\paragraph*{Correctness.} 
We prove the correctness by bottom-up induction along the syntax tree $\tree_\cwexp$. In the base case $G_t$ only consists of the single vertex $v$ and we can either delete $v$ or assign some color to $v$. Together with the edge case handling, this is implemented by the formula for the base case.

For the relabel case, notice that $G_t = G_{t'}$, $V^i_t = \emptyset$, $V^j_t = V^i_{t'} \cup V^j_{t'}$, and $V^\ell_t = V^\ell_{t'}$ for all $\ell \in [k] \setminus \{i,j\}$. Let $f'$ be a candidate in the recurrence of $\DP_t(f)$ and $\coloring' \in \psols_{t'}(f')$ be a minimizer in the definition of $\DP_{t'}(f')$, then we also have that $\coloring' \in \psols_t(f)$ since $\coloring'(V^j_t) \setminus \{\del\} = (\coloring'(V^i_{t'}) \setminus \{\del\}) \cup (\coloring'(V^j_{t'}) \setminus \{\del\}) = f'(i) \cup f'(j) = f(j)$. Hence, the recurrence is an upper bound on $\DP_t(f)$. 

In the other direction, let $\coloring$ be a minimizer in the definition of $\DP_t(f)$ and consider $f'$ with $f'(\ell) = \coloring(V^\ell_{t'}) \setminus \{\del\}$ for all $\ell \in [k]$. Then $f'$ satisfies $f'(i) \cup f'(j) = f(j)$ and $\coloring \in Q_{t'}(f')$, so $f'$ is also considered in the recurrence and the recurrence is a lower bound on $\DP_t(f)$.

For the join case, notice that for $\coloring \in \psols_{t'}(f) \supseteq \psols_t(f)$ it holds that $\coloring \in \psols_t(f)$ if and only if $\coloring(V^i_{t'}) \cap \coloring(V^j_{t'}) \subseteq \{\del\}$ as otherwise $\coloring$ cannot induce a coloring of $G_t - \delset$.

For the union case, a feasible solution $\coloring$ of $G_t$ induces feasible solutions $\coloring_1 \equiv \coloring\big|_{V_{t_1}}$ of $G_{t_1}$ and $\coloring_2 \equiv \coloring\big|_{V_{t_2}}$ of $G_{t_2}$ such that $\coloring_1(V^\ell_t) \cup \coloring_2(V^\ell_t) = \coloring(V^\ell_t)$ for all $\ell \in [k]$ and vice versa.
\end{proof}

This algorithm has a straightforward extension that can also handle polynomially large vertex costs in running time $\Oh^*((2^\ncols)^k)$. For even larger costs it is not clear how to compute the table entries for the union nodes quickly enough.

\section{Lower Bound for Deletion to $\ncols$-colorable}
\label{sec:dtc_lb}

In this section, we prove two tight parameterized lower bounds for \DTC. Compared to the outline, \cref{sec:outline}, we present the lower bounds in reverse order, i.e., we first present the dense setting in \cref{sec:dtc_lb_cw} and then the sparse setting in \cref{sec:dtc_lb_td}. Since the constructions are so similar, this allows us to give the more complicated proof for the dense setting in full detail, while only discussing the changes needed for the sparse setting. Before presenting the lower bounds, we construct a family of $(\ncols + 1)$-critical graphs in \cref{sec:critical_graphs} that will serve as important gadgets in both lower bound constructions.

\subsection{Construction of Critical Graphs}
\label{sec:critical_graphs}

We say that a graph $G$ is \emph{$t$-critical} if $\chi(G) = t$ and $\chi(G - v) = t - 1$ for all $v \in V(G)$.  The odd cycles form a family of $3$-critical graphs of treewidth $2$ that contains graphs of arbitrarily large size. To obtain lower bounds for \DTC with $\ncols > 2$ instead of \OCT, we instead construct a family of $(\ncols + 1)$-critical graphs of treewidth $\ncols$ that contains graphs of arbitrarily large size.

\begin{mythm}
  \label{thm:critical_family}
  Let $\ncols \geq 2$. There exists a family $\crit^\ncols$ of graphs such that 
  \begin{itemize}
    \item $H$ is $(\ncols + 1)$-critical for all $H \in \crit^\ncols$,
    \item $\pw(H) = \tw(H) = \ncols$ for all $H \in \crit^\ncols$,
    \item for every $s \in \NN$, there exists a graph $H \in \crit^\ncols$ with $s \leq |V(H)| \leq s + \ncols$.
  \end{itemize}
\end{mythm}

\begin{dfn}[\cite{Diestel12, Hajos61}]
  Given two graphs $G$ and $H$ and edges $\{v,w\} \in E(G)$ and $\{x, y\} \in E(H)$. We obtain a graph $G' = (V', E')$ by applying \emph{Hajós' construction}, where $V' = (V(G) \cup V(H) \cup \{s\}) \setminus \{v, x\}$ and $E' = E(G - v) \cup E(H - x) \cup \{\{s,u\} \sep \{v,w\} \neq \{v,u\} \in E(G)\} \cup \{\{s,u\} \sep \{x,y\} \neq \{x,u\} \in E(H)\} \cup \{w, y\}$. That is, we remove the edges $\{v,w\}$ and $\{x,y\}$, identify $v$ and $x$ into a single vertex called $s$, and add the edge $\{w,y\}$.
\end{dfn}

\begin{dfn}[\cite{Diestel12, Hajos61}]
  Let $t \geq 3$. A graph $G = (V, E)$ is \emph{$t$-constructible} if it can be constructed from the following operations, starting with the complete graph $K_t$:
  \begin{itemize}
    \item Let $G$ and $H$ be $t$-constructible graphs, then the graph obtained by applying Hajós' construction to $G$ and $H$ (wrt.\ edges $\{v,w\} \in E(G)$ and $\{x,y\} \in E(H)$ and identifying $v$ and $x$) is $t$-constructible.
    \item Let $G$ be a $t$-constructible graph and $u$ and $v$ two non-adjacent vertices in $G$. The graph formed by adding the edge $\{u,v\}$ to $G$ and then contracting $\{u, v\}$ is $t$-constructible. 
  \end{itemize}
\end{dfn}

\begin{lem}[\cite{Diestel12, Hajos61}]
  \label{thm:constructible_chromatic}
  Every $t$-constructible graph $G$ requires at least $t$ colors, i.e. $\chi(G) \geq t$.
\end{lem}

\begin{figure}[h]
  \centering
  \begin{tikzpicture}
	\begin{pgfonlayer}{nodelayer}
		\node [style=filled] (0) at (0, 3) {};
		\node [style=filled] (1) at (2, 3) {};
		\node [style=filled] (2) at (4, 3) {};
		\node [style=filled] (3) at (1, 5) {};
		\node [style=filled] (4) at (3, 5) {};
		\node [style=filled] (5) at (1, 3.85) {};
		\node [style=filled] (6) at (3, 3.85) {};
		\node [style=filled] (7) at (5, 3) {};
		\node [style=filled] (8) at (7, 3) {};
		\node [style=filled] (9) at (6, 5) {};
		\node [style=filled] (10) at (6, 3.85) {};
		\node [style=filled] (14) at (8, 3.85) {};
		\node [style=filled] (15) at (10, 3.85) {};
		\node [style=filled] (16) at (9, 3) {};
		\node [style=filled] (17) at (11, 3) {};
		\node [style=filled] (18) at (8, 5) {};
		\node [style=filled] (19) at (10, 5) {};
		\node [style=filled] (20) at (1, 1.75) {};
		\node [style=filled] (21) at (0, 1.15) {};
		\node [style=filled] (22) at (2, 1.15) {};
		\node [style=filled] (23) at (0.425, 0) {};
		\node [style=filled] (24) at (1.575, 0) {};
		\node [style=filled] (25) at (3, 1.75) {};
		\node [style=filled] (27) at (4, 1.15) {};
		\node [style=filled] (28) at (2.425, 0) {};
		\node [style=filled] (29) at (3.575, 0) {};
		\node [style=filled] (30) at (6, 1.75) {};
		\node [style=filled] (31) at (5, 1.15) {};
		\node [style=filled] (32) at (7, 1.15) {};
		\node [style=filled] (33) at (5.425, 0) {};
		\node [style=filled] (34) at (6.575, 0) {};
		\node [style=filled] (35) at (8, 1.75) {};
		\node [style=filled] (36) at (9, 1.15) {};
		\node [style=filled] (37) at (7.425, 0) {};
		\node [style=filled] (38) at (8.575, 0) {};
		\node [style=filled] (39) at (10.075, 1.75) {};
		\node [style=filled] (40) at (11.075, 1.15) {};
		\node [style=filled] (41) at (9.5, 0) {};
		\node [style=filled] (42) at (10.65, 0) {};
		\node [style=filled, label={above:$b_1$}] (43) at (14.825, 3.6) {};
		\node [style=filled, label={left:$a_1$}] (44) at (13.575, 2.575) {};
		\node [style=filled, label={below:$c_{1,1}$}] (45) at (14.15, 1.25) {};
		\node [style=filled, label={below:$c_{1,2}$}] (46) at (15.5, 1.25) {};
		\node [style=filled, label={above:$b_2$}] (47) at (17.325, 3.6) {};
		\node [style=filled, label={above:$a_2$}] (48) at (16.075, 2.575) {};
		\node [style=filled, label={below:$c_{2,1}$}] (49) at (16.65, 1.25) {};
		\node [style=filled, label={below:$c_{2,2}$}] (50) at (18, 1.25) {};
		\node [style=filled, label={above:$b_3$}] (51) at (19.825, 3.6) {};
		\node [style=filled, label={above:$a_3$}] (52) at (18.575, 2.575) {};
		\node [style=filled, label={right:$a_4$}] (53) at (21.075, 2.575) {};
		\node [style=filled, label={below:$c_{3,1}$}] (54) at (19.15, 1.25) {};
		\node [style=filled, label={below:$c_{3,2}$}] (55) at (20.5, 1.25) {};
	\end{pgfonlayer}
	\begin{pgfonlayer}{edgelayer}
		\draw [style=thick] (3) to (0);
		\draw [style=thick] (0) to (1);
		\draw [style=thick] (1) to (2);
		\draw [style=thick] (2) to (4);
		\draw [style=thick] (4) to (6);
		\draw [style=thick] (6) to (1);
		\draw [style=thick] (1) to (5);
		\draw [style=thick] (5) to (3);
		\draw [style=thick] (5) to (0);
		\draw [style=thick] (6) to (2);
		\draw [style=thick] (3) to (4);
		\draw [style=thick] (9) to (7);
		\draw [style=thick] (7) to (8);
		\draw [style=thick] (8) to (10);
		\draw [style=thick] (10) to (9);
		\draw [style=thick] (10) to (7);
		\draw [style=thick] (8) to (16);
		\draw [style=thick] (16) to (17);
		\draw [style=thick] (17) to (15);
		\draw [style=thick] (15) to (19);
		\draw [style=thick] (19) to (18);
		\draw [style=thick] (18) to (9);
		\draw [style=thick] (18) to (14);
		\draw [style=thick] (14) to (8);
		\draw [style=thick] (14) to (16);
		\draw [style=thick] (16) to (15);
		\draw [style=thick] (19) to (17);
		\draw [style=thick] (20) to (21);
		\draw [style=thick] (21) to (23);
		\draw [style=thick] (23) to (24);
		\draw [style=thick] (24) to (22);
		\draw [style=thick] (20) to (23);
		\draw [style=thick] (20) to (24);
		\draw [style=thick] (22) to (21);
		\draw [style=thick] (22) to (23);
		\draw [style=thick] (21) to (24);
		\draw [style=thick] (28) to (29);
		\draw [style=thick] (29) to (27);
		\draw [style=thick] (25) to (28);
		\draw [style=thick] (25) to (29);
		\draw [style=thick] (27) to (28);
		\draw [style=thick] (25) to (27);
		\draw [style=thick] (27) to (22);
		\draw [style=thick] (22) to (28);
		\draw [style=thick] (22) to (29);
		\draw [style=thick] (20) to (25);
		\draw [style=thick] (30) to (31);
		\draw [style=thick] (31) to (33);
		\draw [style=thick] (33) to (34);
		\draw [style=thick] (34) to (32);
		\draw [style=thick] (30) to (33);
		\draw [style=thick] (30) to (34);
		\draw [style=thick] (32) to (31);
		\draw [style=thick] (32) to (33);
		\draw [style=thick] (31) to (34);
		\draw [style=thick] (37) to (38);
		\draw [style=thick] (38) to (36);
		\draw [style=thick] (35) to (37);
		\draw [style=thick] (35) to (38);
		\draw [style=thick] (36) to (37);
		\draw [style=thick] (36) to (32);
		\draw [style=thick] (32) to (37);
		\draw [style=thick] (32) to (38);
		\draw [style=thick] (30) to (35);
		\draw [style=thick] (41) to (42);
		\draw [style=thick] (42) to (40);
		\draw [style=thick] (39) to (41);
		\draw [style=thick] (39) to (42);
		\draw [style=thick] (40) to (41);
		\draw [style=thick] (39) to (40);
		\draw [style=thick] (36) to (41);
		\draw [style=thick] (36) to (40);
		\draw [style=thick] (36) to (42);
		\draw [style=thick] (35) to (39);
		\draw [style=thick] (43) to (44);
		\draw [style=thick] (44) to (45);
		\draw [style=thick] (45) to (46);
		\draw [style=thick] (43) to (45);
		\draw [style=thick] (43) to (46);
		\draw [style=thick] (44) to (46);
		\draw [style=thick] (48) to (49);
		\draw [style=thick] (49) to (50);
		\draw [style=thick] (47) to (49);
		\draw [style=thick] (47) to (50);
		\draw [style=thick] (48) to (50);
		\draw [style=thick] (52) to (54);
		\draw [style=thick] (54) to (55);
		\draw [style=thick] (55) to (53);
		\draw [style=thick] (51) to (54);
		\draw [style=thick] (51) to (55);
		\draw [style=thick] (53) to (52);
		\draw [style=thick] (53) to (54);
		\draw [style=thick] (52) to (55);
		\draw [style=thick] (46) to (48);
		\draw [style=thick] (48) to (44);
		\draw [style=thick] (48) to (45);
		\draw [style=thick] (50) to (52);
		\draw [style=thick] (52) to (48);
		\draw [style=thick] (52) to (49);
		\draw [style=thick] (43) to (47);
		\draw [style=thick] (47) to (51);
		\draw [style=thick] (51) to (53);
	\end{pgfonlayer}
\end{tikzpicture}
  \caption{The graphs $H^4_2, H^4_3, H^5_2, H^5_3$ and the vertex labels of $H^5_3$.}
  \label{fig:critical_graphs}
\end{figure}

We set $H_1^t = K_t$ with vertex labels $V(H_1^t) = \{a_1, a'_1, b_1, c_{1, 1}, c_{1, 2}, \ldots, c_{1, t - 3}\}$. For $\gamma \geq 2$, we obtain $H_\gamma^t$ by performing Hajós' construction on $H_{\gamma - 1}^t$ and $K_t$, where we label the vertices of $K_t$ by $a_\gamma, a'_\gamma, b_\gamma, c_{\gamma, 1}, c_{\gamma, 2}, \ldots, c_{\gamma, t - 3}$, with respect to the edges $\{b_{\gamma - 1}, a'_{\gamma - 1}\} \in E(H_{\gamma - 1}^t)$ and $\{a_\gamma, b_\gamma\} \in E(K_t)$ and identifying $a'_{\gamma - 1}$ with $a_\gamma$. Having finished the construction, we will not require the $a'$-labels anymore, in $H_\gamma^t$ we  set $a_{\gamma + 1} = a'_\gamma$, see \cref{fig:critical_graphs}. 

We define $\crit^{t - 1} = \{H^t_\gamma \sep \gamma \in \NN\}$ for $t \geq 3$. Observe that for $t = 3$ we obtain exactly the odd cycles. We establish the following properties of the graphs $H^t_\gamma$ which directly implies \cref{thm:critical_family}.

\begin{lem}
  \label{thm:critical_properties}
  The graphs $H^t_\gamma$, $t \geq 3$, $\gamma \in \NN$, have the following properties:
  \begin{align*}
    \textbf{1. } |V(H^t_\gamma)| = (t - 1) \gamma + 1. && \textbf{2. } H^t_\gamma \text{ is $t$-critical.} && \textbf{3. } \pw(H^t_\gamma) = \tw(H^t_\gamma) = t - 1.
  \end{align*}
\end{lem}

\begin{proof}
  Property 1 can be easily shown by induction over $\gamma$. Now, we can prove that $H^t_\gamma$ is $t$-critical. By \cref{thm:constructible_chromatic} we have that $\chi(H^t_\gamma) \geq t$. Let $v \in V(H^t_\gamma)$ be arbitrary. We have to show that $\chi(H^t_\gamma - v) \leq t - 1$. We distinguish whether $v$ is labeled by $a$, $b$, or $c$, see \cref{fig:critical_coloring}.
  
  \begin{figure}[h]
    \centering
    \begin{tikzpicture}
	\begin{pgfonlayer}{nodelayer}
		\node [style=filled, label={above:$2$}] (0) at (12.5, 7.1) {};
		\node [style=filled, label={above:$1$}] (1) at (11.25, 6.075) {};
		\node [style=filled, label={below:$3$}] (2) at (11.825, 4.75) {};
		\node [style=filled, label={below:$4$}] (3) at (13.175, 4.75) {};
		\node [style=filled, label={above:$2$}] (5) at (13.75, 6.075) {};
		\node [style=filled, label={below:$3$}] (6) at (14.325, 4.75) {};
		\node [style=filled, label={below:$4$}] (7) at (15.675, 4.75) {};
		\node [style=filled, label={above:$1$}] (8) at (17.5, 7.1) {};
		\node [style=filled, label={above:$1$}] (9) at (16.25, 6.075) {};
		\node [style=filled, label={below:$3$}] (11) at (16.825, 4.75) {};
		\node [style=filled, label={below:$4$}] (12) at (18.175, 4.75) {};
		\node [style=filled, label={above:$2$}] (13) at (20, 7.1) {};
		\node [style=filled, label={above:$2$}] (14) at (18.75, 6.075) {};
		\node [style=filled, label={above:$1$}] (15) at (21.25, 6.075) {};
		\node [style=filled, label={below:$3$}] (16) at (19.325, 4.75) {};
		\node [style=filled, label={below:$4$}] (17) at (20.675, 4.75) {};
		\node [style=filled, label={above:$2$}] (18) at (1.25, 7.1) {};
		\node [style=filled, label={above:$1$}] (19) at (0, 6.075) {};
		\node [style=filled, label={below:$3$}] (20) at (0.575, 4.75) {};
		\node [style=filled, label={below:$4$}] (21) at (1.925, 4.75) {};
		\node [style=filled, label={above:$1$}] (22) at (3.75, 7.1) {};
		\node [style=filled, label={above:$2$}] (23) at (2.5, 6.075) {};
		\node [style=filled, label={below:$3$}] (24) at (3.075, 4.75) {};
		\node [style=filled, label={below:$4$}] (25) at (4.425, 4.75) {};
		\node [style=filled, label={above:$2$}] (26) at (6.25, 7.1) {};
		\node [style=filled, label={below:$3$}] (28) at (5.575, 4.75) {};
		\node [style=filled, label={below:$4$}] (29) at (6.925, 4.75) {};
		\node [style=filled, label={above:$1$}] (30) at (8.75, 7.1) {};
		\node [style=filled, label={above:$1$}] (31) at (7.5, 6.075) {};
		\node [style=filled, label={above:$2$}] (32) at (10, 6.075) {};
		\node [style=filled, label={below:$3$}] (33) at (8.075, 4.75) {};
		\node [style=filled, label={below:$4$}] (34) at (9.425, 4.75) {};
		\node [style=filled, label={above:$2$}] (35) at (6.75, 2.35) {};
		\node [style=filled, label={above:$1$}] (36) at (5.5, 1.325) {};
		\node [style=filled, label={below:$3$}] (37) at (6.075, 0) {};
		\node [style=filled, label={below:$4$}] (38) at (7.425, 0) {};
		\node [style=filled, label={above:$1$}] (39) at (9.25, 2.35) {};
		\node [style=filled, label={above:$2$}] (40) at (8, 1.325) {};
		\node [style=filled, label={below:$3$}] (41) at (8.575, 0) {};
		\node [style=filled, label={below:$4$}] (42) at (9.925, 0) {};
		\node [style=filled, label={above:$4$}] (43) at (11.75, 2.35) {};
		\node [style=filled, label={above:$1$}] (44) at (10.5, 1.325) {};
		\node [style=filled, label={below:$3$}] (45) at (11.075, 0) {};
		\node [style=filled, label={above:$2$}] (47) at (14.25, 2.35) {};
		\node [style=filled, label={above:$2$}] (48) at (13, 1.325) {};
		\node [style=filled, label={above:$1$}] (49) at (15.5, 1.325) {};
		\node [style=filled, label={below:$3$}] (50) at (13.575, 0) {};
		\node [style=filled, label={below:$4$}] (51) at (14.925, 0) {};
	\end{pgfonlayer}
	\begin{pgfonlayer}{edgelayer}
		\draw [style=thick] (0) to (1);
		\draw [style=thick] (1) to (2);
		\draw [style=thick] (2) to (3);
		\draw [style=thick] (0) to (2);
		\draw [style=thick] (0) to (3);
		\draw [style=thick] (1) to (3);
		\draw [style=thick] (5) to (6);
		\draw [style=thick] (6) to (7);
		\draw [style=thick] (5) to (7);
		\draw [style=thick] (9) to (11);
		\draw [style=thick] (11) to (12);
		\draw [style=thick] (8) to (11);
		\draw [style=thick] (8) to (12);
		\draw [style=thick] (9) to (12);
		\draw [style=thick] (3) to (5);
		\draw [style=thick] (5) to (1);
		\draw [style=thick] (5) to (2);
		\draw [style=thick] (7) to (9);
		\draw [style=thick] (9) to (5);
		\draw [style=thick] (9) to (6);
		\draw [style=thick] (14) to (16);
		\draw [style=thick] (16) to (17);
		\draw [style=thick] (17) to (15);
		\draw [style=thick] (13) to (16);
		\draw [style=thick] (13) to (17);
		\draw [style=thick] (15) to (14);
		\draw [style=thick] (15) to (16);
		\draw [style=thick] (14) to (17);
		\draw [style=thick] (13) to (15);
		\draw [style=thick] (12) to (14);
		\draw [style=thick] (14) to (9);
		\draw [style=thick] (14) to (11);
		\draw [style=thick] (8) to (13);
		\draw [style=thick] (18) to (19);
		\draw [style=thick] (19) to (20);
		\draw [style=thick] (20) to (21);
		\draw [style=thick] (18) to (20);
		\draw [style=thick] (18) to (21);
		\draw [style=thick] (19) to (21);
		\draw [style=thick] (23) to (24);
		\draw [style=thick] (24) to (25);
		\draw [style=thick] (22) to (24);
		\draw [style=thick] (22) to (25);
		\draw [style=thick] (23) to (25);
		\draw [style=thick] (28) to (29);
		\draw [style=thick] (26) to (28);
		\draw [style=thick] (26) to (29);
		\draw [style=thick] (21) to (23);
		\draw [style=thick] (23) to (19);
		\draw [style=thick] (23) to (20);
		\draw [style=thick] (18) to (22);
		\draw [style=thick] (22) to (26);
		\draw [style=thick] (31) to (33);
		\draw [style=thick] (33) to (34);
		\draw [style=thick] (34) to (32);
		\draw [style=thick] (30) to (33);
		\draw [style=thick] (30) to (34);
		\draw [style=thick] (32) to (31);
		\draw [style=thick] (32) to (33);
		\draw [style=thick] (31) to (34);
		\draw [style=thick] (30) to (32);
		\draw [style=thick] (29) to (31);
		\draw [style=thick] (31) to (28);
		\draw [style=thick] (26) to (30);
		\draw [style=thick] (35) to (36);
		\draw [style=thick] (36) to (37);
		\draw [style=thick] (37) to (38);
		\draw [style=thick] (35) to (37);
		\draw [style=thick] (35) to (38);
		\draw [style=thick] (36) to (38);
		\draw [style=thick] (40) to (41);
		\draw [style=thick] (41) to (42);
		\draw [style=thick] (39) to (41);
		\draw [style=thick] (39) to (42);
		\draw [style=thick] (40) to (42);
		\draw [style=thick] (44) to (45);
		\draw [style=thick] (43) to (45);
		\draw [style=thick] (38) to (40);
		\draw [style=thick] (40) to (36);
		\draw [style=thick] (40) to (37);
		\draw [style=thick] (42) to (44);
		\draw [style=thick] (44) to (40);
		\draw [style=thick] (44) to (41);
		\draw [style=thick] (35) to (39);
		\draw [style=thick] (39) to (43);
		\draw [style=thick] (48) to (50);
		\draw [style=thick] (50) to (51);
		\draw [style=thick] (51) to (49);
		\draw [style=thick] (47) to (50);
		\draw [style=thick] (47) to (51);
		\draw [style=thick] (49) to (48);
		\draw [style=thick] (49) to (50);
		\draw [style=thick] (48) to (51);
		\draw [style=thick] (47) to (49);
		\draw [style=thick] (48) to (44);
		\draw [style=thick] (48) to (45);
		\draw [style=thick] (43) to (47);
	\end{pgfonlayer}
\end{tikzpicture}
    \caption{The three coloring cases in the proof of \cref{thm:critical_properties} for the graph $H^5_4$.}
    \label{fig:critical_coloring}
  \end{figure}
  
  \textbf{Case $v = a_i$, $i \in [\gamma + 1]$:} We construct a proper $(t-1)$-coloring $\varphi \colon V(H^t_\gamma - v) \rightarrow [t-1]$. For $\ell \in [\gamma]$, set $\varphi(b_\ell) = (\ell \mod 2) + 1$. For $\ell \in [\gamma + 1]$, set $\varphi(a_\ell) = (\ell + 1 \mod 2) + 1$ if $\ell < i$ and set $\varphi(a_\ell) = (\ell \mod 2) + 1$ if $\ell > i$. For $\ell \in [\gamma]$ and $k \in [t - 3]$, set $\varphi(c_{\ell,k}) = k + 2$. Now, for all $\ell \in [\gamma]$, the $(t-1)$-clique induced by $a_\ell, a_{\ell+1}, c_{\ell,1}, \ldots, c_{\ell,t-3}$ is colored properly and the graph induced by $N[b_\ell]$ is colored properly. Since these cover all edges, $\varphi$ must be a proper coloring.
  
  \textbf{Case $v = b_i$, $i \in [\gamma]$:} Again, we construct a proper $(t-1)$-coloring $\varphi$. For $\ell \in [\gamma]$, set $\varphi(b_\ell) = (\ell \mod 2) + 1$ if $\ell < i$ and set $\varphi(b_\ell) = (\ell + 1 \mod 2) + 1$ if $\ell > i$. For $\ell \in [\gamma + 1]$, set $\varphi(a_\ell) = (\ell + 1 \mod 2) + 1$. For $\ell \in [\gamma]$ and $k \in [t-3]$, set $\varphi(c_{\ell, k}) = k + 2$. Similarly to the previous case, $\varphi$ must be a proper coloring.
  
  \textbf{Case $v = c_{i,k}$, $i \in [\gamma]$, $k \in [t-3]$:} We construct a proper $(t-1)$-coloring $\varphi$ as follows. We set $\varphi(b_i) = k + 2$. For $\ell \in [\gamma] \setminus \{i\}$, set $\varphi(b_\ell) = (\ell \mod 2) + 1$ if $\ell < i$ and set $\varphi(b_\ell) = (\ell + 1 \mod 2) + 1$ if $\ell > i$. For $\ell \in [\gamma + 1]$, set $\varphi(a_\ell) = (\ell + 1 \mod 2) + 1$. For $(\ell, k') \in ([\gamma] \times [t - 3]) \setminus \{(i,k)\}$, set $\varphi(c_{\ell,k'}) = k' + 2$. Again, like in the previous cases, $\varphi$ must be a proper coloring.
  
  This shows that $H^t_\gamma$ is $t$-critical. We proceed by showing that $\pw(H^t_\gamma) = \tw(H^t_\gamma) = t - 1$. Note that by contracting all vertices in $V(H^t_\gamma) \setminus \{a_1, a_2, b_1, c_{1, 1}, \ldots, c_{1, t-3}\}$ into $b_1$, we obtain $K_t$ and hence $\pw(G) \geq \tw(G) \geq \tw(K_t) = t - 1$, as treewidth is a minor-monotone parameter. For the other direction, we construct a path decomposition $\TT$ of width $t - 1$. Let $\TT$ be the path on  $2\gamma - 1$ vertices denoted by $w_1, \ldots, w_{2\gamma - 1}$. For $\ell \in [\gamma]$, we define the bag $\bag_{w_{2\ell - 1}} = \{a_\ell, a_{\ell + 1}, b_\ell, c_{\ell, 1}, c_{\ell, 2}, \ldots, c_{\ell, t - 3}\}$ and for $\ell \in [\gamma - 1]$, we define the bag $\bag_{w_{2\ell}} = \{b_\ell, b_{\ell + 1}, a_{\ell + 1}\}$. This yields a path decomposition of width $t - 1$ and hence establishes property 3.
\end{proof}

\subsection{Dense Setting}
\label{sec:dtc_lb_cw}

We again view \emph{solutions} to \DTC as functions $\coloring \colon V(G) \rightarrow [\ncols] \cup \{\del\}$ satisfying $\coloring(u) = \bot$ or $\coloring(v) = \bot$ or $\coloring(u) \neq \coloring(v)$ for every edge $\{u,v\} \in E(G)$. 

We will prove a lower bound for \DTC parameterized by the size of a twinclass-modulator (TCM) to treewidth $\ncols$. This implies a lower bound for parameterization by twinclass-treedepth due to \cref{thm:tcm_tctd}. By \cref{thm:hierarchy_cliquewidth}, this further extends to lower bounds for parameterization by twinclass-pathwidth and clique-width. Hence, we will see that the running time of the algorithm from \cref{thm:dtc_algo} is tight, unless \SETH is false.

\begin{mythm}
  \label{thm:dtc_lower_bound}
  If \DTC can be solved in time $\Oh^*((2^{\ncols} - \eps)^{|\tcm|})$ for some $\ncols \geq 2$ and $\eps > 0$, where $\tcm$ is a TCM to treewidth $\ncols$, then \SETH is false.
\end{mythm}

Assume that we can solve \DTC in time $\Oh^*((2^\ncols - \eps)^{|\tcm|})$ for some $\ncols \geq 2$ and $\eps > 0$. For all clause sizes $\clss$, we provide a reduction from $\clss$-\SAT with $\nvars$ variables to $\DTC$ with a TCM to treewidth $\ncols$ of size approximately $\nvars \log_{2^\ncols}(2) = \nvars / \ncols$ to convert from base $2^\ncols$ to base $2$. Combining this reduction with the assumed faster algorithm will imply a faster algorithm for $\clss$-\SAT, thus violating \SETH. From now on, we consider $\ncols$ to be fixed.

Before we begin with the actual construction, we describe several gadgets that will be important.

\begin{figure}
 \centering
 \begin{tikzpicture}
	\begin{pgfonlayer}{nodelayer}
		\node [style=filled] (0) at (-1, 0) {};
		\node [style=filled] (1) at (3, 0) {};
		\node [style=filled] (2) at (1, 3.5) {};
		\node [style=filled] (3) at (5, 0) {};
		\node [style=filled] (4) at (9, 0) {};
		\node [style=filled] (5) at (7, 3.5) {};
		\node [style=filled] (6) at (7, 1.5) {};
		\node [style=none] (7) at (-1, 0.75) {$u$};
		\node [style=none] (8) at (3, 0.75) {$v$};
		\node [style=none] (9) at (5, 0.75) {$u$};
		\node [style=none] (10) at (9, 0.75) {$v$};
		\node [style=none] (11) at (1, 4.25) {$w_1$};
		\node [style=none] (12) at (7, 4.25) {$w_1$};
		\node [style=none] (13) at (7, 0.75) {$w_2$};
	\end{pgfonlayer}
	\begin{pgfonlayer}{edgelayer}
		\draw [style=thick] (0) to (1);
		\draw [style=thick] (0) to (2);
		\draw [style=thick] (2) to (1);
		\draw [style=thick] (5) to (6);
		\draw [style=thick] (6) to (3);
		\draw [style=thick] (3) to (5);
		\draw [style=thick] (5) to (4);
		\draw [style=thick] (4) to (6);
		\draw [style=thick] (3) to (4);
	\end{pgfonlayer}
\end{tikzpicture}
 \caption{A deletion edge between $u$ and $v$ for $\ncols = 2$ and $\ncols = 3$.}
 \label{fig:dtc_del_edge}
\end{figure}

\subparagraph*{Deletion edges.} Let $u$ and $v$ be two vertices. By adding a \emph{deletion edge} between $u$ and $v$, we mean adding $\ncols - 1$ vertices $w_1, \ldots, w_{\ncols - 1}$ and edges so that $\{u,v,w_1,w_2, \ldots, w_{\ncols - 1}\}$ is a clique of size $\ncols + 1$, see \cref{fig:dtc_del_edge}. The vertices $w_1, \ldots, w_{\ncols - 1}$ will not receive any further incident edges. Therefore, any solution to \DTC has to delete at least one vertex in this clique and we can assume that the deleted vertex is $u$ or $v$.

\subparagraph*{Arrows.} Our construction relies on being able to propagate deletions throughout the graph. This is done via so-called \emph{arrows}. We distinguish between two types of arrows, namely \emph{thin arrows} and \emph{thick arrows}. Thin arrows arise already in the sparse setting of the lower bound construction of Lokshtanov et al.~\cite{LokshtanovMS18} for the restricted case of $\ncols = 2$, i.e., \OCT parameterized by pathwidth. A thin arrow simply propagates a deletion from a single vertex to another vertex. Whereas the newly introduced thick arrows propagate a deletion to a single vertex only if sufficiently many vertices have been deleted in some true twinclass of size $\ncols$. The construction of thick arrows is described in the proof of \cref{thm:arrow_behavior}.

Let us proceed with the construction of thin arrows. Given two single vertices $u$ and $v$, \emph{adding a thin arrow} from $u$ to $v$ means adding a new vertex $w$ and a deletion edge between $u$ and $w$ and a deletion edge between $w$ and $v$. We remark that, compared to the construction of Lokshtanov et al.~\cite{LokshtanovMS18}, we have shortened the thin arrows. The construction of thin arrows is symmetric, but the direction will be important later for the description of a packing that witnesses a lower bound on the required budget. The idea is that when $u$ is not deleted, we delete $w$ to resolve both incident deletion edges; when $u$ is deleted, the first deletion edge is resolved and we can afford to delete $v$ to resolve the second deletion edge. The former is called the \emph{passive} solution of the thin arrow and the latter is the \emph{active} solution. Using exchange arguments, one can see that it is sufficient to only consider solutions that on each thin arrow use either the passive solution or active solution.

\begin{lem}
  \label{thm:arrow_behavior}
  Let $U$ be a set of $\ncols$ true twins and $v$ be a vertex that is not adjacent to $U$ and $\ell \in [\ncols]$. There is a graph $\arrow = \arrow_\ell(U,v)$ with the following properties:
  \begin{itemize}
   \item Any solution $\coloring$ satisfies $|V(\arrow - U) \cap \delset| \geq \ell$.
   \item If $|V(\arrow - U) \cap \delset| = \ell$, then $\coloring(v) = \del$ implies that $|U \cap \delset| \geq \ell$.
   \item Any solution $\tilde{\coloring}$ for $\tilde{G} = G[(V \setminus V(\arrow)) \cup U \cup \{v\}]$ with $\tilde{\coloring}(v) \neq \del$ or $|U \cap \tilde{\coloring}^{-1}(\del)| \geq \ell$ can be extended to a solution $\coloring$ for $G$ with $|V(\arrow - U) \cap \delset| = \ell$.
   \item $\tw(\arrow - U) \leq \ncols$ via a winning cops-and-robber-strategy starting with a cop on $v$.
  \end{itemize}
\end{lem}

\begin{figure}
 \centering
 \begin{tikzpicture}
	\begin{pgfonlayer}{nodelayer}
		\node [style=empty 30pt] (17) at (-2.5, 0) {$K_r$};
		\node [style=empty 30pt] (19) at (7.5, 0) {$K_{r - \ell}$};
		\node [style=none] (20) at (2.5, -4) {$\cdots$};
		\node [style=none] (21) at (1, -4.5) {};
		\node [style=none] (22) at (4, -4.5) {};
		\node [style=none] (23) at (2.5, -5) {};
		\node [style=none] (24) at (2.5, -5.5) {$I_{\ell - 1}$};
		\node [style=none] (25) at (-2.5, -1.5) {$U$};
		\node [style=filled 10pt] (26) at (3.5, -4) {};
		\node [style=filled 10pt] (27) at (1.5, -4) {};
		\node [style=filled 10pt] (28) at (11.5, 0) {};
		\node [style=none] (29) at (11.5, -0.75) {$v$};
		\node [style=filled 10pt] (30) at (1, 0) {};
		\node [style=filled 10pt] (31) at (2, 0) {};
		\node [style=filled 10pt] (32) at (4, 0) {};
		\node [style=none] (33) at (3, 0) {$\cdots$};
		\node [style=none] (34) at (2.5, 1.5) {};
		\node [style=none] (35) at (0, 0) {};
		\node [style=none] (36) at (5, 0) {};
		\node [style=none] (37) at (2.5, -1.5) {};
		\node [style=none] (38) at (4.5, 1) {};
		\node [style=none] (39) at (0.5, 1.5) {};
		\node [style=none] (40) at (4.5, 1.5) {};
		\node [style=none] (41) at (2.5, 2) {};
		\node [style=none] (42) at (2.5, 2.5) {$K_\ell$};
		\node [style=none] (43) at (-2.75, 2.5) {$A_\ell(U,v)\colon$};
	\end{pgfonlayer}
	\begin{pgfonlayer}{edgelayer}
		\draw [style=thin, in=105, out=-90] (21.center) to (23.center);
		\draw [style=thin, in=-90, out=75] (23.center) to (22.center);
		\draw [style=thick] (19) to (28);
		\draw [style=thick] (30) to (31);
		\draw [style=thick, bend left] (31) to (32);
		\draw [style=thick, bend left=45] (30) to (32);
		\draw [style=thick, in=180, out=-90] (35.center) to (37.center);
		\draw [style=thick, in=-90, out=0] (37.center) to (36.center);
		\draw [style=thick, in=360, out=90] (36.center) to (34.center);
		\draw [style=thick, in=90, out=-180] (34.center) to (35.center);
		\draw [style=del] (30) to (27);
		\draw [style=del] (27) to (31);
		\draw [style=del] (27) to (32);
		\draw [style=del] (30) to (26);
		\draw [style=del] (31) to (26);
		\draw [style=del] (32) to (26);
		\draw [style=thick] (36.center) to (19);
		\draw [style=thick] (17) to (35.center);
		\draw [style=thick, bend left=45, looseness=0.75] (38.center) to (28);
		\draw [style=thin, in=-120, out=90, looseness=0.75] (39.center) to (41.center);
		\draw [style=thin, in=90, out=-60, looseness=0.75] (41.center) to (40.center);
	\end{pgfonlayer}
\end{tikzpicture}
 \caption{The construction of a thick $\ell$-arrow $\arrow_\ell(U,v)$. The double-lined edges represent deletion edges. An edge attached to a circle or ellipse represents a \emph{join}, i.e., a bundle of edges.}
 \label{fig:dtc_arrow}
\end{figure}
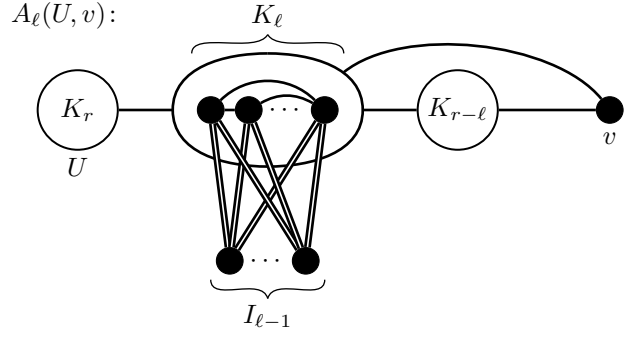

\begin{proof}
  The graph $\arrow = \arrow_\ell(U,v)$ is constructed by adding a clique $K_\ell$ on $\ell$ vertices, a clique $K_{\ncols - \ell}$ on $\ncols - \ell$ vertices, an independent set $I_{\ell - 1}$ on $\ell - 1$ vertices and adding all edges between $U$ and $K_\ell$, all edges between $K_\ell$ and $K_{\ncols - \ell}$, all edges between $K_\ell$ and $v$, all edges between $K_{\ncols - \ell}$ and $v$, and a deletion edge between each vertex of $K_\ell$ and each vertex of $I_{\ell - 1}$, cf.\ \cref{fig:dtc_arrow}. Observe that the $K_\ell$, $K_{\ncols - \ell}$, and $v$ together form a clique of size $\ncols + 1$. Only $U$ and $v$ can have further incident edges. 
  
  We first show that any solution $\coloring$ has to delete at least $\ell$ vertices in $\arrow - U$. Suppose that no deletion occurs in $K_\ell$, then $\ell - 1$ disjoint deletion edges between $K_\ell$ and $I_{\ell - 1}$ remain and at least one deletion has to occur in $K_{\ncols - \ell} \cup \{v\}$. On the other hand, suppose that $1 \leq d \leq \ell$ deletions occur in $K_\ell$, then $\ell - d$ disjoint deletion edges between $K_\ell$ and $I_{\ell - 1}$ remain. In either case, we must delete at least $\ell$ vertices.
  
  Next, we argue that a solution $\coloring$ deleting exactly $\ell$ vertices in $\arrow - U$ cannot delete $v$ if less than $\ell$ deletions occur in $U$. Suppose that $d < \ell$ deletions occur in $U$, then at least $\ell - d \geq 1$ deletions must occur in $K_\ell$, because $U$ and $K_\ell$ together form a clique of size $\ncols + \ell$. Let $d' \geq 1$ be the number of deletions in $K_\ell$, then like in the previous paragraph $\ell - d'$ disjoint deletion edges between $K_\ell$ and $I_{\ell - 1}$ remain. Hence, we must already perform a total of $\ell$ deletions in $K_\ell$ and the deletion edges, so $v$ cannot be deleted.
  
  Consider any solution $\tilde{\coloring}$ for $\tilde{G}$. If $\tilde{\coloring}(v) \neq \del$, then we can extend $\tilde{\coloring}$ to a solution $\coloring$ for $G$ by deleting $K_\ell$ completely. This deletes all edges between $U$ and $A$ and one endpoint of each deletion edge. The $K_{\ncols - \ell}$ is only attached to $v$ and can always be properly colored. We give all vertices in $I_{\ell - 1}$ the same color and the cliques attached to $I_{\ell - 1}$ due to the deletion edges are all disjoint and of size $\ncols - 1$ each, so they can be properly colored as well. Hence, $\coloring$ is a valid solution for $G$. 
  
  If $|U \cap \tilde{\coloring}^{-1}(\del)| \geq \ell$ and $\tilde{\coloring}(v) = \del$, then we can extend $\tilde{\coloring}$ to a solution $\coloring$ for $G$ by deleting $I_{\ell - 1}$ completely. To $K_\ell$ we assign a subset of the colors $[\ncols] \setminus \tilde{\coloring}(U)$, this is possible since at least $\ell$ deletions occur in $U$. The $K_\ell$ and $K_{\ncols - \ell}$ together form a clique of size $\ncols$ that we can properly color. The cliques of the deletion edges are all disjoint as before and only attached to $K_\ell$ and have size $\ncols - 1$ each, hence they can also be properly colored. Hence, $\coloring$ is a valid solution for $G$.
  
  Finally, we argue that $\tw(\arrow - U) \leq \ncols$ using the omniscient cops-and-robber-game. We begin by placing $\ncols + 1$ cops on $K_\ell \cup K_{\ncols - \ell} \cup \{v\}$, then we remove the $\ncols - \ell + 1$ cops on $K_{\ncols - \ell} \cup \{v\}$. Now, all remaining connected components correspond to a vertex of $I_{\ell - 1}$. Fix the component the robber escaped to and place one cop on the corresponding vertex of $I_{\ell - 1}$. Now, the robber must have escaped to the inside of some deletion edge between $K_\ell$ and $I_{\ell - 1}$. We remove the $\ell - 1$ cops from $K_\ell$ that are not on the endpoint of this deletion edge. Finally, we place $\ncols - 1$ cops on the remaining vertices of the deletion edge and capture the robber. Throughout the strategy, we have never placed more than $\ncols + 1$ cops simultaneously and hence the claimed treewidth bound holds by \cref{thm:cops_robber}.
\end{proof}

Similarly to the thin arrows, we say that a thick $\ell$-arrow $\arrow_\ell(U,v)$ is \emph{active} if the considered solution $\coloring$ satisfies $|U \cap \delset| \geq \ell$ and hence \cref{thm:arrow_behavior} allows us to assume that $\coloring(v) = \del$; otherwise, we say that $\arrow_\ell(U,v)$ is \emph{passive}.

\subparagraph*{Construction setup.} We will now begin with the construction of the \DTC instance. Consider a $\clss$-\SAT instance $\formula$ with $\nvars$ variables and $\nclss$ clauses. We enumerate the clauses and simply refer to them by their number. Depending only on $\eps$ and $\ncols$, we will choose an integer $\vgrpsize$; we will describe how to choose $\vgrpsize$ later. We partition the variables of $\formula$ into groups of size at most $\vgrpsize$, resulting in $\ngrps = \lceil \nvars / \vgrpsize \rceil$ groups which will be indexed by $i$. Next, we choose the smallest integer $\grpsize$ such that $\grpsize$ is divisible by $2^{\ncols}$ and $(2^{\ncols})^{\grpsize} \frac{(2^{\ncols} - 1)!}{2^{(2^\ncols)}} \grpsize^{-(2^\ncols)} \geq 2^\vgrpsize$. 

Usually, a simple base conversion from a base that is a power of $2$ to base $2$ allows for a construction where the handling of groups is a lot less technical. Unfortunately, since we require the trick of Cygan et al.~\cite{CyganDLMNOPSW16} that enlarges the groups, this is not possible here. We will now describe the construction of the \DTC instance $G = G(\formula, \ncols, \vgrpsize)$.

\subparagraph*{Central twinclasses.} The central vertices of $G$ will form the TCM to treewidth $\ncols$. For each variable group $i \in [\ngrps]$, we create a set $\ttcset_i \subseteq \tcpartition(G)$ consisting of $\grpsize$ true twinclasses of size $\ncols$ each; there are no edges between the different twinclasses.  The partial solution induced by a solution $\coloring$ of \DTC on the twinclasses in $\ttcset_i$ will correspond to a truth assignment for the $i$-th variable group.

\subparagraph*{Central clique.} We create a clique $F = \{f_1, \ldots, f_\ncols\}$ on $\ncols$ vertices which also belongs to the central part. The vertices of $F$ will not be twins; each vertex of $F$ has its own twinclass.  This concludes the construction of the central part which will form the modulator $\tcm$.

Our construction will ensure that no vertex of $F$ can be deleted; we will prove so later. This allows us to use $F$ to simulate \LC constraints, i.e., we can forbid specific colors at a vertex. To do so, we normalize the considered solutions by assuming that $\coloring(f_s) = s$ for all $s \in [\ncols]$ and then no vertex $v$ adjacent to some $f_s$ can receive color $s$. We say that a solution $\coloring$ is \emph{normalized} if $F \cap \delset = \emptyset$ and $\coloring(f_s) = s$ for all $s \in [\ncols]$.

We mention here that in principle a twinclass-modulator could contain arbitrarily large twinclasses. In our case, every twinclass has size at most $\ncols = \Oh(1)$, so the number $|\bigcup(\tcm)|$ of vertices in $\tcm$ is linear in the size $|\tcm|$ of $\tcm$.

\subparagraph*{Budget.} The budget $\budget = \cost_\packing + \ngrps \ncols \grpsize / 2$ for the \DTC instance is split into two parts. The first part $\cost_\packing$ is due to a vertex-disjoint packing $\packing$ of subgraphs that will be described later and the second part $\ngrps \ncols \grpsize / 2$ is allocated to the vertices in the twinclass-modulator.

\subparagraph*{Partial solution structure.}
We want to ensure that under our budget restriction no vertex of $F$ can be deleted, as otherwise the simulation of \LC constraints will not work. This entails that all considered solutions should perform the same number of deletions. On the gadgets that will be attached to the central vertices, this does not pose a big issue. However, it does on the central vertices: to obtain the desired lower bound, for each true twinclass $U \in \ttcset$ and solution $\coloring$, it should in principle be possible for $\coloring(U)$ to attain every subset of $[\ncols]$, so that we have $2^\ncols$ states per twinclass. However, the number of deletions in $U$ depends on $|\coloring(U)|$. We can partition the possible states into different \emph{levels} depending on the number $\ell$ of deletions they incur in the true twinclass. A state can incur any number of deletions between $0$ and $\ncols$.

Using slightly larger groups like Cygan et al.~\cite{CyganDLMNOPSW16} allows us to overcome these issues. By using slightly more twinclasses per variable group $i \in [\ngrps]$ than necessitated by the base conversion, we can consider solutions with a special structure on the central vertices. It is designed in such a way that all solutions obeying this structure incur the same number of deletions on the central vertices while still allowing all states to appear the same number of times. Due to the slight increase in the number of twinclasses, we can still encode sufficiently many truth assignments into such structured solutions. 

For every $i \in [\ngrps]$, $U \in \ttcset_i$, $C \subseteq [\ncols]$, we fix some representative solution $\coloring_{i,U,C} \colon U \rightarrow [\ncols] \cup \{\del\}$ to \DTC on $U$ with $\coloring_{i,U,C}(U) \setminus \{\del\} = C$. The point of this is that we only distinguish between different solutions $\coloring$ based on the sets $\coloring(U)$, $U \in \ttcset_i$, and not based on the actual mappings $\coloring\big|_U$, $U \in \ttcset_i$. Since there are exactly $\binom{\ncols}{\ncols - \ell} = \binom{\ncols}{\ell}$ states that incur exactly $\ell$ deletions each and we want to allow for the possibility of all states appearing the same number of times, we arrive at the following definition for each group $i \in [\ngrps]$:
\begin{align*}
  \scolorings_i = \{\coloring & \colon \bigcup(\ttcset_i) \rightarrow [\ncols] \cup \{\del\} \sep \coloring \text{ is solution of \DTC on $\bigcup(\ttcset_i)$} \\
  & \left|\left\{U \in \ttcset_i \sep \coloring(U) \setminus \{\del\} \in \binom{[\ncols]}{\ell}\right\}\right| = \binom{\ncols}{\ell} \frac{\grpsize}{2^{\ncols}} \text{ for all } \ell = 0, \ldots, \ncols \text{ and } \\  
& \text{ with } \coloring\big|_U \equiv \coloring_{i,U,\coloring(U) \setminus \{\del\}} \text{ for all } U \in \ttcset_i \}.
\end{align*}
Informally, this means that we only consider partial solutions that pick a representative solution on each twinclass in $\ttcset_i$ and for every level $\ell \in \{0\} \cup [\ncols]$, i.e., a possible number of deletions, there is a fixed number $\binom{\ncols}{\ell} \frac{\grpsize}{2^\ncols}$ of twinclasses in $\ttcset_i$ where the partial solution attains a state of level $\ell$. \cref{thm:scoloring_deletions} gives the total number of deletions for a partial solution $\coloring_i \in \scolorings_i$ and \cref{thm:scolorings_bound} studies the size of $\scolorings_i$.

\begin{lem}
  \label{thm:scoloring_deletions}
  If $\coloring \in \scolorings_i$, then $|\delset| = \ncols \grpsize / 2$.
\end{lem}

\begin{proof}
  \begin{equation*}
    |\delset| = \sum_{\ell = 0}^\ncols (\ncols - \ell) \binom{r}{\ell} \frac{\grpsize}{2^\ncols} 
    = \frac{\grpsize}{2^\ncols} \ncols \sum_{\ell = 0}^{\ncols - 1} \binom{\ncols - 1}{\ell} 
    = \frac{\grpsize}{2^\ncols} \ncols 2^{\ncols - 1}
    = \ncols \grpsize / 2. \qedhere
  \end{equation*}
\end{proof}

\begin{lem}
 \label{thm:scolorings_bound}
  If $\grpsize \geq 2^{\ncols}$, then we have that $|\scolorings_i| \geq (2^{\ncols})^{\grpsize} \frac{(2^{\ncols} - 1)!}{2^{(2^\ncols)}} \grpsize^{-(2^\ncols)}$.
\end{lem}

\begin{proof}
  Observe that 
	\begin{equation*}  
  \left|\left\{\coloring \in \scolorings_i \sep \left|\left\{U \in \ttcset_i \sep \coloring(U) \setminus \{\del\} = C\right\}\right| = \frac{\grpsize}{2^{\ncols}} \text{ for all } C \subseteq [\ncols]\right\}\right| = \binom{\grpsize}{\frac{\grpsize}{2^{\ncols}}, \ldots, \frac{\grpsize}{2^{\ncols}}} = x,
	\end{equation*}  
where $x$ is the central multinomial coefficient. The central multinomial coefficient $x$ is a maximum of the function $(c_1, \ldots, c_{2^{\ncols}}) \mapsto \binom{\grpsize}{c_1, \ldots, c_{2^{\ncols}}}$. Hence, $x$ is a maximum term in the sum of the multinomial theorem, i.e., $(2^\ncols)^\grpsize = \sum_{c_1 + \cdots + c_{2^\ncols} = \grpsize} \binom{\grpsize}{c_1, \ldots, c_\ncols}$. The number of terms in the sum of the multinomial theorem is the number of weak compositions of $\grpsize$ into $2^{\ncols}$ parts which is $\binom{\grpsize + 2^{\ncols} - 1}{\grpsize}$. Bounding this term from above using $\grpsize \geq 2^{\ncols}$, we obtain
  \begin{equation*}
    \binom{\grpsize + 2^{\ncols} - 1}{\grpsize} 
    = \frac{1}{(2^{\ncols} - 1)!} (\grpsize + 1) \cdots (\grpsize + 2^{\ncols} - 1) 
    \leq \frac{1}{(2^{\ncols} - 1)!} (2\grpsize)^{(2^{\ncols})} = \frac{2^{2^\ncols}}{(2^{\ncols} - 1)!} \grpsize^{(2^\ncols)}.
  \end{equation*}
By the multinomial theorem, we hence obtain that 
  \begin{equation*}
    |\scolorings_i| \geq x \geq (2^{\ncols})^{\grpsize} \frac{(2^{\ncols} - 1)!}{2^{(2^\ncols)}} \grpsize^{-(2^\ncols)}. \qedhere
  \end{equation*}
\end{proof}
 
Hence, if we choose $\grpsize$ as discussed previously, then we can pick for each group $i \in [\ngrps]$ an efficiently computable injective mapping $\embedding_i \colon \{0,1\}^{\vgrpsize} \rightarrow \scolorings_i$ that maps truth assignments of the $i$-th variable group to partial solutions on $\ttcset_i$ with the desired structure.

\begin{figure}
  \centering
  \input{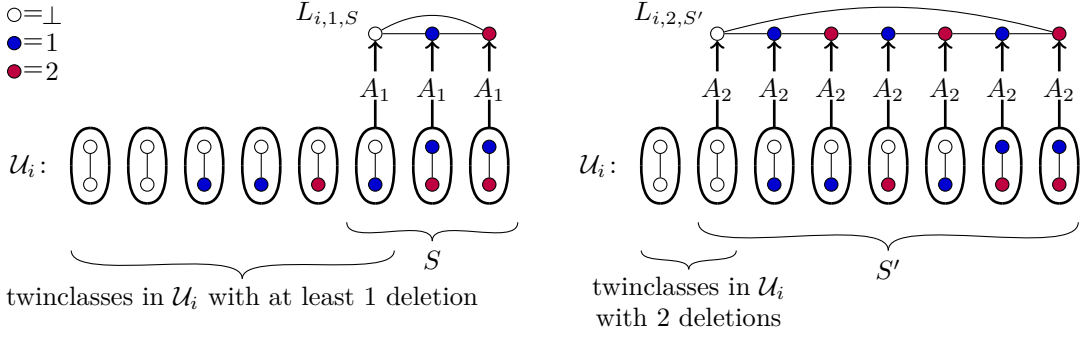}
  \caption{The construction of structure gadgets $L_{i,1,S}$ and $L_{i,2,S'}$ for the case of $\ncols = 2$, i.e., \OCT, and $\grpsize = 8$ and its connection to the central twinclasses. The large ellipses represent twinclasses and the arrows represent thick $1$-arrows or thick $2$-arrows. The depicted solution $\coloring_i$ belongs to $\scolorings_i$ and propagates a deletion to every attached structure gadget. Note that $\coloring_i$ does not use each color set the same number of times.}
  \label{fig:dtc_structure_gadget}
\end{figure}

\subparagraph*{Structure gadgets.} We proceed by constructing the \emph{structure gadgets} to enforce that the partial solution on $\ttcset_i$ belongs to $\scolorings_i$. For every group $i \in [\ngrps]$, number of deletions $\ell \in [\ncols]$, set of twinclasses $\setfam \subseteq \ttcset_i$ with $|\setfam| = \binom{\ncols}{\geq (\ncols - \ell + 1)}\frac{\grpsize}{2^{\ncols}} + 1$, we add an $(\ncols + 1)$-critical graph $L_{i, \ell, \setfam} \in \crit^\ncols$ consisting of at least $|\setfam|$ vertices. For every $U \in \setfam$, we pick a private vertex $v$ in $L_{i, \ell, \setfam}$ and add the thick $\ell$-arrow $\arrow_\ell(U,v)$. We create $1 + \ngrps \ncols \grpsize / 2 = 1 + (\budget - \cost_\packing)$ copies of $L_{i, \ell, \setfam}$ and the incident thick arrows. This concludes the construction of the structure gadget, cf.\ \cref{fig:dtc_structure_gadget}.

Due to our budget constraint, no solution will be able to perform too many deletions in the central vertices. The idea of the structure gadget is that if a solution $\coloring$ performs too few deletions in the twinclasses of $\ttcset_i$, then we can find an $\ell$ and an $\setfam \subseteq \ttcset_i$ with $|\setfam| = \binom{\ncols}{\geq (\ncols - \ell + 1)}\frac{\grpsize}{2^{\ncols}} + 1$ such that $\coloring$ deletes less than $\ell$ vertices in each $U \in \setfam$. Hence, all thick $\ell$-arrows leading to $L_{i, \ell, \setfam}$ and its copies have to be passive. To resolve all these $(\ncols + 1)$-critical graphs, we have to spend one extra unit of budget per copy which is not accounted for by the packing. Due to the large number of copies, this implies that we must violate our budget constraint.

\subparagraph{Color-set-gadgets.} The next step is to construct gadgets that can distinguish between the different partial solutions in $\scolorings_i$. To do so, these gadgets must be able to react not only based on the number of deletions inside a twinclass, but also based on the set of colors used for a twinclass. In \cref{thm:decode_behavior}, we construct \emph{color-set-gadgets} $\decode_C(U,v)$ using \LC constraints simulated by the central clique $F$. Given a solution $\coloring$ and a set of colors $C \subsetneq [\ncols]$, the color-set-gadget $\decode_C(U,v)$ propagates a deletion to the vertex $v$ when the set of colors $\coloring(U) \setminus \{\del\}$ used on a set of true twins $U$ is a subset of $C$, i.e., $\coloring(U) \setminus \{\del\} \subseteq C$. We say that $\decode_C(U,v)$ is \emph{active} if a deletion is propagated to $v$ and \emph{passive} otherwise.

\newcommand{\ovc}{{\overline{c}}}

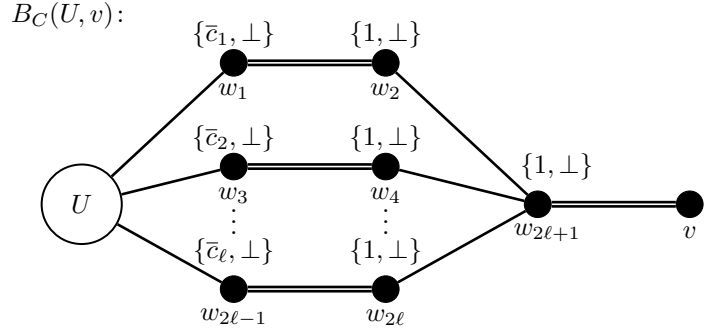
\begin{figure}
 \centering
 \begin{tikzpicture}
	\begin{pgfonlayer}{nodelayer}
		\node [style=empty 30pt] (0) at (-2, 1) {$U$};
		\node [style=none] (3) at (2, 0.75) {$\vdots$};
		\node [style=none] (7) at (6, 0.75) {$\vdots$};
		\node [style=none] (11) at (2, 5.5) {$\{\overline{c}_1, \bot\}$};
		\node [style=none] (13) at (2, -0.25) {$\{\overline{c}_\ell, \bot\}$};
		\node [style=none] (14) at (6, 5.5) {$\{1, \bot\}$};
		\node [style=none] (15) at (6, 2.75) {$\{1, \bot\}$};
		\node [style=none] (16) at (10.5, 2) {$\{1, \bot\}$};
		\node [style=none] (17) at (6, -0.25) {$\{1, \bot\}$};
		\node [style=filled 10pt] (18) at (2, 4.75) {};
		\node [style=filled 10pt] (19) at (6, 4.75) {};
		\node [style=filled 10pt] (20) at (2, 2) {};
		\node [style=filled 10pt] (21) at (6, 2) {};
		\node [style=filled 10pt] (22) at (2, -1.25) {};
		\node [style=filled 10pt] (23) at (6, -1.25) {};
		\node [style=none] (24) at (2, 4) {$w_1$};
		\node [style=none] (25) at (6, 4) {$w_2$};
		\node [style=none] (26) at (2, 1.25) {$w_3$};
		\node [style=none] (27) at (6, 1.25) {$w_4$};
		\node [style=none] (28) at (2, -2) {$w_{2\ell - 1}$};
		\node [style=none] (29) at (6, -2) {$w_{2\ell}$};
		\node [style=filled 10pt] (32) at (10, 1) {};
		\node [style=filled 10pt] (33) at (14, 1) {};
		\node [style=none] (34) at (10.25, 0.25) {$w_{2\ell+1}$};
		\node [style=none] (35) at (14, 0.25) {$v$};
		\node [style=none] (36) at (-2.25, 6) {$B_C(U,v)\colon$};
		\node [style=none] (37) at (2, 2.75) {$\{\overline{c}_2, \bot\}$};
	\end{pgfonlayer}
	\begin{pgfonlayer}{edgelayer}
		\draw [style=del] (18) to (19);
		\draw [style=del] (20) to (21);
		\draw [style=del] (22) to (23);
		\draw [style=thick] (0) to (22);
		\draw [style=thick] (0) to (20);
		\draw [style=thick] (0) to (18);
		\draw [style=thick] (19) to (32);
		\draw [style=thick] (21) to (32);
		\draw [style=thick] (23) to (32);
		\draw [style=del] (32) to (33);
	\end{pgfonlayer}
\end{tikzpicture}
 \caption{The construction of a color-set-gadget $\decode_C(U,v)$ with $[\ncols] \setminus C = \{\ovc_1, \ldots, \ovc_\ell\}$. The double-lined edges represent deletion edges. An edge attached to a circle or ellipse represents a \emph{join}, i.e., a bundle of edges. The sets depict the allowed states at a vertex.}
 \label{fig:dtc_color_set}
\end{figure}

\begin{lem}
  \label{thm:decode_behavior}
  Let $U$ be a set consisting of at most $\ncols$ true twins and $v$ be a vertex that is not adjacent to $U$ and let $C \subsetneq [\ncols]$ with $|C| \leq |U|$. There is a graph $\decode = \decode_C(U,v)$ with the following properties:
  \begin{itemize}
   \item Any solution $\coloring$, also unnormalized ones, satisfies $|V(\decode - U) \cap \delset| \geq (\ncols - |C|) + 1$.
   \item If $\coloring$ is a normalized solution with $|V(\decode - U) \cap \delset| = (\ncols - |C|) + 1$, then $\coloring(v) = \del$ implies that $\coloring(U) \setminus \{\del\} \subseteq C$.
   \item Any normalized solution $\tilde{\coloring}$ for $\tilde{G} = G[(V \setminus V(\decode)) \cup U \cup \{v\}]$ with $\tilde{\coloring}(v) \neq \del$ or $\tilde{\coloring}(U) \setminus \{\del\} \subseteq C$ can be extended to a normalized solution $\coloring$ for $G$ with $|V(\decode - U) \cap \delset| = (\ncols - |C|) + 1$.
   \item $\tw(\decode - U) \leq \ncols$ via a winning cops-and-robber-strategy starting with a cop on $v$.
  \end{itemize}
\end{lem}

\begin{proof}
  Let $[\ncols] \setminus C = \{\ovc_1, \ldots, \ovc_\ell\}$. We add $2\ell + 1$ vertices $w_1, \ldots, w_{2\ell + 1}$ and add all edges between $w_{2i - 1}$, $i \in [\ell]$, and $U$, all edges between $w_{2i}$, $i \in [\ell]$, and $w_{2\ell + 1}$, a deletion edge between $w_{2\ell + 1}$ and $v$ and a deletion edge between each pair $w_{2i - 1}$ and $w_{2i}$ for $i \in [\ell]$. We also add all edges between $\{w_{2i} \sep i \in [\ell]\} \cup \{w_{2\ell + 1}\}$ and $F \setminus \{f_1\}$. Finally, for each $i \in [\ell]$, vertex $w_{2i - 1}$ is adjacent to all vertices in $F \setminus \{f_{\ovc_i}\}$. Hence, considering a normalized solution, the vertices in $\{w_{2i} \sep i \in [\ell]\} \cup \{w_{2\ell + 1}\}$ may only receive color $1$ or be deleted and vertex $w_{2i - 1}$, $i \in [\ell]$, may only receive color $\ovc_i$ or be deleted. See \cref{fig:dtc_color_set} for a depiction of the construction.
  
  The lower bound on the number of deletions for possibly unnormalized solutions follows by noticing that $\decode - U$ contains $\ell + 1 = (\ncols - |C|) + 1$ disjoint deletion edges.
  
  Suppose that $\coloring$ is a normalized solution with $|V(\decode - U) \cap \delset| = \ell + 1$ and $\coloring(U) \setminus \{\del\} \not\subseteq C$, then there exists some $\ovc_i \in ([\ncols] \setminus C) \cap \coloring(U)$. Due to the constraints enforced by $F$, we must have that $\coloring(w_{2i - 1}) = \del$ and $\coloring(w_{2i}) = 1$ by the bound on the number of deletions. This in turn forces $\coloring(w_{2\ell + 1}) = \del$ and hence $v$ cannot be deleted without violating the deletion bound.
  
  Suppose that $\tilde{\coloring}$ is a normalized solution of $\tilde{G}$ with $\tilde{\coloring}(v) \neq \del$, then we can extend to a normalized solution $\coloring$ of $G$ by deleting all vertices in $\{w_{2i - 1} \sep i \in [\ell]\} \cup \{w_{2\ell + 1}\}$. It is easy to see that the remainder can be colored correctly. 
  
  Suppose that $\tilde{\coloring}$ is a normalized solution of $\tilde{G}$ with $\tilde{\coloring}(U) \setminus \{\del\} \subseteq C$ and $\tilde{\coloring}(v) = \del$. Since $([\ncols] \setminus C) \cap \tilde{\coloring}(U) = \emptyset$, we can extend $\tilde{\coloring}$ to a normalized solution $\coloring$ of $G$ by setting $\coloring(w_{2i - 1}) = \ovc_i$, $\coloring(w_{2i}) = \del$ for all $i \in [\ell]$, and $\coloring(w_{2\ell + 1}) = 1$. It is easy to see that the remainder can be colored correctly.
  
  We argue that $\tw(\decode - U) \leq \ncols$ using the omniscient cops-and-robber-game. We begin by placing cops on $v$ and $w_{2\ell + 1}$. This splits the graph into $\ell + 1$ connected components, one per deletion edge. We place two cops on the endpoints of the deletion edge the robber escaped to. If $v$ or respectively $w_{2\ell + 1}$ is not an endpoint of the considered deletion edge, then we remove the cop from $v$ or respectively $w_{2\ell + 1}$. Finally, we place $\ncols - 1$ cops on the vertices inside the deletion edge and capture the robber. This proves the treewidth bound by \cref{thm:cops_robber}.
\end{proof}

In the case of $C = \emptyset$ one could also use a thick $\ncols$-arrow from \cref{thm:arrow_behavior} instead. We will only invoke \cref{thm:decode_behavior} with $|U| = \ncols$ in the dense setting of the \DTC lower bound; we use the case $|U| = 1$ for the sparse version of the lower bound. 

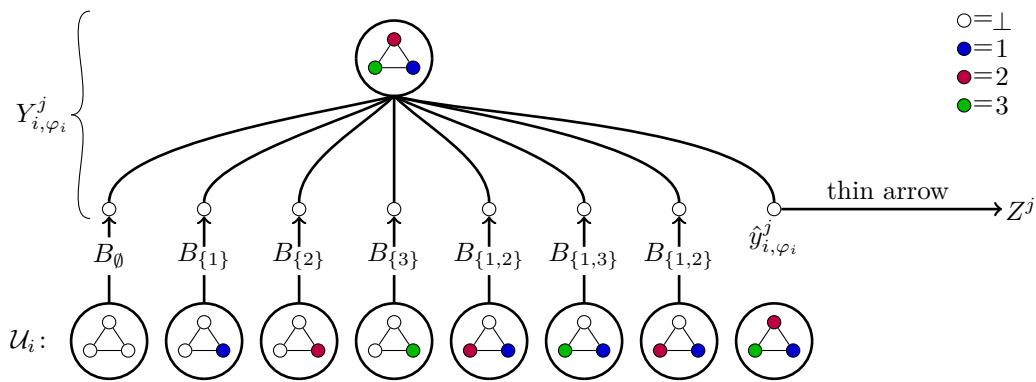
\begin{figure}[h]
 \centering
 \begin{tikzpicture}
	\begin{pgfonlayer}{nodelayer}
		\node [style=none] (0) at (3.5, -1.5) {};
		\node [style=none] (1) at (2.5, -2.5) {};
		\node [style=none] (2) at (1.5, -1.5) {};
		\node [style=none] (3) at (2.5, -0.5) {};
		\node [style=filled W] (4) at (2, -1.75) {};
		\node [style=filled B] (5) at (5.5, -1.75) {};
		\node [style=none] (6) at (6, -1.5) {};
		\node [style=none] (7) at (5, -2.5) {};
		\node [style=none] (8) at (4, -1.5) {};
		\node [style=none] (9) at (5, -0.5) {};
		\node [style=filled W] (10) at (4.5, -1.75) {};
		\node [style=none] (11) at (8.5, -1.5) {};
		\node [style=none] (12) at (7.5, -2.5) {};
		\node [style=none] (13) at (6.5, -1.5) {};
		\node [style=none] (14) at (7.5, -0.5) {};
		\node [style=filled W] (15) at (7, -1.75) {};
		\node [style=none] (16) at (11, -1.5) {};
		\node [style=none] (17) at (10, -2.5) {};
		\node [style=none] (18) at (9, -1.5) {};
		\node [style=none] (19) at (10, -0.5) {};
		\node [style=filled W] (20) at (9.5, -1.75) {};
		\node [style=filled B] (21) at (13, -1.75) {};
		\node [style=none] (22) at (13.5, -1.5) {};
		\node [style=none] (23) at (12.5, -2.5) {};
		\node [style=none] (24) at (11.5, -1.5) {};
		\node [style=none] (25) at (12.5, -0.5) {};
		\node [style=filled B] (26) at (15.5, -1.75) {};
		\node [style=none] (27) at (16, -1.5) {};
		\node [style=none] (28) at (15, -2.5) {};
		\node [style=none] (29) at (14, -1.5) {};
		\node [style=none] (30) at (15, -0.5) {};
		\node [style=none] (31) at (18.5, -1.5) {};
		\node [style=none] (32) at (17.5, -2.5) {};
		\node [style=none] (33) at (16.5, -1.5) {};
		\node [style=none] (34) at (17.5, -0.5) {};
		\node [style=filled R] (35) at (20, -1) {};
		\node [style=filled B] (36) at (20.5, -1.75) {};
		\node [style=none] (37) at (21, -1.5) {};
		\node [style=none] (38) at (20, -2.5) {};
		\node [style=none] (39) at (19, -1.5) {};
		\node [style=none] (40) at (20, -0.5) {};
		\node [style=filled W] (41) at (2.5, -1) {};
		\node [style=filled W] (42) at (3, -1.75) {};
		\node [style=filled W] (43) at (5, -1) {};
		\node [style=filled W] (44) at (7.5, -1) {};
		\node [style=filled W] (45) at (10, -1) {};
		\node [style=filled W] (46) at (12.5, -1) {};
		\node [style=filled W] (47) at (15, -1) {};
		\node [style=filled W] (48) at (17.5, -1) {};
		\node [style=filled R] (49) at (8, -1.75) {};
		\node [style=filled G] (50) at (10.5, -1.75) {};
		\node [style=filled R] (51) at (12, -1.75) {};
		\node [style=filled G] (52) at (14.5, -1.75) {};
		\node [style=filled R] (53) at (17, -1.75) {};
		\node [style=filled B] (54) at (18, -1.75) {};
		\node [style=filled G] (55) at (19.5, -1.75) {};
		\node [style=none] (67) at (0.5, -1.5) {$\mathcal{U}_i \colon$};
		\node [style=filled W] (76) at (25, 7) {};
		\node [style=filled B] (77) at (25, 6.25) {};
		\node [style=filled R] (78) at (25, 5.5) {};
		\node [style=filled G] (79) at (25, 4.75) {};
		\node [style=none] (80) at (25.5, 7) {=};
		\node [style=none] (81) at (25.5, 6.25) {=};
		\node [style=none] (82) at (25.5, 5.5) {=};
		\node [style=none] (83) at (25.5, 4.75) {=};
		\node [style=none] (84) at (26, 7) {$\bot$};
		\node [style=none] (85) at (26, 6.25) {1};
		\node [style=none] (86) at (26, 5.5) {2};
		\node [style=none] (87) at (26, 4.75) {3};
		\node [style=filled W] (88) at (2.5, 2) {};
		\node [style=filled W] (89) at (5, 2) {};
		\node [style=filled W] (90) at (7.5, 2) {};
		\node [style=filled W] (91) at (10, 2) {};
		\node [style=filled W] (92) at (12.5, 2) {};
		\node [style=filled W] (93) at (15, 2) {};
		\node [style=filled W] (94) at (17.5, 2) {};
		\node [style=filled R] (95) at (10, 6.5) {};
		\node [style=filled B] (96) at (10.5, 5.75) {};
		\node [style=none] (97) at (11, 6) {};
		\node [style=none] (98) at (10, 5) {};
		\node [style=none] (99) at (9, 6) {};
		\node [style=none] (100) at (10, 7) {};
		\node [style=filled G] (101) at (9.5, 5.75) {};
		\node [style=filled W] (102) at (20, 2) {};
		\node [style=none] (103) at (2.5, 0.75) {$B_\emptyset$};
		\node [style=none] (104) at (5, 0.75) {$B_{\{1\}}$};
		\node [style=none] (105) at (7.5, 0.75) {$B_{\{2\}}$};
		\node [style=none] (106) at (10, 0.75) {$B_{\{3\}}$};
		\node [style=none] (107) at (12.5, 0.75) {$B_{\{1,2\}}$};
		\node [style=none] (108) at (15, 0.75) {$B_{\{1,3\}}$};
		\node [style=none] (109) at (17.5, 0.75) {$B_{\{1,2\}}$};
		\node [style=none] (110) at (2.5, 0.25) {};
		\node [style=none] (111) at (5, 0.25) {};
		\node [style=none] (112) at (7.5, 0.25) {};
		\node [style=none] (113) at (10, 0.25) {};
		\node [style=none] (114) at (12.5, 0.25) {};
		\node [style=none] (115) at (15, 0.25) {};
		\node [style=none] (116) at (17.5, 0.25) {};
		\node [style=none] (117) at (2.5, 1.25) {};
		\node [style=none] (118) at (5, 1) {};
		\node [style=none] (119) at (7.5, 1) {};
		\node [style=none] (120) at (10, 1) {};
		\node [style=none] (121) at (12.5, 1) {};
		\node [style=none] (122) at (15, 1) {};
		\node [style=none] (123) at (17.5, 1) {};
		\node [style=none] (124) at (20, 1.25) {$\hat{y}^j_{i,\varphi_i}$};
		\node [style=none] (125) at (26, 2) {};
		\node [style=none] (126) at (23, 2.5) {thin arrow};
		\node [style=none] (127) at (26.5, 2) {$Z^j$};
		\node [style=none] (128) at (0.75, 4.5) {$Y^j_{i,\varphi_i}$};
		\node [style=none] (129) at (2, 1.75) {};
		\node [style=none] (130) at (2, 7.25) {};
		\node [style=none] (131) at (1.5, 4.5) {};
	\end{pgfonlayer}
	\begin{pgfonlayer}{edgelayer}
		\draw [style=thick, bend left=45] (2.center) to (3.center);
		\draw [style=thick, bend left=45] (3.center) to (0.center);
		\draw [style=thick, bend left=45] (0.center) to (1.center);
		\draw [style=thick, bend left=45] (1.center) to (2.center);
		\draw [style=thick, bend left=45] (8.center) to (9.center);
		\draw [style=thick, bend left=45] (9.center) to (6.center);
		\draw [style=thick, bend left=45] (6.center) to (7.center);
		\draw [style=thick, bend left=45] (7.center) to (8.center);
		\draw (10) to (5);
		\draw [style=thick, bend left=45] (13.center) to (14.center);
		\draw [style=thick, bend left=45] (14.center) to (11.center);
		\draw [style=thick, bend left=45] (11.center) to (12.center);
		\draw [style=thick, bend left=45] (12.center) to (13.center);
		\draw [style=thick, bend left=45] (18.center) to (19.center);
		\draw [style=thick, bend left=45] (19.center) to (16.center);
		\draw [style=thick, bend left=45] (16.center) to (17.center);
		\draw [style=thick, bend left=45] (17.center) to (18.center);
		\draw [style=thick, bend left=45] (24.center) to (25.center);
		\draw [style=thick, bend left=45] (25.center) to (22.center);
		\draw [style=thick, bend left=45] (22.center) to (23.center);
		\draw [style=thick, bend left=45] (23.center) to (24.center);
		\draw [style=thick, bend left=45] (29.center) to (30.center);
		\draw [style=thick, bend left=45] (30.center) to (27.center);
		\draw [style=thick, bend left=45] (27.center) to (28.center);
		\draw [style=thick, bend left=45] (28.center) to (29.center);
		\draw [style=thick, bend left=45] (33.center) to (34.center);
		\draw [style=thick, bend left=45] (34.center) to (31.center);
		\draw [style=thick, bend left=45] (31.center) to (32.center);
		\draw [style=thick, bend left=45] (32.center) to (33.center);
		\draw (36) to (35);
		\draw [style=thick, bend left=45] (39.center) to (40.center);
		\draw [style=thick, bend left=45] (40.center) to (37.center);
		\draw [style=thick, bend left=45] (37.center) to (38.center);
		\draw [style=thick, bend left=45] (38.center) to (39.center);
		\draw (42) to (4);
		\draw (4) to (41);
		\draw (41) to (42);
		\draw (43) to (10);
		\draw (43) to (5);
		\draw (44) to (15);
		\draw (45) to (20);
		\draw (20) to (50);
		\draw (50) to (45);
		\draw (44) to (49);
		\draw (49) to (15);
		\draw (51) to (46);
		\draw (46) to (21);
		\draw (21) to (51);
		\draw (52) to (47);
		\draw (47) to (26);
		\draw (26) to (52);
		\draw (53) to (48);
		\draw (48) to (54);
		\draw (54) to (53);
		\draw (55) to (36);
		\draw (55) to (35);
		\draw (96) to (95);
		\draw [style=thick, bend left=45] (99.center) to (100.center);
		\draw [style=thick, bend left=45] (100.center) to (97.center);
		\draw [style=thick, bend left=45] (97.center) to (98.center);
		\draw [style=thick, bend left=45] (98.center) to (99.center);
		\draw (101) to (96);
		\draw (101) to (95);
		\draw [style=thick] (98.center) to (91);
		\draw [style=thick, in=90, out=-135, looseness=0.75] (98.center) to (90);
		\draw [style=thick, in=90, out=-45, looseness=0.75] (98.center) to (92);
		\draw [style=thick, in=90, out=-30, looseness=0.50] (98.center) to (93);
		\draw [style=thick, in=90, out=-15, looseness=0.50] (98.center) to (94);
		\draw [style=thick, in=90, out=-150, looseness=0.50] (98.center) to (89);
		\draw [style=thick, in=195, out=90, looseness=0.50] (88) to (98.center);
		\draw [style=thick] (3.center) to (110.center);
		\draw [style=thick] (9.center) to (111.center);
		\draw [style=thick] (14.center) to (112.center);
		\draw [style=thick] (19.center) to (113.center);
		\draw [style=thick] (25.center) to (114.center);
		\draw [style=thick] (30.center) to (115.center);
		\draw [style=thick] (34.center) to (116.center);
		\draw [style=thick arrow] (117.center) to (88);
		\draw [style=thick arrow] (118.center) to (89);
		\draw [style=thick arrow] (119.center) to (90);
		\draw [style=thick arrow] (120.center) to (91);
		\draw [style=thick arrow] (121.center) to (92);
		\draw [style=thick arrow] (122.center) to (93);
		\draw [style=thick arrow] (123.center) to (94);
		\draw [style=thick, in=90, out=-15, looseness=0.50] (98.center) to (102);
		\draw [style=thick arrow] (102) to (125.center);
		\draw [in=-180, out=0, looseness=0.75] (131.center) to (130.center);
		\draw [in=-180, out=0, looseness=0.75] (131.center) to (129.center);
	\end{pgfonlayer}
\end{tikzpicture}
 \caption{The decoding gadget $Y^j_{i, \coloring_i}$ and its connections to the central twinclasses and the clause gadget for the case $\ncols = 3$ and $\grpsize = 8$. The large circles represent twinclasses and the arrows represent color-set-gadgets or thin arrows. The solution chosen on the central twinclasses is $\coloring_i$. Note that it is allowed that $\coloring_i$ uses the same color set on several twinclasses in $\ttcset_i$.}
 \label{fig:dtc_decode}
\end{figure}

\subparagraph*{Decoding gadgets.} The color-set-gadgets allow us to construct \emph{decoding gadgets} that can distinguish between the different partial solutions in $\scolorings_i$. While the color-set-gadgets only check for inclusion and not equality, this is nonetheless sufficient to distinguish solutions in $\scolorings_i$ due to their structure.

We begin with the construction now. For the $j$-th clause, group $i \in [\ngrps]$, solution $\coloring_i \in \scolorings_i$, we construct a gadget $Y^j_{i, \coloring_i}$ as follows. The gadget $Y^j_{i, \coloring_i}$ consists of a large independent set joined to a $K_\ncols$, i.e., adding all edges between both sets. In other words, $Y^j_{i, \coloring_i}$ is a complete $(\ncols + 1)$-partite graph with one large independent set and all other independent sets in the partition are singletons. More precisely, the large independent set consists of $(1 - 2^{-\ncols})\grpsize + 1$ vertices (recall that $2^\ncols$ divides $\grpsize$). One of these vertices is distinguished and denoted by $\hat{y}^j_{i, \coloring_i}$. This concludes the construction, see \cref{fig:dtc_decode} for a depiction of the construction.

In particular, the $K_\ncols$ and $\hat{y}^j_{i, \coloring_i}$ induce a complete graph of size $\ncols + 1$. Hence, any solution must delete at least one vertex in this complete graph and due to the budget constraint exactly one vertex has to be deleted. The distinguished vertex $\hat{y}^j_{i, \coloring_i}$ can only be deleted if also all other vertices in the large independent set are deleted. By appropriately adding color-set-gadgets, we will ensure that this can only be achieved if $\coloring_i$ is chosen on $\ttcset_i$.

For the $j$-th clause, group $i \in [\ngrps]$, solution $\coloring_i \in \scolorings_i$, twinclass $U \in \ttcset_i$ with $\coloring_i(U) \neq [\ncols]$, we pick a private vertex $v \neq \hat{y}^j_{i, \coloring_i}$ in the large independent set of $Y^j_{i, \coloring_i}$ and add the color-set-gadget $\decode_{\coloring_i(U) \setminus \{\del\}}(U,v)$ and denote this instance of the color-set-gadget by $W^j_{i, \coloring_i, U}$. We will see, using the properties of solutions in $\scolorings_i$, that if the solution on $\ttcset_i$ diverges from $\coloring_i$, then at least one $W^j_{i,\coloring_i,v}$ will be passive. Otherwise, all $W^j_{i,\coloring_i,v}$ will be active, allowing us to delete the vertex $\hat{y}^j_{i, \coloring_i}$ in $Y^j_{i,\coloring_i}$.

\subparagraph{Clause gadgets.} For the $j$-th clause, we add an $(\ncols + 1)$-critical graph, denoted $Z^j \in \crit^\ncols$, consisting of at least $\clss 2^{\vgrpsize}$ vertices. For every group $i \in [\ngrps]$ and solution $\coloring_i \in \scolorings_i$ such that $\embedding_i^{-1}(\coloring_i)$ is a partial truth assignment satisfying the $j$-th clause, we pick a private vertex $v$ in $Z^j$ and add a thin arrow from $\hat{y}^j_{i, \coloring_i}$ to $v$. Since every clause consists of at most $\clss$ literals, at most $\clss$ groups can contain a literal of the $j$-th clause and every such group has at most $2^{\vgrpsize}$ assignments that satisfy the $j$-th clause, hence $Z^j$ contains enough vertices so that we can always pick a private one. This concludes the construction of the clause gadget.

The idea of the clause gadget is that we can only afford to delete a vertex in $Z^j$ if we have picked a partial solution on some $\ttcset_i$ that corresponds to a satisfying truth assignment of the $j$-th clause, which will propagate a deletion to $Z^j$ via the decoding gadgets.

\begin{figure}[h]
  \centering
  \begin{tikzpicture}
	\begin{pgfonlayer}{nodelayer}
		\node [style=none] (0) at (-11, 5) {$\mathcal{U}_i$};
		\node [style=none] (1) at (-1, 5) {$\mathcal{U}_{i'}$};
		\node [style=none] (2) at (-13.25, 7.75) {$L_{i,1,S_{i,1}}$};
		\node [style=none] (3) at (-13.25, 6.75) {$L_{i,1,S'_{i,1}}$};
		\node [style=none] (4) at (-8.75, 7.75) {$L_{i,2,S_{i,2}}$};
		\node [style=none] (5) at (-8.75, 6.75) {$L_{i,2,S'_{i,2}}$};
		\node [style=none] (6) at (-13.25, 6) {$\vdots$};
		\node [style=none] (7) at (-8.75, 6) {$\vdots$};
		\node [style=none] (8) at (-3.25, 7.75) {$L_{i',1,S_{i',1}}$};
		\node [style=none] (9) at (-3.25, 6.75) {$L_{i',1,S'_{i',1}}$};
		\node [style=none] (10) at (-3.25, 6) {$\vdots$};
		\node [style=none] (11) at (1.5, 7.75) {$L_{i',2,S_{i',2}}$};
		\node [style=none] (12) at (1.5, 6.75) {$L_{i',2,S'_{i',2}}$};
		\node [style=none] (13) at (1.5, 6) {$\vdots$};
		\node [style=none] (15) at (-12, 7.75) {};
		\node [style=none] (16) at (-12, 6.75) {};
		\node [style=none] (18) at (-11, 5.5) {};
		\node [style=none] (19) at (-10, 7.75) {};
		\node [style=none] (20) at (-10, 6.75) {};
		\node [style=none] (21) at (-2, 7.75) {};
		\node [style=none] (22) at (-2, 6.75) {};
		\node [style=none] (23) at (-1, 5.5) {};
		\node [style=none] (24) at (0, 7.75) {};
		\node [style=none] (25) at (0, 6.75) {};
		\node [style=none] (26) at (-11.75, 8.75) {$A_1$};
		\node [style=none] (27) at (-10.25, 8.75) {$A_2$};
		\node [style=none] (28) at (-14, 1.25) {$Y_{i,\varphi_i}^j$};
		\node [style=none] (29) at (-12, 1.25) {$Y_{i,\psi_i}^j$};
		\node [style=none] (30) at (-10, 1.25) {$Y_{i,\varphi_i}^{j'}$};
		\node [style=none] (31) at (-8, 1.25) {$Y_{i,\psi_i}^{j'}$};
		\node [style=none] (32) at (-14, 2) {};
		\node [style=none] (33) at (-12, 2) {};
		\node [style=none] (34) at (-10, 2) {};
		\node [style=none] (35) at (-8, 2) {};
		\node [style=none] (36) at (-11, 4.5) {};
		\node [style=none] (37) at (-4, 1.25) {$Y_{i',\varphi_{i'}}^j$};
		\node [style=none] (38) at (-2, 1.25) {$Y_{i',\psi_{i'}}^j$};
		\node [style=none] (39) at (0, 1.25) {$Y_{i',\varphi_{i'}}^{j'}$};
		\node [style=none] (40) at (2, 1.25) {$Y_{i',\psi_{i'}}^{j'}$};
		\node [style=none] (41) at (-4, 2) {};
		\node [style=none] (42) at (-2, 2) {};
		\node [style=none] (43) at (0, 2) {};
		\node [style=none] (44) at (2, 2) {};
		\node [style=none] (45) at (-1, 4.5) {};
		\node [style=none] (46) at (-19, 3.75) {color-set-gadgets $B_C\colon$};
		\node [style=none] (47) at (-18.25, 8.75) {thick arrows $A_\ell\colon$};
		\node [style=none] (48) at (-1.75, 8.75) {$A_1$};
		\node [style=none] (49) at (-0.25, 8.75) {$A_2$};
		\node [style=none] (50) at (-11, -2) {$Z^j$};
		\node [style=none] (51) at (-1, -2) {$Z^{j'}$};
		\node [style=none] (52) at (-14, 0.5) {};
		\node [style=none] (53) at (-12, 0.5) {};
		\node [style=none] (54) at (-10, 0.5) {};
		\node [style=none] (55) at (-8, 0.5) {};
		\node [style=none] (56) at (-4, 0.5) {};
		\node [style=none] (57) at (-2, 0.5) {};
		\node [style=none] (58) at (0, 0.5) {};
		\node [style=none] (59) at (2, 0.5) {};
		\node [style=none] (60) at (-11.25, -1.5) {};
		\node [style=none] (61) at (-11.5, -2) {};
		\node [style=none] (62) at (-10.5, -1.75) {};
		\node [style=none] (63) at (-10.5, -2.25) {};
		\node [style=none] (64) at (-1.75, -1.75) {};
		\node [style=none] (65) at (-1.75, -2.25) {};
		\node [style=none] (66) at (-0.25, -2) {};
		\node [style=none] (67) at (-0.75, -1.5) {};
		\node [style=none] (68) at (-17.75, -0.25) {thin arrows:};
		\node [style=none] (72) at (-18.5, 7) {structure gadgets:};
		\node [style=none] (73) at (-18, -2) {clause gadgets:};
		\node [style=none] (74) at (-19, 3) {(attached to clique $F$)};
		\node [style=none] (75) at (-18.75, 5) {central twinclasses:};
		\node [style=none] (76) at (-18.5, 1.25) {decoding gadgets:};
		\node [style=none] (77) at (-15.5, 7) {};
		\node [style=none] (78) at (-15, 8) {};
		\node [style=none] (79) at (-15, 6) {};
	\end{pgfonlayer}
	\begin{pgfonlayer}{edgelayer}
		\draw [style=thick arrow, in=180, out=60, looseness=0.75] (18.center) to (19.center);
		\draw [style=thick arrow, in=180, out=45] (18.center) to (20.center);
		\draw [style=thick arrow, in=0, out=120, looseness=0.75] (18.center) to (15.center);
		\draw [style=thick arrow, in=0, out=135] (18.center) to (16.center);
		\draw [style=thick arrow, in=180, out=60, looseness=0.75] (23.center) to (24.center);
		\draw [style=thick arrow, in=-180, out=45] (23.center) to (25.center);
		\draw [style=thick arrow, in=0, out=120, looseness=0.75] (23.center) to (21.center);
		\draw [style=thick arrow, in=0, out=135] (23.center) to (22.center);
		\draw [style=thick arrow, in=90, out=-150, looseness=0.75] (36.center) to (32.center);
		\draw [style=thick arrow, in=90, out=-120] (36.center) to (33.center);
		\draw [style=thick arrow, in=90, out=-60] (36.center) to (34.center);
		\draw [style=thick arrow, in=90, out=-30, looseness=0.75] (36.center) to (35.center);
		\draw [style=thick arrow, in=90, out=-150, looseness=0.75] (45.center) to (41.center);
		\draw [style=thick arrow, in=90, out=-120] (45.center) to (42.center);
		\draw [style=thick arrow, in=90, out=-60] (45.center) to (43.center);
		\draw [style=thick arrow, in=90, out=-30, looseness=0.75] (45.center) to (44.center);
		\draw [style=thick arrow, in=180, out=-90, looseness=0.75] (52.center) to (61.center);
		\draw [style=thick arrow, in=90, out=-90] (53.center) to (60.center);
		\draw [style=thick arrow, in=0, out=-90, looseness=0.50] (56.center) to (62.center);
		\draw [style=thick arrow, in=0, out=-90, looseness=0.50] (57.center) to (63.center);
		\draw [style=thick arrow, in=180, out=-90, looseness=0.50] (54.center) to (65.center);
		\draw [style=thick arrow, in=180, out=-90, looseness=0.50] (55.center) to (64.center);
		\draw [style=thick arrow, in=90, out=-90] (58.center) to (67.center);
		\draw [style=thick arrow, in=0, out=-90, looseness=0.75] (59.center) to (66.center);
		\draw [in=-180, out=0, looseness=1.75] (77.center) to (78.center);
		\draw [in=180, out=0, looseness=1.75] (77.center) to (79.center);
	\end{pgfonlayer}
\end{tikzpicture}
  \caption{An overview of the construction of the graph $G(\formula, \ncols, \vgrpsize)$ for the case $\ncols = 2$.}
  \label{fig:dtc_overview_dense}
\end{figure}

\subparagraph*{Packing.} We will now construct a packing $\packing$ of vertex-disjoint graphs that will fully explain the budget outside of the central vertices. For every thin arrow from $u$ to $v$ in the construction, we add the $K_{\ncols + 1}$ induced by the deletion edge incident to $v$ to the packing $\packing$. For every thick $\ell$-arrow $\arrow_\ell(U,v)$ in the construction, we add $\arrow_\ell(U,v) - U$ to the packing $\packing$ and by \cref{thm:arrow_behavior} these graphs require at least $\ell$ deletions each. For every color-set-gadget $\decode_{C}(U,v)$ in the construction, we add $\decode_{C}(U,v) - U$ to the packing $\packing$ and by \cref{thm:decode_behavior} these graphs require at least $(\ncols - |C|) + 1$ deletions each. Finally, for every $Y^j_{i, \coloring_i}$, we add the $K_{\ncols + 1}$ induced by $\hat{y}^j_{i, \coloring_i}$ and the $K_\ncols$ of $Y^j_{i, \coloring_i}$ to the packing $\packing$. Let $\cost_\packing$ denote the cost of the packing $\packing$.

Observe that no vertex will be the head of several arrows or color-set-gadgets in our construction, hence the graphs in $\packing$ are indeed vertex-disjoint. Furthermore, no graph in $\packing$ intersects the central vertices and the cost of $\packing$ is independent of the partial solution chosen on the central vertices. Finally, notice that the cost of $\decode_{C}(U,v) - U$ is fully explained by the $(\ncols - |C|) + 1$ disjoint deletion edges and hence does not rely on any \LC constraints simulated by the central clique $F$.

This concludes the construction of $G = G(\formula, \ncols, \vgrpsize)$, cf.\ \cref{fig:dtc_overview_dense}. We proceed by showing the correctness of the reduction. From now on, we might omit the range of the indices for the sake of readability, but we keep the meaning of the indices consistent throughout. By $i$ we denote a group of variables, $\coloring_i$ denotes a partial solution in $\scolorings_i$, and $U$ denotes a true twinclass in $\ttcset_i$ with $\coloring_i(U) \neq [\ncols]$.

\begin{mythm}
 \label{thm:sat_to_dtc}
 Let $\formula$ be a $\clss$-\SAT instance with $\nvars$ variables and $\nclss$ clauses. Let $G = G(\formula, \ncols, \vgrpsize)$ and $\packing$ be the graph and packing as constructed above and let $\budget = \cost_\packing + \ngrps \ncols \grpsize / 2$. If $\formula$ has a satisfying assignment $\tassign$, then there is a solution $\coloring$ of the \DTC instance $(G, \budget)$.
\end{mythm}

\begin{proof}
  We start with the construction of $\coloring$ on the central vertices. We color the central clique $F$ with $\coloring(f_s) = s$ for $s \in [\ncols]$, i.e., $\coloring$ will be a normalized solution. For each group $i$ of variables, let $\tassign_i \in \{0,1\}^\vgrpsize$ be the partial truth assignment on group $i$ induced by $\tassign$. For each group $i$, we set $\coloring\big|_{\bigcup(\ttcset_i)} = \coloring_i := \embedding_i(\tassign_i) \in \scolorings_i$. By \cref{thm:scoloring_deletions}, this results in exactly $\ngrps \ncols \grpsize / 2$ deletions. Hence, only the budget $\cost_\packing$ for the packing $\packing$ remains, so for every graph in $\packing$ we have to match the lower bound on the number of deletions for this graph. 
  
  Whenever we delete the tail $u$ or head $v$ of a thin arrow from $u$ to $v$, then we use the active solution on this arrow; otherwise, we use the passive solution. Either type of solution results in exactly one deletion in the deletion edge incident to $v$ which belongs to $\packing$. 
  
  Similarly, for any thick $\ell$-arrow $\arrow_\ell(U,v)$, we extend to the solution with $|V(\arrow_\ell(U,v) - U) \cap \delset| = \ell$ as given by \cref{thm:arrow_behavior}. If $|U \cap \delset| \geq \ell$, then this solution is active and satisfies $\coloring(v) = \del$. 
  
  For every color-set-gadget $\decode_C(U,v)$, we extend to the solution with $|V(\decode_C(U,v) - U) \cap \delset| = (\ncols - |C|) + 1$ as given by \cref{thm:decode_behavior}. If $\coloring(U) \setminus \{\del\} \subseteq C$, then this solution is active and satisfies $\coloring(v) = \del$. 
  
  For every decoding gadget $Y^j_{i, \coloring_i}$, we delete the distinguished vertex $\hat{y}^j_{i,\coloring_i}$. For every $j$-th clause, group $i$, and solution $\psi_i \in \scolorings_i \setminus \{\coloring_i\}$, we pick one of the vertices in the $K_\ncols$ of $Y^j_{i, \psi_i}$ and delete it. This deletes exactly one vertex in the $K_{\ncols + 1}$ formed by $\hat{y}^j_{i, \coloring_i}$ or $\hat{y}^j_{i, \psi_i}$ and the $K_\ncols$. 
  
  This concludes the description of the deletions; for each graph in the packing $\packing$, we match the lower bound on the number of deletions for this graph. It remains to show that the remainder of $G$ can be properly $\ncols$-colored. This can be easily seen for the thin arrows. For thick $\ell$-arrows $\arrow_\ell(U,v)$ it follows from \cref{thm:arrow_behavior}. For the color-set-gadgets $W^j_{i, \psi_i, U}$, where $\psi_i \in \scolorings_i$, it follows from \cref{thm:decode_behavior}. It remains to handle the structure gadgets $L_{i,\ell,S}$, the decoding gadgets $Y^j_{i, \psi_i}$, $\psi_i \in \scolorings_i$, and the clause gadgets $Z^j$. 
  
  By \cref{thm:arrow_behavior} and \cref{thm:decode_behavior}, we can appropriately extend solutions to $\arrow_\ell(U,v)$ and $\decode_C(U,v)$ for every color choice on the vertex $v$. Hence, when handling the aforementioned gadgets, it does not matter how we color the undeleted vertices. It is sufficient to verify that enough deletions occur and that the remainder of the gadget can be $\ncols$-colored. 
  
  Consider some group $i$, some $\ell \in [\ncols]$, some $\setfam \subseteq \ttcset_i$ with $|\setfam| = \binom{\ncols}{\geq (\ncols - \ell + 1)}\frac{\grpsize}{2^{\ncols}} + 1$, and some copy of $L_{i, \ell, \setfam}$. Since $L_{i, \ell, \setfam}$ is an $(\ncols + 1)$-critical graph that does not belong to $\packing$, we must perform at least one deletion in $L_{i, \ell, \setfam}$ and this deletion has to be propagated to $L_{i, \ell, \setfam}$ by some thick $\ell$-arrow $\arrow_\ell(U,v)$, $U \in \ttcset_i$. Since $\coloring_i \in \scolorings_i$, we have that $\left|\left\{U \in \ttcset_i \sep \coloring_i(U) \setminus \{\del\} \in \binom{[\ncols]}{\leq (\ncols - \ell)}\right\}\right| = \binom{\ncols}{\leq (\ncols - \ell)} \frac{\grpsize}{2^\ncols}$, i.e., there are exactly $\binom{\ncols}{\leq (\ncols - \ell)} \frac{\grpsize}{2^\ncols}$ twinclasses in $\ttcset_i$ where $\coloring_i$ deletes at least $\ell$ vertices in each. Due to $\binom{\ncols}{\leq (\ncols - \ell)} \frac{\grpsize}{2^\ncols} + |\setfam| = \grpsize + 1 > \grpsize = |\ttcset_i|$, there must be at least one twinclass $U \in \setfam$ with $\coloring_i(U) \setminus \{\del\} \in \binom{[\ncols]}{\leq (\ncols - \ell)}$, i.e., where $\coloring_i$ deletes at least $\ell$ vertices. Therefore, the thick $\ell$-arrow $\arrow_\ell(U,v)$ leading to this copy of $L_{i,\ell,\setfam}$ propagates a deletion to $L_{i,\ell,\setfam}$ by \cref{thm:arrow_behavior} and thereby resolves this $(\ncols + 1)$-critical graph. 
  
  Consider some $Y^j_{i, \psi_i}$, where $\psi_i \in \scolorings_i$, if $\psi_i \neq \coloring_i$, then one of the vertices in the $K_\ncols$ is deleted and this $(\ncols + 1)$-partite graph is resolved. If $\psi_i = \coloring_i$, then we claim that the large independent set is fully deleted. In the construction of $\coloring$, we distributed a deletion to $\hat{y}^j_{i, \coloring_i}$. All other vertices $v$ of the large independent set are hit by some $W^j_{i, \coloring_i, U}$, where $\coloring_i(U) \setminus \{\del\} \neq [\ncols]$. By \cref{thm:decode_behavior}, we see that $W^j_{i, \coloring_i, U}  = \decode_{\coloring_i(U) \setminus \{\del\}}(U,v)$ propagates a deletion to $v$ due to \cref{thm:decode_behavior}. Due to $\coloring_i \in \scolorings_i$, there are $\binom{\ncols}{\leq(\ncols - 1)} \frac{\grpsize}{2^\ncols} = (1 - 2^{-\ncols})\grpsize$ twinclasses $U \in \ttcset_i$ with $\coloring_i(U) \setminus \{\del\} \neq [\ncols]$, thus matching the size of the large independent set of $Y^j_{i, \coloring_i}$ with the exception of $\hat{y}^j_{i, \coloring_i}$. This shows that the large independent set is fully deleted and the $(\ncols + 1)$-partite graph $Y^j_{i, \coloring_i}$ is resolved.
  
  Finally, consider some $Z^j$, then there is some group $i$ such that $\tassign_i$ satisfies the $j$-th clause because $\tassign$ is a satisfying assignment of $\formula$. By construction $\coloring_i = \embedding_i(\tassign_i)$ and the vertex $\hat{y}^j_{i, \coloring_i}$ is deleted, so the thin arrow from $\hat{y}^j_{i, \coloring_i}$ to $Z^j$ is active and propagates a deletion to $Z^j$. Therefore, the $(\ncols + 1)$-critical graph $Z^j$ is resolved as well.
\end{proof}

\begin{mythm}
  \label{thm:enforced_structure}
  Let $\formula$ be a $\clss$-\SAT instance with $\nvars$ variables and $\nclss$ clauses. Let $G = G(\formula, \ncols, \vgrpsize)$ be the graph as constructed above and let $\budget = \cost_\packing + \ngrps \ncols \grpsize / 2$. If $\coloring$ is a solution of the \DTC instance $(G,\budget)$, then $|\delset| = \budget$. Furthermore, there is a normalized solution $\psi$ with $|\psi^{-1}(\del)| = |\delset| = \budget$ and $\psi\big|_{\bigcup(\ttcset_i)} \in \scolorings_i$ for all $i \in [\ngrps]$.
\end{mythm}

\begin{proof}
  As argued in the construction of $\packing$, the packing $\packing$ forces $\coloring$ to spend at least $\cost_\packing$ units of budget outside of the central vertices. We claim that for any group $i \in [\ngrps]$ and number of deletions $\ell \in [\ncols]$, there must be at least $\binom{\ncols}{\leq (\ncols - \ell)} \frac{\grpsize}{2^\ncols}$ twinclasses $U \in \ttcset_i$ such that $|U \cap \delset| \geq \ell$, i.e., $\coloring$ deletes at least $\ell$ vertices in each of those $U$. We first define $\setfam_i^\ell(\coloring) = \{U \in \ttcset_i \sep |U \cap \delset| = \ell\}$ and $\setfam_i^{\geq \ell}(\coloring) = \setfam_i^\ell(\coloring) \cup \setfam_i^{\ell + 1}(\coloring) \cup \cdots \cup \setfam_i^\ncols(\coloring)$ for all $i \in [\ngrps]$ and $\ell \in \{0\} \cup [\ncols]$.
  
  Suppose that there is some group $i$ and some number of deletions $\ell \in [\ncols]$ with $|\setfam_i^{\geq \ell}(\coloring)| < \binom{\ncols}{\leq (\ncols - \ell)} \frac{\grpsize}{2^\ncols}$. Then there exists a family $\setfam \subseteq \ttcset_i$ of twinclasses with $|\setfam| = \binom{\ncols}{\geq (\ncols - \ell + 1)} \frac{\grpsize}{2^\ncols} + 1$ and $\setfam_i^{\geq \ell}(\coloring) \cap \setfam = \emptyset$. Consider any copy of $L_{i, \ell, \setfam}$ and the thick $\ell$-arrows $\arrow_\ell(U,v)$ from $\setfam \subseteq \ttcset_i$ to $L_{i, \ell, \setfam}$. Due to $\setfam_i^{\geq \ell}(\coloring) \cap \setfam = \emptyset$, the solution $\coloring$ deletes at most $\ell - 1$ vertices in each $U \in \setfam$, and by \cref{thm:arrow_behavior} this implies that no vertex in the copy of $L_{i, \ell, \setfam}$ is deleted, unless we pay for at least one extra deletion not accounted for by the packing $\packing$ per copy of $L_{i, \ell, \setfam}$. However, since $L_{i, \ell, \setfam}$ is an $(\ncols + 1)$-critical graph, we must delete at least one vertex in each copy. This would exceed our remaining budget, because there are $1 + \ngrps \ncols \grpsize / 2 > \budget - \cost_\packing$ copies of $L_{i, \ell, \setfam}$. Hence, we must have that $|\setfam_i^{\geq \ell}(\coloring)| \geq \binom{\ncols}{\leq(\ncols - \ell)} \frac{\grpsize}{2^\ncols}$ for all groups $i \in [\ngrps]$ and deletions $\ell \in [\ncols]$. 
  
  We will now argue that we must have $|\setfam_i^\ell(\coloring)| = \binom{\ncols}{\ncols - \ell} \frac{\grpsize}{2^\ncols}$ for all $i \in [\ngrps]$ and $\ell \in [\ncols]$ due to our budget constraint. Consider some group $i$, using the previous inequalities we can derive the following bound on the total number of deletions performed on $\ttcset_i$ by $\coloring$:
  \begin{align*}
    \sum_{\ell = 1}^\ncols \ell |\setfam_i^\ell(\coloring)| & = \ncols |\setfam_i^{\ncols}(\coloring)| + \sum_{\ell = 1}^{\ncols - 1} \ell (|\setfam_i^{\geq \ell}(\coloring)| - |\setfam_i^{\geq (\ell + 1)}(\coloring)|) = \sum_{\ell = 1}^{\ncols} |\setfam_i^{\geq \ell}(\coloring)| \\
    & \geq \sum_{\ell = 1}^\ncols \binom{\ncols}{\leq(\ncols - \ell)} \frac{\grpsize}{2^\ncols} = \sum_{\ell = 0}^{\ncols - 1} (\ncols - \ell) \binom{\ncols}{\ell} \frac{\grpsize}{2^\ncols} = \ncols \grpsize / 2, 
  \end{align*}
  where the last inequality follows from the computation in the proof of \cref{thm:scoloring_deletions}. Summing over the lower bound for all groups, we see that this uses up the whole remaining budget and hence the inequality must be tight for every group. Therefore, also $|\setfam_i^{\geq \ell}(\coloring)| = \binom{\ncols}{\leq(\ncols - \ell)} \frac{\grpsize}{2^\ncols}$ and $|\setfam_i^{\ell}(\coloring)| = |\setfam_i^{\geq \ell}(\coloring)| - |\setfam_i^{\geq (\ell + 1)}(\coloring)| = \binom{\ncols}{\ncols - \ell} \frac{\grpsize}{2^\ncols}$ for all $i \in [\ngrps]$ and $\ell \in [\ncols]$. Since $\sum_{\ell = 0}^\ncols |\setfam_i^{\ell}(\coloring)| = \grpsize$, this also implies that $|\setfam_i^0(\coloring)| = \binom{\ncols}{\ncols} \frac{\grpsize}{2^\ncols} = \frac{\grpsize}{2^\ncols}$ for all groups $i$.
  
  No vertex of the central clique $F$ can be deleted, because $\coloring$ spends its whole budget on the packing $\packing$ and the families of twinclasses $\ttcset_i$, $i \in [\ngrps]$. Hence, we can assume by permuting the colors that $\coloring(f_s)  = s$, $s \in [\ncols]$, i.e., $\coloring$ is a normalized solution. 
  
  Finally, the equations $|\setfam_i^{\ell}(\coloring)| = \binom{\ncols}{\ncols - \ell} \frac{\grpsize}{2^\ncols} = \binom{\ncols}{\ell} \frac{\grpsize}{2^\ncols}$, $i \in [\ngrps]$, $\ell \in \{0\} \cup [\ncols]$, imply that for every group $i$ there exists some $\psi_i \in \scolorings_i$ with $\coloring(U) = \psi_i(U)$ for all $U \in \ttcset_i$. Since all $U \in \ttcset_i$ are twinclasses, exchanging $\coloring\big|_{\bigcup(\ttcset_i)}$ with $\psi_i$ for every group $i$ still yields a solution $\psi$ for \DTC with the desired properties. 
\end{proof}

\begin{mythm}
 \label{thm:dtc_to_sat}
 Let $\formula$ be a $\clss$-\SAT instance with $\nvars$ variables and $\nclss$ clauses. Let $G = G(\formula, \ncols, \vgrpsize)$ be the graph as constructed above and let $\budget = \cost_\packing + \ngrps \ncols \grpsize / 2$. If the \DTC instance $(G, \budget)$ has a solution $\coloring$, then $\formula$ has a satisfying assignment $\tassign$.
\end{mythm}

\begin{proof}
 By \cref{thm:enforced_structure}, we know that $|\delset| = \budget$ and can assume that $\coloring$ is a normalized solution with $\coloring_i := \coloring\big|_{\bigcup(\ttcset_i)} \in \scolorings_i$ for all groups $i$. Hence, the deletions are fully explained by the packing $\packing$ and the properties of solutions in $\scolorings_i$. 
 
 For every thin arrow from $u$ to $v$, thick $\ell$-arrow $\arrow_\ell(U,v)$, color-set-gadget $\decode_C(U,v)$, by \cref{thm:arrow_behavior} and \cref{thm:decode_behavior}, we can assume that $\coloring$ chooses the active solution on this gadget if the appropriate condition at $u$ or $U$ is satisfied. 
 
 Consider the $j$-th clause and some group $i$, we begin by arguing that $\coloring(\hat{y}^j_{i, \psi_i}) = \del$ only if $\psi_i = \coloring_i$. First, observe that $\hat{y}^j_{i, \psi_i}$ may only be deleted if all other vertices of the large independent set in $Y^j_{i, \psi_i}$ are deleted, otherwise $Y^j_{i, \psi_i}$ as a complete $(\ncols + 1)$-partite graph is not resolved or we exceed the packing budget on the complete graph induced by $\hat{y}^j_{i, \psi_i}$ and the $K_\ncols$. Consider some $\psi_i \in \scolorings_i \setminus \{\coloring_i\}$. We will distinguish three cases.
 \begin{enumerate} 
  \item Suppose that there is some $U \in \ttcset_i$ such that $|\psi_i(U) \setminus \{\del\}| < |\coloring_i(U) \setminus \{\del\}|$. This implies that $\psi_i(U) \setminus \{\del\} \neq [\ncols]$, hence the gadget $W^j_{i, \psi_i, U} = \decode_{\psi_i(U) \setminus \{\del\}}(U,v)$, where $v$ is in the large independent set of $Y^j_{i, \psi_i}$, exists. The cardinality inequality also implies that $\coloring_i(U) \setminus \{\del\} \not\subseteq \psi_i(U) \setminus \{\del\}$, so by \cref{thm:decode_behavior} and the budget being tight on the graphs of the packing $\packing$, we see that $v$ cannot be deleted. Hence, $\hat{y}^j_{i, \psi_i}$ cannot be deleted either.
  \item Suppose that there is some $U \in \ttcset_i$ such that $|\psi_i(U) \setminus \{\del\}| > |\coloring_i(U) \setminus \{\del\}|$, then the first case must also apply for some $U' \in \ttcset_i$, as otherwise $\psi_i$ would perform too few deletions in $\ttcset_i$ and hence $\psi_i \notin \scolorings_i$.
  \item Since the first two cases do not apply, it follows that $|\psi_i(U) \setminus \{\del\}| = |\coloring_i(U) \setminus \{\del\}|$ for all $U \in \ttcset_i$. Due to $\psi_i \neq \coloring_i$, there must be some $U \in \ttcset_i$ such that $\psi_i(U) \setminus \{\del\} \neq \coloring_i(U) \setminus \{\del\}$. This implies that $\psi_i(U) \setminus \{\del\} \neq [\ncols]$ and hence the gadget $W^j_{i, \psi_i, U} = \decode_{\psi_i(U) \setminus \{\del\}}(U,v)$, where $v$ is in the large independent set of $Y^j_{i, \psi_i}$, exists. Since $|\psi_i(U) \setminus \{\del\}| = |\coloring_i(U) \setminus \{\del\}|$ and $\psi_i(U) \setminus \{\del\} \neq \coloring_i(U) \setminus \{\del\}$, we must have that $\coloring_i(U) \setminus \{\del\} \not\subseteq \psi_i(U) \setminus \{\del\}$, hence, as in the first case, $v$ cannot be deleted and neither can $\hat{y}^j_{i, \psi_i}$.
 \end{enumerate}
 This proves the claim regarding the deletion of $\hat{y}^j_{i, \psi_i}$. 
 
 Now, consider the gadget $Z^j$. To resolve $Z^j$, there has to be a thin arrow from some $\hat{y}^j_{i, \psi_i}$ to $Z^j$ that is active since $Z^j$ is not part of the packing $\packing$. By the previous claim, this implies that $\psi_i = \coloring_i$. Furthermore, by construction of $G$ such a thin arrow only exists if the partial truth assignment $\tassign_i = \embedding_i^{-1}(\coloring_i)$ satisfies the $j$-th clause. Note that the definition of $\tassign_i$ is independent of the considered clause.
 
 For some groups $i$, the partial solution $\coloring_i$ might not be in the image of $\embedding_i$, in that case $\tassign_i$ is an arbitrary partial truth assignment for the $i$-th variable group. By the previous argument, the truth assignment $\tassign = \bigcup_{i = 1}^{\ngrps} \tassign_i$ must satisfy all clauses of $\formula$.
\end{proof}

\begin{figure}
 \centering
 \begin{tikzpicture}
	\begin{pgfonlayer}{nodelayer}
		\node [style=wide rectangle] (0) at (5, 8) {$Z^j$};
		\node [style=wide rectangle] (1) at (-6, 3) {$Y^j_{i,\varphi_i}$};
		\node [style=wide rectangle] (2) at (8, 3) {$Y^j_{i', \varphi_{i'}}$};
		\node [style=wide rectangle] (3) at (1, 3) {$Y^j_{i,\psi_i}$};
		\node [style=wide rectangle] (4) at (15, 3) {$Y^j_{i', \psi_{i'}}$};
		\node [style=wide rectangle] (5) at (-7.75, 0) {$W^j_{i,\varphi_i, U}$};
		\node [style=wide rectangle] (6) at (-4.25, 0) {$W^j_{i,\varphi_i, U'}$};
		\node [style=wide rectangle] (7) at (-0.75, 0) {$W^j_{i,\psi_i, U}$};
		\node [style=wide rectangle] (8) at (2.75, 0) {$W^j_{i,\psi_i, U'}$};
		\node [style=wide rectangle] (9) at (6.25, 0) {$W^j_{i',\varphi_{i'}, U}$};
		\node [style=wide rectangle] (10) at (9.75, 0) {$W^j_{i',\varphi_{i'}, U'}$};
		\node [style=wide rectangle] (11) at (13.25, 0) {$W^j_{i',\psi_{i'}, U}$};
		\node [style=wide rectangle] (12) at (16.75, 0) {$W^j_{i',\psi_{i'}, U'}$};
	\end{pgfonlayer}
	\begin{pgfonlayer}{edgelayer}
		\draw [style=thin] (5) to (1);
		\draw [style=thin] (1) to (6);
		\draw [style=thin] (7) to (3);
		\draw [style=thin] (3) to (8);
		\draw [style=thin] (9) to (2);
		\draw [style=thin] (2) to (10);
		\draw [style=thin] (11) to (4);
		\draw [style=thin] (4) to (12);
		\draw [style=thin] (1) to (0);
		\draw [style=thin] (0) to (3);
		\draw [style=thin] (0) to (2);
		\draw [style=thin] (0) to (4);
	\end{pgfonlayer}
\end{tikzpicture}
 \caption{A high-level overview of the component corresponding to the $j$-th clause after removing the TCM $\tcm$. The arrows leading to $Z^j$ denote thin arrows.}
 \label{fig:dtc_clause_component}
\end{figure}
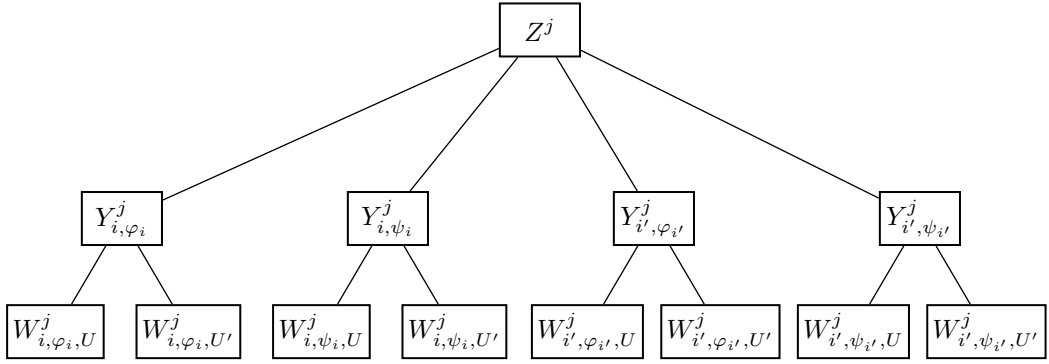

\begin{mythm}
 \label{thm:dtc_tw_bound}
 Let $G = G(\formula, \ncols, \vgrpsize)$ be the graph as constructed above. The set $\tcm = (\bigcup_{i = 1}^\ngrps \ttcset_i) \cup \{\{f_s\} \sep s \in [\ncols]\}$ is a TCM for $G$ of size $\ngrps \grpsize + \ncols$ to treewidth $\ncols$, i.e., $|\tcm| =  \ngrps \grpsize + \ncols$ and $\tw(G - \bigcup(\tcm)) \leq \ncols$.
\end{mythm}

\begin{proof}
 It follows from the construction of $G$ that $\tcm$ has size $\sum_{i = 1}^\ngrps |\ttcset_i| + |F| = \ngrps \grpsize + \ncols$ and only consists of twinclasses. The fact that all sets in $\tcm$ are twinclasses and not just sets of twins can be seen by considering the various color-set-gadgets $B_C(U,v)$ in the construction. It remains to argue that the treewidth of $G' := G - \bigcup(\tcm)$ is at most $\ncols$. We will show this by using the omniscient-cops-and-robber-game. The remaining graph $G'$ consists of several connected components, namely one connected component per copy of $L_{i, \ell, \setfam}$ and one connected component for every clause. Observe that the heads of thick $\ell$-arrows $\arrow_\ell(U,v)$, of color-set-gadgets $\decode_C(U,v)$, of thin arrows, and also the tails of thin arrows consist of only a single vertex. So, if the robber escapes through one of these gadgets and there is no other path back, then a single cop suffices to prevent the robber from going back.
 
 We begin by handling the connected components corresponding to a copy of $L_{i, \ell, \setfam}$. On a high level, these components look like stars; the center is $L_{i, \ell, \setfam}$ and for every $U \in \setfam$ the remainder $\arrow_\ell(U,v) - U$ of a thick $\ell$-arrow is attached to $L_{i, \ell, \setfam}$. By \cref{thm:critical_family}, $L_{i, \ell, \setfam}$ has pathwidth $\ncols$. With $\ncols + 1$ cops we sweep from left to right through the bags of the pathwidth decomposition of $L_{i, \ell, \setfam}$ until the robber escapes to one of the attached $\arrow_\ell(U,v) - U$. When that happens we remove all cops, except the cop on the head $v$ of the considered thick $\ell$-arrow from $L_{i, \ell, \setfam}$. We can then capture the robber using the strategy given by \cref{thm:arrow_behavior}.
 
 Next, we handle the connected components corresponding to a clause. \cref{fig:dtc_clause_component} gives a high-level overview how these components look, notice that at this level there are no cycles. Our strategy makes use of this hierarchy, we start at $Z^j$ and chase the robber downwards, making sure that the robber cannot go back upwards. By \cref{thm:critical_family}, $Z^j$ has pathwidth $\ncols$. Again, we sweep with $\ncols + 1$ cops from left to right through the bags of the pathwidth decomposition of $Z^j$, until the robber escapes through some thin arrow coming from some $Y^j_{i, \coloring_i}$. Then, we remove all cops except the one on the head of the thin arrow and place one cop on the tail $\hat{y}^j_{i, \coloring_i}$ of the thin arrow. Either the robber escapes to the thin arrow, where we can capture it easily using $\ncols - 1$ further cops or the robber escapes to $Y^j_{i, \coloring_i}$.
 
 If the robber escapes to $Y^j_{i, \coloring_i}$, then we remove the cop from the head of the thin arrow and place $\ncols$ cops on the $K_\ncols$ in $Y^j_{i, \coloring_i}$. This leaves us with one connected component $W^j_{i, \coloring_i, U} = \decode_{C}(U,v)$ for each $U \in \ttcset_i$ with $\coloring_i(U) \neq [\ncols]$. Consider the connected component the robber escaped to and move the cop from $\hat{y}^j_{i, \coloring_i}$ to $v$. Since there are still cops on the $K_\ncols$ in $Y^j_{i, \coloring_i}$, the robber cannot go back. We can now remove these cops on $K_\ncols$ and capture the robber using the strategy given by \cref{thm:decode_behavior}.
 
 This concludes the strategy. Since we have never placed more than $\ncols + 1$ cops simultaneously, we see by \cref{thm:cops_robber} that $G'$ has treewidth at most $\ncols$.
\end{proof}

\begin{proof}[Proof of \cref{thm:dtc_lower_bound}]
 Suppose there is some $\ncols \geq 2$ and $\eps > 0$ such that \DTC can be solved in time $\Oh^*((2^\ncols - \eps)^{|\tcm|})$, where $\tcm$ is a TCM to treewidth $\ncols$. We will show that there exists a $\delta < 1$ such that for any clause size $\clss$, we can solve $\clss$-\SAT in time $\Oh^*(2^{\delta \nvars})$, where $\nvars$ is the number of variables. This contradicts \SETH by \cref{thm:seth_hitting_set} and hence implies the desired lower bound.
 
 Given a $\clss$-\SAT instance $\formula$, we construct the graph $G = G(\formula, \ncols, \vgrpsize)$, where $\vgrpsize$ depends only on $\ncols$ and $\eps$ and will be chosen later. We can consider $\ncols$, $\eps$, and $\clss$ as constants, hence also $\vgrpsize = \vgrpsize(\ncols, \eps)$ and $\grpsize = \grpsize(\vgrpsize)$ are constant. First, we will argue that $G$ has polynomial size. The number of vertices in the various gadgets can be estimated as follows, where we have an additional summand $\ncols$ for the $(\ncols + 1)$-critical graphs, since \cref{thm:critical_family} only constructs them in increments of $\ncols$:
 \begin{itemize}
  \item $|V(G[\bigcup(\tcm)])| = \sum_{i = 1}^\ngrps |\ttcset_i|\ncols + |F| = \ngrps \grpsize \ncols + \ncols = \lceil \frac{\nvars}{\vgrpsize} \rceil \grpsize \ncols + \ncols = \Oh(\nvars)$.
  \item $|V(L_{i, \ell, \setfam})| \leq |\setfam| + \ncols \leq |\ttcset_i| + \ncols = \grpsize + \ncols = \Oh(1)$.
  \item $|V(Z^j)| \leq \clss 2^{\vgrpsize} + \ncols = \Oh(1)$.
  \item $|V(W^j_{i, \coloring_i, U})| \leq (\ncols + 1)^2 = \Oh(1)$.
  \item $|V(\arrow_\ell(U,v))| \leq \ncols^3 + 3\ncols + 1 = \Oh(1)$.
  \item Every thin arrow consists of $2\ncols + 1$ vertices. 
 \end{itemize}
 Next, we bound how often each gadget appears:
 \begin{itemize}
  \item $i$ can take on $\ngrps = \lceil \frac{\nvars}{\vgrpsize} \rceil = \Oh(n)$ values.      
  \item There are at most $\Oh(\nvars^\clss)$ clauses, hence $j$ can take on $\Oh(\nvars^\clss)$ values.
  \item $\ell$ can take on $\ncols = \Oh(1)$ values.
  \item For fixed $i$, $\setfam \subseteq \ttcset_i$, so $\setfam$ can take on $2^{|\ttcset_i|} = 2^\grpsize = \Oh(1)$ values.
  \item We create $1 + \ngrps \ncols \grpsize / 2 = \Oh(n)$ copies of each $L_{i, \ell, \setfam}$.
  \item For fixed $i$, $\coloring_i \in \scolorings_i$ and $|\scolorings_i| \leq (2^\ncols)^\grpsize = \Oh(1)$.
  \item For fixed $i$, $U \in \ttcset_i$, so $U$ can take on $|\ttcset_i| = \grpsize = \Oh(1)$ values.
  \item We create at most one (thin or thick) arrow per vertex in some $Z^j$ or $L_{i, \ell, \setfam}$. 
 \end{itemize}
 Together, all these bounds show that $G$ has polynomial size and following the construction of $G$, one can easily see that $G$ can also be constructed in polynomial time. Furthermore, from the proof of \cref{thm:dtc_tw_bound} we can obtain a tree decomposition of $G' := G - \bigcup(\tcm)$ of width $\ncols$ in polynomial time.
 
 To analyze the running time resulting from applying the reduction and running the assumed algorithm for \DTC, we first bound $\grpsize$ as follows:
 \begin{equation}
  \label{eq:grpsize_bound}
  \grpsize \leq \frac{\vgrpsize}{\ncols} + 2^{\ncols + 1} \left\lceil \log_{2^\ncols}\left({\vgrpsize}/{\ncols}\right) \right\rceil + 2^\ncols + 1.
 \end{equation}
 
 In the construction of $G = G(\formula, \ncols, \vgrpsize)$, we chose $\grpsize$ as the smallest integer such that $\grpsize$ is divisible by $2^\ncols$ and such that the quantity  from \cref{thm:scolorings_bound}, which we denote by $x$, is larger than $2^\vgrpsize$. The summand $2^\ncols$ in \cref{eq:grpsize_bound} accounts for the divisibility. It remains to show that the second property is satisfied, for this we work with $p = \frac{\vgrpsize}{\ncols} + 2^{\ncols + 1} \lceil \log_{2^\ncols}({\vgrpsize}/{\ncols}) \rceil + 1$. 
 
 We first observe that $(2^\ncols - 1)! / 2^{(2^\ncols)} \geq \frac{1}{4}$ for all $\ncols \geq 2$. Hence, $x \geq (2^\ncols)^\grpsize \grpsize^{-(2^\ncols)} / 4 =: x'$. Furthermore, observe that we have $2 {\vgrpsize}/{\ncols} \geq \grpsize$ for sufficiently large $\vgrpsize$. We proceed by showing that $x' \geq 2^\vgrpsize$:
 \begin{align*}
   x' & \geq (2^\ncols)^{\frac{\vgrpsize}{\ncols} + 2^{\ncols + 1} \lceil \log_{2^\ncols}({\vgrpsize}/{\ncols}) \rceil} \grpsize^{-(2^\ncols)}
   \geq 2^\vgrpsize \left(\frac{\vgrpsize}{\ncols}\right)^{(2^{\ncols + 1})} \grpsize^{-(2^\ncols)} \geq 2^\vgrpsize \left(\frac{\vgrpsize}{\ncols}\right)^{(2^{\ncols + 1})} \left(2 \frac{\vgrpsize}{\ncols}\right)^{-(2^\ncols)} \\
   & \geq 2^\vgrpsize \left({\vgrpsize}/{\ncols}\right)^{(2^{\ncols})} 2^{-(2^\ncols)} \geq 2^\vgrpsize,
 \end{align*}
 where we use ${\vgrpsize}/{\ncols} \geq 2$ for sufficiently large $\vgrpsize$ in the last inequality.
 
 Running the assumed algorithm for \DTC on $G$ decides, by \cref{thm:sat_to_dtc} and \cref{thm:dtc_to_sat}, the satisfiability of $\formula$. Since $G$ can be constructed in polynomial time and has parameter value $|\tcm| = \lceil \frac{\nvars}{\vgrpsize} \rceil \grpsize + \ncols$ by \cref{thm:dtc_tw_bound}, we can solve $\clss$-\SAT in time 
 \begin{align*}
  \Oh^*((2^\ncols - \eps)^{|\tcm|}) & = \Oh^*((2^\ncols - \eps)^{\lceil \frac{\nvars}{\vgrpsize} \rceil \grpsize + \ncols}) \leq \Oh^*((2^\ncols - \eps)^{\frac{\nvars}{\vgrpsize} \grpsize}) \\
  & \leq \Oh^*((2^\ncols - \eps)^{\frac{\nvars}{\ncols}} (2^\ncols - \eps)^{\frac{\nvars}{\vgrpsize}(2^{\ncols + 1} \lceil \log_{2^\ncols}({\vgrpsize}/{\ncols}) \rceil + 2^\ncols + 2)}).
 \end{align*}
 We have that $\Oh^*((2^\ncols - \eps)^{\frac{\nvars}{\ncols}}) \leq \Oh^*(2^{\delta_1 \nvars})$ for some $\delta_1 < 1$. It remains to upper bound the second factor, we will again use that $\vgrpsize$ can be chosen sufficiently large:
 \begin{alignat*}{4}
  \phantom{\leq} \quad & \Oh^*\left((2^\ncols - \eps)^{\frac{\nvars}{\vgrpsize}(2^{\ncols + 1} \lceil \log_{2^\ncols}({\vgrpsize}/{\ncols}) \rceil + 2^\ncols + 2)}\right) \quad & \leq\quad & \Oh^*\left((2^\ncols - \eps)^{2^{\ncols + 2} \frac{\nvars}{\vgrpsize} \log_{2^\ncols}({\vgrpsize}/{\ncols})}\right) \\
  \leq \quad & \Oh^*\left((2^\ncols - \eps)^{2^{\ncols + 2} \frac{\nvars}{\vgrpsize} \log_{2^\ncols}(\vgrpsize)}\right)\quad  & \leq \quad & \Oh^*\left(\left(\left(2^\ncols - \eps\right)^{2^{\ncols + 2} \frac{\log_{2^\ncols}(\vgrpsize)}{\vgrpsize}}\right)^{\nvars}\right) \\
  \leq \quad & \Oh^*\left(2^{\delta_2 \nvars}\right), & &
 \end{alignat*}
 where $\delta_2$ can be chosen to be arbitrarily close to $0$ by making $\vgrpsize$ sufficiently large. By choosing $\delta_2$ small enough so that $\delta := \delta_1 + \delta_2 < 1$, we obtain that $\clss$-\SAT can be solved in time $\Oh^*(2^{\delta \nvars})$. Since the choice of $\vgrpsize$ is independent of $\clss$, this running time holds for all $\clss$ and therefore \SETH would be false.
\end{proof}
 
\begin{cor}
 If \DTC can be solved in time $\Oh^*((2^\ncols - \eps)^{\tctd(G)})$ for some $\ncols \geq 2$ and $\eps > 0$, then \SETH is false.
\end{cor}

\begin{proof}
 Follows from \cref{thm:dtc_lower_bound} and \cref{thm:tcm_tctd}.
\end{proof}

\subsection{Sparse Setting}
\label{sec:dtc_lb_td}

We again view \emph{solutions} to \DTC as functions $\coloring \colon V(G) \rightarrow [\ncols] \cup \{\del\}$ with the property discussed in the outline, cf.~\cref{sec:outline}.

In this subsection, we show how to adapt the lower bound for the dense setting to the sparse setting. We will give the construction in full detail again, but since the principle of constructions is so similar we will only explain how to adapt the previous proofs. One can see that the lower bound is tight by a routine application of dynamic programming on tree decompositions.

\begin{mythm}
 \label{thm:dtc_lower_bound_td}
 If \DTC can be solved in time $\Oh^*((\ncols + 1 - \eps)^{|\modulator|})$ for some $\ncols \geq 2$ and $\eps > 0$, where $\modulator$ is a modulator to treewidth $\ncols$, then \SETH is false.
\end{mythm}

The remainder of this section is devoted to proving \cref{thm:dtc_lower_bound_td}. Assume that we can solve \DTC in time $\Oh^*((\ncols + 1 - \eps)^{|\modulator|})$ for some $\ncols \geq 2$ and $\eps > 0$. We provide for all clause sizes $\clss$, a reduction from $\clss$-\SAT with $\nvars$ variables to \DTC with a modulator to treewidth $\ncols$ of size approximately $\nvars \log_{\ncols + 1}(2)$. Together with the assumed faster algorithm, this implies a faster algorithm for $\clss$-\SAT, thus violating \SETH. We will consider $\ncols$ to be fixed from now on.

\paragraph*{Construction.} Consider a $\clss$-\SAT instance $\formula$ with $\nvars$ variables and $\nclss$ clauses. We enumerate the clauses and refer to them by their number. We pick an integer $\vgrpsize$ which only depends on $\eps$ and $\ncols$; we will describe how to choose $\vgrpsize$ later. We partition the variables of $\formula$ into groups of size at most $\vgrpsize$, resulting in $\ngrps = \lceil \nvars / \vgrpsize \rceil$ groups which will be indexed by $i$. Next, we choose the smallest integer $\grpsize$ such that $\grpsize$ is divisible by $\ncols + 1$ and $(\ncols + 1)^\grpsize \frac{\grpsize! \ncols!}{(\grpsize + \ncols)!} \geq 2^\vgrpsize$. We will now describe the construction of the \DTC instance $G = G(\formula, \ncols, \vgrpsize)$. 

\subparagraph*{Comparison to dense setting.}
The principle behind construction is essentially the same as for the dense setting, cf. \cref{sec:dtc_lb_cw}, but the gadgets can be simplified. The central twinclasses will simply be single vertices now. Hence, we will not have to distinguish between different numbers of deletions in a twinclass, instead we simply distinguish between whether a vertex is deleted or not. Most notably, this allows us to use thin arrows instead of thick $\ell$-arrows and the structure of the considered partial solutions on the central vertices simplifies significantly. The remaining gadgets structurally stay the same, but their size may change.

\subparagraph*{Construction of central vertices.} The central vertices of $G$ form the modulator to treewidth $\ncols$. For each variable group $i \in [\ngrps]$, we create an independent set $U_i$ consisting of $\grpsize$ vertices. Furthermore, we again have a central clique $F = \{f_s \sep s \in [\ncols]\}$ consisting of $\ncols$ vertices that is used to simulate \LC constraints. A solution $\coloring \colon V \rightarrow [\ncols] \cup \{\del\}$ satisfying $\delset \cap F = \emptyset$ and $\coloring(f_s) = s$ for all $s \in [\ncols]$ is called a \emph{normalized solution}.

As described in the outline, we only want to consider solutions that delete a fixed number of vertices per group $U_i$. By using slightly more vertices $\grpsize$ than enforced by the base conversion, sufficiently many solutions remain if we define $\scolorings_i$ as follows:
\begin{equation*}
 \scolorings_i = \{\coloring \colon U_i \rightarrow [\ncols] \cup \{\del\} \sep |\delset| = \grpsize / (\ncols + 1) \}
\end{equation*}

\begin{lem}
  We have that $|\scolorings_i| \geq (\ncols + 1)^\grpsize \frac{\grpsize! \ncols!}{(\grpsize + \ncols)!}$ for all $i \in [\ngrps]$.
\end{lem}

\begin{proof}
  Since $\{\coloring \colon U_i \rightarrow [\ncols] \sep |\coloring^{-1}(c)| = \grpsize / (\ncols + 1) \text{ for all } c \in [\ncols] \cup \{\bot\}\} \subseteq \scolorings_i$, we see that $|\scolorings_i| \geq \binom{\grpsize}{\frac{\grpsize}{\ncols + 1}, \ldots, \frac{\grpsize}{\ncols + 1}} = x$, where $x$ is the central multinomial coefficient. It can be seen that $x$ is a maximum of the function $(c_1, \ldots, c_{\ncols + 1}) \mapsto \binom{\grpsize}{c_1, \ldots, c_{\ncols + 1}}$. The number of summands in the multinomial theorem $\sum_{c_1 + \cdots c_{\ncols + 1} = \grpsize} \binom{\grpsize}{c_1, \ldots, c_{\ncols + 1}} = (\ncols + 1)^\grpsize$ is $\binom{\grpsize + \ncols}{\grpsize} = \frac{(\grpsize + \ncols)!}{\grpsize! \ncols!}$, which corresponds to the number of weak compositions of $\grpsize$ into $\ncols + 1$ parts, and $x$ is one of them. Hence, we see that $|\scolorings_i| \geq x \geq (\ncols + 1)^\grpsize \frac{\grpsize! \ncols!}{(\grpsize + \ncols)!}$.
\end{proof}

Hence, by the choice of $\grpsize$, we can pick for each group $i \in [\ngrps]$ an efficiently computable injective mapping $\embedding_i \colon \{0,1\}^\vgrpsize \rightarrow \scolorings_i$ that maps truth assignments of the $i$-th variable group to partial solutions on $U_i$ with the desired structure.

\subparagraph*{Budget.} The budget $\budget = |\packing| + \ngrps \grpsize / (\ncols + 1)$ consists of two parts again. The first part $|\packing|$ is allocated to a vertex-disjoint packing $\packing$ of $(\ncols + 1)$-critical graphs and the second part $\ngrps \grpsize / (\ncols + 1)$ is allocated to the central vertices. 

\subparagraph*{Enforcing structure on central vertices.} For every group $i \in [\ngrps]$ and subset $S \subseteq U_i$ with $|S| = \frac{\ncols}{\ncols + 1} \grpsize + 1$, we add an $(\ncols + 1)$-critical graph, denoted $L_{i,S}$, consisting of at least $|S|$ vertices. For every $u \in S$, we pick a private vertex $v$ in $L_{i,S}$ and add a thin arrow from $u$ to $v$. We create $1 + \ngrps \grpsize / (\ncols + 1) = 1 + (\budget - |\packing|)$ copies of $L_{i,S}$ and the incident arrows. This concludes the construction of the structure gadget.

\subparagraph*{Decoding gadgets.} For the $j$-th clause, group $i \in [\ngrps]$, partial solution $\coloring_i \in \scolorings_i$, we construct a decoding gadget $Y^j_{i, \coloring_i}$ as follows. The gadget $Y^j_{i, \coloring_i}$ consists of a large independent set joined to a $K_\ncols$, i.e., adding all edges between both sets. In other words, $Y^j_{i, \coloring_i}$ is a complete $(\ncols + 1)$-partite graph with one large independent set and all other independent sets in the partition are singletons. The large independent set consists of $\grpsize + 1$ vertices, where one of these vertices is distinguished and denoted by $\hat{y}^j_{i, \coloring_i}$. This concludes the construction of $Y^j_{i, \coloring_i}$. 

For the $j$-th clause, group $i \in [\ngrps]$, partial solution $\coloring_i \in \scolorings_i$, vertex $u \in U_i$, we pick a private vertex $v \neq \hat{y}^j_{i, \coloring_i}$ in the large independent set of $Y^j_{i, \coloring_i}$ and add the decoding gadget $\decode_{\{\coloring_i(u)\} \setminus \{\del\}}(\{u\},v)$ and denote this instance of the decoding gadget by $W^j_{i, \coloring_i, u}$. Observe that in the dense setting we did not attach decoding gadgets to twinclasses $U$ with $\coloring_i(U) \neq [\ncols]$, whereas there is no such exception in the sparse setting.

\subparagraph*{Clause gadgets.} For the $j$-th clause, we add an $(\ncols + 1)$-critical graph, denoted $Z^j$, consisting of at least $\clss 2^\vgrpsize$ vertices. For every group $i \in [\ngrps]$ and partial solution $\coloring_i \in \scolorings_i$ such that $\embedding^{-1}(\coloring_i)$ is a partial truth assignment satisfying the $j$-th clause, we pick a private vertex $v$ in $Z^j$ and add a thin arrow from $\hat{y}^j_{i, \coloring_i}$ to $v$. We ensured that $Z^j$ is large enough so that we can always pick such a private vertex. This concludes the construction of $G(\formula, \ncols, \vgrpsize)$, cf.\ \cref{fig:dtc_decode_sparse}.

\begin{figure}[h]
  \centering
  \begin{tikzpicture}
	\begin{pgfonlayer}{nodelayer}
		\node [style=none] (0) at (7, 5) {$\mathcal{U}_i$};
		\node [style=none] (1) at (17, 5) {$\mathcal{U}_{i'}$};
		\node [style=none] (2) at (5, 8) {$L_{i,S_i}$};
		\node [style=none] (3) at (7, 8) {$L_{i,S'_i}$};
		\node [style=none] (16) at (12, 5) {$\cdots$};
		\node [style=none] (17) at (7, 5.5) {};
		\node [style=none] (27) at (4, 1.25) {$Y_{i,\varphi_i}^j$};
		\node [style=none] (28) at (6, 1.25) {$Y_{i,\psi_i}^j$};
		\node [style=none] (29) at (8, 1.25) {$Y_{i,\varphi_i}^{j'}$};
		\node [style=none] (30) at (10, 1.25) {$Y_{i,\psi_i}^{j'}$};
		\node [style=none] (31) at (4, 2) {};
		\node [style=none] (32) at (6, 2) {};
		\node [style=none] (33) at (8, 2) {};
		\node [style=none] (34) at (10, 2) {};
		\node [style=none] (35) at (7, 4.5) {};
		\node [style=none] (36) at (14, 1.25) {$Y_{i',\varphi_{i'}}^j$};
		\node [style=none] (37) at (16, 1.25) {$Y_{i',\psi_{i'}}^j$};
		\node [style=none] (38) at (18, 1.25) {$Y_{i',\varphi_{i'}}^{j'}$};
		\node [style=none] (39) at (20, 1.25) {$Y_{i',\psi_{i'}}^{j'}$};
		\node [style=none] (40) at (14, 2) {};
		\node [style=none] (41) at (16, 2) {};
		\node [style=none] (42) at (18, 2) {};
		\node [style=none] (43) at (20, 2) {};
		\node [style=none] (44) at (17, 4.5) {};
		\node [style=none] (45) at (-1.25, 3.75) {color-gadgets $B_C$, $|C| \leq 1$:};
		\node [style=none] (46) at (1, 6.5) {thin arrows:};
		\node [style=none] (49) at (7, -2) {$Z^j$};
		\node [style=none] (50) at (17, -2) {$Z^{j'}$};
		\node [style=none] (51) at (4, 0.5) {};
		\node [style=none] (52) at (6, 0.5) {};
		\node [style=none] (53) at (8, 0.5) {};
		\node [style=none] (54) at (10, 0.5) {};
		\node [style=none] (55) at (14, 0.5) {};
		\node [style=none] (56) at (16, 0.5) {};
		\node [style=none] (57) at (18, 0.5) {};
		\node [style=none] (58) at (20, 0.5) {};
		\node [style=none] (59) at (6.75, -1.5) {};
		\node [style=none] (60) at (6.5, -2) {};
		\node [style=none] (61) at (7.5, -1.75) {};
		\node [style=none] (62) at (7.5, -2.25) {};
		\node [style=none] (63) at (16.25, -1.75) {};
		\node [style=none] (64) at (16.25, -2.25) {};
		\node [style=none] (65) at (17.75, -2) {};
		\node [style=none] (66) at (17.25, -1.5) {};
		\node [style=none] (67) at (1, -0.5) {thin arrows:};
		\node [style=none] (72) at (5, 7.5) {};
		\node [style=none] (73) at (7, 7.5) {};
		\node [style=none] (74) at (8.75, 8) {$\cdots$};
		\node [style=none] (75) at (15, 8) {$L_{i',S_{i'}}$};
		\node [style=none] (76) at (17, 8) {$L_{i',S'_{i'}}$};
		\node [style=none] (77) at (17, 5.5) {};
		\node [style=none] (78) at (15, 7.5) {};
		\node [style=none] (79) at (17, 7.5) {};
		\node [style=none] (80) at (18.75, 8) {$\cdots$};
		\node [style=none] (81) at (0, 8) {structure gadgets:};
		\node [style=none] (82) at (0.5, 5) {central vertices:};
		\node [style=none] (83) at (0.25, 1.25) {decoding gadgets:};
		\node [style=none] (84) at (0.5, -2) {clause gadgets:};
		\node [style=none] (86) at (-1, 2.75) {(attached to clique $F$)};
	\end{pgfonlayer}
	\begin{pgfonlayer}{edgelayer}
		\draw [style=thick arrow, in=90, out=-150, looseness=0.75] (35.center) to (31.center);
		\draw [style=thick arrow, in=90, out=-120] (35.center) to (32.center);
		\draw [style=thick arrow, in=90, out=-60] (35.center) to (33.center);
		\draw [style=thick arrow, in=90, out=-30, looseness=0.75] (35.center) to (34.center);
		\draw [style=thick arrow, in=90, out=-150, looseness=0.75] (44.center) to (40.center);
		\draw [style=thick arrow, in=90, out=-120] (44.center) to (41.center);
		\draw [style=thick arrow, in=90, out=-60] (44.center) to (42.center);
		\draw [style=thick arrow, in=90, out=-30, looseness=0.75] (44.center) to (43.center);
		\draw [style=thick arrow, in=180, out=-90, looseness=0.75] (51.center) to (60.center);
		\draw [style=thick arrow, in=90, out=-90] (52.center) to (59.center);
		\draw [style=thick arrow, in=0, out=-90, looseness=0.50] (55.center) to (61.center);
		\draw [style=thick arrow, in=0, out=-90, looseness=0.50] (56.center) to (62.center);
		\draw [style=thick arrow, in=180, out=-90, looseness=0.50] (53.center) to (64.center);
		\draw [style=thick arrow, in=180, out=-90, looseness=0.50] (54.center) to (63.center);
		\draw [style=thick arrow, in=90, out=-90] (57.center) to (66.center);
		\draw [style=thick arrow, in=0, out=-90, looseness=0.75] (58.center) to (65.center);
		\draw [style=thick arrow, in=-90, out=135, looseness=0.75] (17.center) to (72.center);
		\draw [style=thick arrow] (17.center) to (73.center);
		\draw [style=thick arrow, in=-90, out=135, looseness=0.75] (77.center) to (78.center);
		\draw [style=thick arrow] (77.center) to (79.center);
	\end{pgfonlayer}
\end{tikzpicture}
  \caption{An overview of the construction of the graph $G(\formula, \ncols, \vgrpsize)$ for the case $\ncols = 2$.}
  \label{fig:dtc_decode_sparse}
\end{figure}

\subparagraph*{Packing.} We construct a vertex-disjoint packing $\packing$ of $(\ncols + 1)$-critical graphs that fully explains the budget outside of the central vertices. For every thin arrow from $u$ to $v$ in the construction, we add the $K_{\ncols + 1}$ induced by the deletion edge incident to the head $v$ of the arrow to the packing $\packing$. Every $\decode_{\{\coloring_i(u)\} \setminus \{\del\}}(\{u\},v)$ contains, depending on whether $\coloring_i(u) \neq \del$ or not, $\ncols$ or $\ncols + 1$ disjoint $K_{\ncols + 1}$ corresponding to deletion edges, one of them being incident to $v$, and we add all these $K_{\ncols + 1}$ to the packing $\packing$. Finally, for every $Y^j_{i, \coloring_i}$, we add the $K_{\ncols + 1}$ induced by $\hat{y}^j_{i, \coloring_i}$ and the $K_\ncols$ of $Y^j_{i, \coloring_i}$ to the packing $\packing$. Since, we have added only $(\ncols + 1)$-critical graphs to $\packing$, the cost of $\packing$ is simply $|\packing|$. 

\begin{mythm}
 Let $\formula$ be a $\clss$-\SAT instance with $\nvars$ variables and $\nclss$ clauses. Let $G = G(\formula, \ncols, \vgrpsize)$ and $\packing$ be the graph and packing as constructed above and let $\budget = |\packing| + \ngrps \grpsize / (\ncols + 1)$. If $\formula$ has a satisfying truth assignment $\tassign$, then there is a solution $\coloring$ of the \DTC instance $(G, \budget)$.
\end{mythm}

\begin{proof}[Proof sketch]
 The proof is very similar to the proof of \cref{thm:sat_to_dtc}. Let $\tassign_i$ be the partial truth assignment of variable group $i$ induced by $\tassign$. We set $\coloring\big|_{U_i} = \embedding_i(\tassign_i) \in \scolorings_i$ for all $i$ and $\coloring(f_s) = s$ for all $s \in [\ncols]$. By definition of $\scolorings_i$, only the budget $|\packing|$ for the packing $\packing$ remains, hence on every graph in $\packing$, we can perform exactly one deletion and nowhere else. 
 
 We propagate deletions along thin arrows and extend the solution across decoding gadgets $\decode_{\{\coloring_i(u)\} \setminus \{\del\}}(\{u\},v)$ as before. For every $Y^j_{i, \coloring_i}$, we delete the distinguished vertex $\hat{y}^j_{i, \coloring_i}$ and for every $j$-th clause, group $i$, and solution $\coloring_i \neq \psi_i \in \scolorings_i$, we pick one of the vertices in the $K_\ncols$ of $Y^j_{i, \coloring_i}$ and delete it. This concludes the description of the deletions. 
 
 It remains to show that the remaining graph can be properly $\ncols$-colored. Comparing to the proof of \cref{thm:sat_to_dtc}, the arguments change slightly for $L_{i,S}$ and $Y^j_{i, \psi_i}$. Consider some copy of some $L_{i,S}$, due to $|S| + |\delset \cap U_i| = (\frac{\ncols}{\ncols + 1} \grpsize + 1) + \grpsize / (\ncols + 1) = \grpsize + 1 > |U_i|$ there is a least one $u \in S$ with $\coloring(u) = \del$. Therefore, the thin arrow from $u$ to $L_{i,S}$ propagates a deletion to $L_{i,S}$ and this $(\ncols + 1)$-critical graph is resolved. 
 
 Consider some $Y^j_{i, \psi_i}$ with $\psi_i \in \scolorings_i$. If $\psi_i \neq \coloring_i$, the arguments works as in the proof of \cref{thm:sat_to_dtc}. If $\psi_i = \coloring_i$, then we claim that the large independent set of $Y^j_{i, \coloring_i}$ is fully deleted. The large independent set has size $\grpsize + 1 = |U_i| + 1$. We distributed one deletion to $\hat{y}^j_{i, \coloring_i}$ and all other vertices $v$ in the large independent set are hit by some $W^j_{i, \coloring_i, u}$, $u \in U_i$. Since $\coloring_i(u) \in \{\coloring_i(u)\}$, we see that $W^j_{i, \coloring_i, u} = \decode_{\{\coloring_i(u)\} \setminus \{\del\}}(\{u\},v)$ propagates a deletion to $v$ due to \cref{thm:decode_behavior}. So, the large independent set is indeed fully deleted and thereby the complete $(\ncols + 1)$-partite graph is resolved.
\end{proof}

\begin{mythm}
 \label{thm:enforced_structure_td}
 Let $\formula$ be a $\clss$-\SAT instance with $\nvars$ variables and $\nclss$ clauses. Let $G = G(\formula, \ncols, \vgrpsize)$ and $\packing$ be the graph and packing as constructed above and let $\budget = |\packing| + \ngrps \grpsize / (\ncols + 1)$. If $\coloring$ is a solution of the \DTC instance $(G, \budget)$, then $|\delset| = \budget$. Furthermore, there is a normalized solution $\psi$ with $|\psi^{-1}(\del)| = |\delset| = \budget$ and $\psi\big|_{U_i} \in \scolorings_i$ for all $i \in [\ngrps]$.
\end{mythm}

\begin{proof}
 Since the packing $\packing$ consists only of $(\ncols + 1)$-critical graphs, at least $|\packing|$ deletions must be performed by $\coloring$ on these graphs. The remainder of the deletions is performed on the central vertices. Suppose that there is some group $i \in [\ngrps]$ such that $|\delset \cap U_i| < \grpsize / (\ncols + 1)$, then there exists an $S \subseteq U_i$ with $|S| = \frac{\ncols}{\ncols + 1} \grpsize + 1$ and $S \cap \delset = \emptyset$. Consider some copy of the corresponding $L_{i,S}$ and notice that all thin arrows leading to $L_{i,S}$ are passive unless we pay for extra deletions. Since $L_{i,S}$ is an $(\ncols + 1)$-critical graph that does not belong to $\packing$, we must perform one additional deletion per copy of $L_{i,S}$. There are $1 + \ngrps \grpsize/(\ncols + 1) > \budget - |\packing|$ copies of $L_{i,S}$, hence this would exceed the available budget. So, we can conclude that $|\delset \cap U_i | \geq \grpsize / (\ncols + 1)$ for all groups $i \in [\ngrps]$.
 
 Together with $|\delset| \leq \budget = |\packing| + \ngrps \grpsize / (\ncols + 1)$, we see that $|\delset| = \budget$ and $|\delset \cap U_i | = \grpsize / (\ncols + 1)$ for all $i \in [\ngrps]$. Hence, $\coloring$ cannot delete any vertex of the central clique $F$ and by permuting the colors, we obtain the desired normalized solution $\psi$ with $\psi\big|_{U_i} \in \scolorings_i$.
\end{proof}

\begin{mythm}
 Let $\formula$ be a $\clss$-\SAT instance with $\nvars$ variables and $\nclss$ clauses. Let $G = G(\formula, \ncols, \vgrpsize)$ and $\packing$ be the graph and packing as constructed above and let $\budget = |\packing| + \ngrps \grpsize / (\ncols + 1)$. If the \DTC instance $(G, \budget)$ has a solution $\coloring$, then $\formula$ has a satisfying truth assignment $\tassign$.
\end{mythm}

\begin{proof}[Proof sketch]
 This proof is very similar to the proof of \cref{thm:dtc_to_sat}. We first invoke \cref{thm:enforced_structure_td}, which implies that $|\delset| = \budget$ and allows us to assume that $\coloring$ is a normalized solution with $\coloring_i := \coloring\big|_{U_i} \in \scolorings_i$. We only diverge from the proof of \cref{thm:dtc_to_sat} when proving that $\coloring(\hat{y}^j_{i, \psi_i}) = \del$ only if $\psi_i = \coloring_i$. 
 
 As before, $\hat{y}^j_{i, \psi_i}$ may only be deleted if the large independent set of $Y^j_{i, \psi_i}$ is fully deleted. Now, consider some $\psi_i \in \scolorings_i \setminus \{\coloring_i\}$. We will distinguish two cases.
 \begin{enumerate}
   \item Suppose that $\psi_i^{-1}(\del) \neq \coloring_i^{-1}(\del)$. Due to $\coloring_i, \psi_i \in \scolorings_i$, there exists some $u \in U_i$ with $\psi_i(u) = \del$ and $\coloring_i(u) \neq \del$. The associated $W^j_{i, \psi_i, u} = \decode_{\emptyset}(\{u\},v)$ cannot propagate a deletion to $v$ in this case by \cref{thm:decode_behavior} and hence the large independent set of $Y^j_{i, \psi_i}$ is not fully deleted.
   \item Suppose that $\psi_i^{-1}(\del) = \coloring_i^{-1}(\del)$, then $\psi_i \neq \coloring_i$ implies that there exists some $u \in U_i \setminus \psi_i^{-1}(\del) = U_i \setminus \coloring_i^{-1}(\del)$ with $\psi_i(u) \neq \coloring_i(u)$. Again, the associated $W^j_{i, \psi_i, u} = \decode_{\{\psi_i(u)\}}(\{u\},v)$ cannot propagate a deletion to $v$ in this case by \cref{thm:decode_behavior} and, again, the large independent set of $Y^j_{i, \psi_i}$ is not fully deleted. 
 \end{enumerate}
 This proves the claim regarding the deletion of $\hat{y}^j_{i, \psi_i}$.
\end{proof}

\begin{mythm}
 \label{thm:dtc_tw_bound_td}
 Let $G = G(\formula, \ncols, \vgrpsize)$ be the graph as constructed above. The set $\modulator = \left(\bigcup_{i = 1}^\ngrps U_i\right) \cup \{f_s \sep s \in [\ncols]\}$ is a modulator for $G$ of size $\ngrps \grpsize + \ncols$ to treewidth $\ncols$, i.e., $|\modulator| = \ngrps \grpsize + \ncols$ and $\tw(G - \modulator) \leq \ncols$.
\end{mythm}

\begin{proof}[Proof sketch]
 The proof is very similar to the proof of \cref{thm:dtc_tw_bound}. The notable difference is that all thick arrows are replaced by thin arrows, but this does not affect the treewidth. The connected components corresponding to a clause have for the sake of the  omniscient cops-and-robber-game the same structure as before, only the number of $W^j_{i, \coloring_i, u}$ incident to some $Y^j_{i, \coloring_i}$ changes. 
\end{proof}

\begin{proof}[Proof of \cref{thm:dtc_lower_bound_td}]
 Suppose there is some $\ncols \geq 2$ and $\eps > 0$ such that \DTC can be solved in time $\Oh^*((\ncols + 1 - \eps)^{|\modulator|})$, where $\modulator$ is a modulator to treewidth $\ncols$. We will show that there exists a $\delta < 1$, such that for any $\clss$, we can solve $\clss$-\SAT in time $\Oh^*(2^{\delta \nvars})$ where $\nvars$ is the number of variables. This contradicts \SETH and hence implies the desired lower bound.
 
  Given a $\clss$-\SAT instance $\formula$, we construct the graph $G = G(\formula, \ncols, \vgrpsize)$, where $\vgrpsize$ depends only on $\ncols$ and $\eps$ and will be chosen later. We can consider $\ncols$, $\eps$, and $\clss$ as constants, hence also $\vgrpsize = \vgrpsize(\ncols, \eps)$ and $\grpsize = \grpsize(\vgrpsize)$ are constant. As in the proof of \cref{thm:dtc_lower_bound}, we can again see that $G$ has polynomial size and can be constructed in polynomial time. Furthermore, we can also obtain a tree decomposition of $G' := G - \modulator$ from the proof of \cref{thm:dtc_tw_bound_td} in polynomial time.
  
  In the construction of $G(\formula, \ncols, \vgrpsize)$, we chose $\grpsize$ as the smallest integer that is divisible by $(\ncols + 1)$ and satisfies $(\ncols + 1)^\grpsize \frac{\grpsize! \ncols!}{(\grpsize + \ncols)!} \geq 2^\vgrpsize$. We let $\gamma = \lceil \log_{\ncols + 1}(2^{\vgrpsize}) \rceil$ and show that $\grpsize \leq \gamma + 3 \ncols \lceil \log_{\ncols + 1} \gamma \rceil + (\ncols + 1)$. The summand $(\ncols + 1)$ will ensure divisibility. We compute that
  \begin{align*}
    (\ncols + 1)^\grpsize \frac{\grpsize! \ncols!}{(\grpsize + \ncols)!} & \geq (\ncols + 1)^\grpsize (2 \grpsize)^{-\ncols} = \frac{(\ncols + 1)^\gamma \gamma^{3\ncols}}{2^\ncols (\gamma + 3 \ncols \lceil \log_{\ncols + 1} \gamma \rceil + (\ncols + 1))^\ncols} \geq \frac{(\ncols + 1)^\gamma \gamma^{3\ncols}}{2^\ncols (5 \ncols \gamma)^\ncols} \\
    & = (\ncols + 1)^\gamma \frac{\gamma^{2\ncols}}{(10 \ncols)^\ncols} \geq (\ncols + 1)^\gamma \geq 2^\vgrpsize,
  \end{align*}
 where we use throughout that $\vgrpsize$ and hence $\gamma$ is sufficiently large.
  
  Running the assumed algorithm for \DTC on $G$ decides, by \cref{thm:sat_to_dtc} and \cref{thm:dtc_to_sat}, the satisfiability of $\formula$. Since $G$ can be constructed in polynomial time, we can solve $\clss$-\SAT in time
  \begin{align*}
   \Oh^*((\ncols + 1 - \eps)^{|\modulator|}) & = \Oh^*((\ncols + 1 - \eps)^{\lceil \frac{\nvars}{\vgrpsize} \rceil \grpsize + \ncols}) \leq \Oh^*((\ncols + 1 - \eps)^{\frac{\nvars}{\vgrpsize} \grpsize}) \\
   & \leq \Oh^*((\ncols + 1 - \eps)^{\frac{\nvars}{\vgrpsize} \log_{\ncols + 1}(2^\vgrpsize)} (\ncols + 1 - \eps)^{\frac{\nvars}{\vgrpsize}(3 \ncols \lceil \log_{\ncols + 1}(\gamma) \rceil + \ncols + 2)}).
  \end{align*}
  For the first factor, we see that 
  \begin{equation*}
   \Oh^*((\ncols + 1 - \eps)^{\frac{\nvars}{\vgrpsize} \log_{\ncols + 1}(2^\vgrpsize)}) = \Oh^*((\ncols + 1 - \eps)^{\nvars \log_{\ncols + 1}(2)}) \leq \Oh^*(2^{\delta_1 \nvars}),
  \end{equation*}
  where $\delta_1 < 1$. We bound the exponent in the second factor as follows
  \begin{align*}
   \frac{\nvars}{\vgrpsize}(3 \ncols \lceil \log_{\ncols + 1}(\gamma) \rceil + \ncols + 2) & \leq \frac{\nvars}{\vgrpsize}(6 \ncols \log_{\ncols + 1}(\gamma)) \leq \frac{\nvars}{\vgrpsize}(6 \ncols \log_{\ncols + 1}(p_0^2)) \\
   & = \frac{\nvars}{\vgrpsize}(12 \ncols \log_{\ncols + 1}(\vgrpsize)) = 12 \ncols \nvars \frac{\log_{\ncols + 1}(\vgrpsize)}{\vgrpsize},
  \end{align*} 
  where we use $\gamma \leq p_0^2$ in the second inequality. This bound shows that 
  \begin{equation*}
   \Oh^*\left((\ncols + 1 - \eps)^{\frac{\nvars}{\vgrpsize}(3 \ncols \lceil \log_{\ncols + 1}(\gamma) \rceil + \ncols + 2)}\right) \leq \Oh^*\left(\left((\ncols + 1 - \eps)^{12 \ncols \frac{\log_{\ncols + 1}(\vgrpsize)}{\vgrpsize}}\right)^\nvars\right) \leq \Oh^*(2^{\delta_2 \nvars}),
  \end{equation*}
  where we can choose $\delta_2$ arbitrarily close to $0$ by making $\vgrpsize$ large enough.
  
By choosing $\delta_2$ so that $\delta := \delta_1 + \delta_2 < 1$, we obtain that $\clss$-\SAT can be solved in time $\Oh^*\left(2^{\delta \nvars}\right)$. Since the choice of $\vgrpsize$ is independent of $\clss$, we obtain this running time for all $\clss$ and hence \SETH has to be false.
\end{proof}

\begin{cor}
 If \DTC can be solved in time $\Oh^*((\ncols + 1 - \eps)^{\td(G)})$ for some $\ncols \geq 2$ and $\eps > 0$, then \SETH is false.
\end{cor}

\begin{proof}
 Follows from \cref{thm:dtc_lower_bound_td} and \cref{thm:mod_to_tw_implies_td}.
\end{proof}

\section{Lower Bound for Vertex Cover}
\label{sec:vc}

In \cref{sec:vc_lb} we prove the lower bound for \VC parameterized by a modulator to pathwidth 2 and in \cref{sec:further_consequences} we show that this lower bound also implies lower bounds for \MC and \KFD. Note that \VC is just another name for \textsc{Deletion to $1$-Colorable}, but the case $\ncols = 1$ is not covered by the previous lower bounds.

\subsection{Lower Bound}
\label{sec:vc_lb}

This section is devoted to establishing the lower bound for \VC when parameterized by a modulator to pathwidth 2, i.e., \cref{thm:vc_lb}. For \VC, we do not need to convert between different bases in the running time. By additionally reducing from $\clss$-\HS instead of $\clss$-\SAT, we obtain a significantly simplified reduction that does not require the trick of Cygan et al.~\cite{CyganDLMNOPSW16}. We construct a graph so that each vertex in the modulator corresponds to an element in the universe of the $\clss$-\HS instance. We construct gadgets of pathwidth at most 2 that simulate the $\clss$-\HS constraints on the modulator. Using \cref{thm:seth_hitting_set}, this implies the desired lower bound for \VC if \SETH is true. \VC can be solved in time $\Oh^*(2^{\cw(G)})$ by an algorithm of Courcelle et al.~\cite{CourcelleMR00} and by \cref{thm:mod_to_tw_implies_td} and \cref{thm:hierarchy_cliquewidth} this implies that the obtained lower bound is tight. So, in the exceptional case of \VC, the complexity when parameterized by clique-width is already explained by the sparse setting and we do not need to consider twinclasses.

\begin{mythm}
 \label{thm:vc_lb}
 If \VC can be solved in time $\Oh^*((2 - \eps)^{|\modulator|})$ for some $\eps > 0$, where $\modulator$ is a modulator to pathwidth at most 2, then \SETH is false.
\end{mythm}

\newcommand{\central}{W}
\newcommand{\PP}[1]{{P}^{#1}}

 Suppose that we can solve \VC in time $\Oh^*((2 - \eps)^{|\modulator|})$ for some $\eps > 0$. Fix an integer $\clss$ and let $(\UU = \{u_1, \ldots, u_n\}, \family)$ be a $\clss$-\HS instance with $|\family| = \nclss$ sets of size at most $\clss$ each, and budget $t$. The elements of the $j$-th set in $\family$, $j \in [m]$, are denoted by $\{u^j_1, u^j_2, \ldots, u^j_{p_j}\}$, $p_j \leq \clss$. We will reduce $(\UU, \family)$ to a \VC instance $G = G(\UU, \family)$ with budget $\budget$, which will be defined later, and which admits  a modulator to pathwidth $2$ of size $\nvars$.
 
\subparagraph*{Construction.}
 We create an independent set of $n$ central vertices $\central = \{w_1, \ldots, w_n\}$. The vertices in $\central$ that are chosen by a vertex cover will correspond to the hitting set. For the $j$-th set $\{u^j_1, u^j_2, \ldots, u^j_{p_j}\} \in \family$, $p_j \leq \clss$, we create a triangle path, denoted by $\PP{j}$: The triangle path $\PP{j}$ consists of $p_j$ vertices $a_{s}^j$, where $s \in [p_j]$, and $2p_j + 2$ vertices $b_{s}^j$, where $s \in [2p_j + 2]$, such that $a_{s}^j, b_{2s}^j, b_{2s + 1}^j$ form a triangle for all $s \in [p_j]$ and $b_1^j, b_2^j \ldots, b_{2p_j + 2}^j$ form a path, see \cref{fig:vertex_cover}. We connect $\PP{j}$ to $\central$ by adding, for all $s \in [p_j]$, an edge between $a_{s}^j$ and $w_{s'}$, where $s' \in [n]$ so that $u_{s'} = u^j_s$. The key property of the triangle path $\PP{j}$ is that it costs more to cover $\PP{j}$ if no vertex in $N(\PP{j}) \subseteq \central$ is inside the vertex cover. Since the budget constraint will be tight, this implies that a vertex cover has to take at least one vertex in each $N(\PP{j})$.
 
 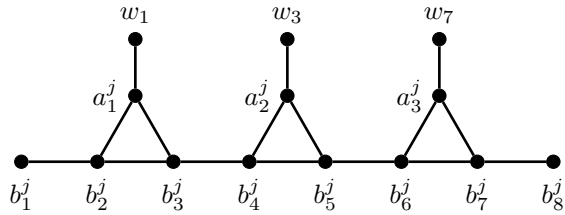
\begin{figure}[h]
  \centering
  \begin{tikzpicture}
	\begin{pgfonlayer}{nodelayer}
		\node [style=filled, label={below:$b_1^j$}] (0) at (0, 0.75) {};
		\node [style=filled, label={below:$b_2^j$}] (2) at (2, 0.75) {};
		\node [style=filled, label={left:$a_1^j$}] (3) at (3, 2.5) {};
		\node [style=filled, label={below:$b_3^j$}] (4) at (4, 0.75) {};
		\node [style=filled, label={below:$b_4^j$}] (6) at (6, 0.75) {};
		\node [style=filled, label={left:$a_2^j$}] (8) at (7, 2.5) {};
		\node [style=filled, label={below:$b_5^j$}] (9) at (8, 0.75) {};
		\node [style=filled, label={below:$b_6^j$}] (11) at (10, 0.75) {};
		\node [style=filled, label={left:$a_3^j$}] (12) at (11, 2.5) {};
		\node [style=filled, label={below:$b_7^j$}] (13) at (12, 0.75) {};
		\node [style=filled, label={below:$b_8^j$}] (14) at (14, 0.75) {};
		\node [style=filled, label={above:$w_1$}] (19) at (3, 4) {};
		\node [style=filled, label={above:$w_3$}] (20) at (7, 4) {};
		\node [style=filled, label={above:$w_7$}] (21) at (11, 4) {};
	\end{pgfonlayer}
	\begin{pgfonlayer}{edgelayer}
		\draw [style=thick] (0) to (2);
		\draw [style=thick] (2) to (4);
		\draw [style=thick] (4) to (6);
		\draw [style=thick] (6) to (9);
		\draw [style=thick] (9) to (11);
		\draw [style=thick] (11) to (13);
		\draw [style=thick] (13) to (14);
		\draw [style=thick] (3) to (2);
		\draw [style=thick] (3) to (4);
		\draw [style=thick] (19) to (3);
		\draw [style=thick] (20) to (8);
		\draw [style=thick] (8) to (6);
		\draw [style=thick] (8) to (9);
		\draw [style=thick] (11) to (12);
		\draw [style=thick] (12) to (13);
		\draw [style=thick] (21) to (12);
	\end{pgfonlayer}
\end{tikzpicture}
  \caption{The triangle path $\PP{j}$, where the $j$-th set consists of $u_1$, $u_3$, and $u_7$.}
  \label{fig:vertex_cover}
 \end{figure}
 
 \begin{lem}
 \label{thm:vc_triangle_path}
 Let $j \in [m]$ and $\cover$ be a vertex cover of $G$, then $|\cover \cap \PP{j}| \geq 2p_j$ and:
 \begin{enumerate}
   \item If $\cover \cap N(\PP{j}) = \emptyset$, then $|\cover \cap \PP{j}| \geq 2p_j + 1$.
   \item If $|\cover \cap N(\PP{j})| \geq 1$, then there is a vertex cover $\cover'$ with $\cover \setminus \PP{j} = \cover' \setminus \PP{j}$ and $|\cover' \cap \PP{j}| \leq 2p_j$.
 \end{enumerate}
 \end{lem}
 
 \begin{proof}
  Let $\cover$ be a vertex cover of $G$. Consider some $\PP{j}$ and notice that the triangles $\{a^j_{s}, b^j_{2s}, b^j_{2s + 1}\}$, $s \in [p_j]$, are vertex-disjoint and any vertex cover has to contain at least 2 vertices in each triangle. Hence, $|\cover \cap \PP{j}| \geq 2p_j$ for all $j \in [m]$.
  
  Now, suppose that $\cover \cap N(\PP{j}) = \emptyset$ for some $j \in [m]$. By assumption, the $p_j$ edges between $\PP{j}$ and $\central$ are not covered by $\cover \cap \central$ and hence $\cover$ must contain $\{a^j_{1} , a^j_{2}, \ldots, a^j_{p_j}\}$. Hence, $\cover$ contains all $a$-vertices of $\PP{j}$ and after removing these from $\PP{j}$ a path on $2p_j + 2$ vertices remains. Therefore, $\cover$ has to contain at least $p_j + 1$ vertices of this path. In total, $\cover$ contains at least $2p_j + 1$ vertices of $\PP{j}$. 
  
  Lastly, suppose that $\cover$ is a vertex cover with $|\cover \cap N(\PP{j})| \geq 1$ for some $j \in [m]$. Let $s^* \in [p_j]$ be the smallest integer such that the neighbor in $\central$ of $a^j_{s^*}$ belongs to $\cover$. We define $\cover'_{P} = \{a^j_{s} \sep s \in [p_j] \setminus \{s^*\} \} \cup \{b^j_{2s} \sep s \in [s^*]\} \cup \{b^j_{2s + 1} \sep s \in [p_j] \setminus [s^* - 1]\}$ and claim that $\cover'_P$ is a vertex cover of $G[\PP{j}]$. Notice that $\cover'_P$ contains $2$ vertices in each triangle $\{a^j_{s}, b^j_{2s}, b^j_{2s + 1}\}$, $s \in [p_j]$, and the edges $\{b^j_{s}, b^j_{s+1}\}$ are covered by the $b^j$ with an even subscript if $s \leq 2s^*$ and otherwise by the $b^j$ with an odd subscript. Hence, $\cover'_P$ is a vertex cover of $G[\PP{j}]$. We define $\cover' = (\cover \setminus \PP{j}) \cup \cover'_P$. All edges between $\PP{j}$ and $\central$ are covered by $\cover' \cap \PP{j}$ except the edge incident to $a^j_{s^*}$ which is covered by $\cover' \cap \central$ by assumption. Therefore, $\cover'$ is a vertex cover of $G$ with $\cover \setminus \PP{j} = \cover' \setminus \PP{j}$ and $|\cover' \cap \PP{j}| = |\cover'_P| = (p_j - 1) + s^* + (p_j + 1 - s^*) = 2p_j$.
 \end{proof}
 
 \begin{lem}
  \label{thm:vc_reduction_correct}
  Let $\clss \in \NN$ and $(\UU, \family)$ be a $\clss$-\HS instance with $|\UU| = \nvars$ and $|\family| = \nclss$. There is a hitting set for $(\UU, \family)$ of size at most $\hsbudget$ if and only if $G = G(\UU, \family)$ has a vertex cover of size at most $\budget = \hsbudget + 2 \sum_{j \in [m]} p_j$.
 \end{lem}
 
 \begin{proof}
 Let $H$ be a hitting set of $(\UU, \family)$ of size $|H| \leq \hsbudget$. Let $\cover_\central = \{w_i \in \central \sep u_i \in H\}$ and consider $\cover = \cover_\central \cup \bigcup_{j = 1}^m \PP{j}$. The set $\cover$ is certainly a vertex cover of $G$ and by applying the second part of  \cref{thm:vc_triangle_path} for every $j \in [m]$, which is possible since $H$ was a hitting set, we obtain a vertex cover $\cover'$ with $\cover' \cap \central = \cover_\central$ and $|\cover' \cap \PP{j}| \leq 2p_j$ for all $j \in [m]$. Hence, $\cover'$ is a vertex cover of $G$ of size at most $\hsbudget + 2 \sum_{j \in [m]} p_j$.
 
 For the other direction, let $\cover$ be a vertex cover of $G$ of size at most $\hsbudget + 2 \sum_{j \in [m]} p_j$. Due to \cref{thm:vc_triangle_path}, we have that $|\cover \cap \PP{j}| \geq 2p_j$ for all $j \in [m]$. Define $J_\cover = \{j \in [m] \sep |\cover \cap \PP{j}| > 2p_j\}$ and $\cover_\central = \cover \cap \central$. The size constraint implies that $|J_\cover| + |\cover_\central| \leq \hsbudget$. If $J_\cover = \emptyset$, then by the first part of \cref{thm:vc_triangle_path} we must have that $\cover \cap N(\PP{j}) \neq \emptyset$ for all $j \in [m]$. By the construction of $G$ this means that $H = \{u_i \in \UU \sep w_i \in \cover_\central\}$ must be a hitting set for $(\UU, \family)$ of size at most $\hsbudget$.
 
 If $J_\cover \neq \emptyset$, then take some $j \in J_\cover$ and some $w_i \in N(\PP{j})$ and construct the vertex cover $\cover' = \cover \cup \{w_i\}$. By the second part of \cref{thm:vc_triangle_path}, there is a vertex cover $\cover''$ with $\cover'' \setminus \PP{j} = \cover' \setminus \PP{j}$ and $|\cover'' \cap \PP{j}| \leq 2p_j$. This implies that $|\cover''| = |\cover| + 1 + |\cover'' \cap \PP{j}| - |\cover \cap \PP{j}| \leq |\cover| + 1 + 2p_j - (2p_j + 1) = |\cover|$ and $J_{\cover''} = J_\cover \setminus \{j\}$. We replace $\cover$ by $\cover''$ and repeat this argument until $J_\cover = \emptyset$ and then obtain the desired hitting set by the previous argument.
 \end{proof} 
 
 \begin{lem}
  \label{thm:vc_param_bound}
  It holds that $\pw(\PP{j}) = \tw(\PP{j}) = 2$ for all $j \in [m]$.
 \end{lem}
 
  \begin{proof}
 We construct a path decomposition $\TT$ for $\PP{j}$ of width $2$. Let $\TT$ be a path on $2p_j + 1$ vertices $t_1, t_2, \ldots t_{2p_j + 1}$. For $s \in [p_j + 1]$, define the bag $\bag_{t_{2s - 1}} = \{b_{2s-1}^j, b_{2s}^j\}$, and for $s \in [p_j]$ define bags $\bag_{t_{2s}} = \{a_s^j, b_{2s}^j, b_{2s + 1}^j\}$. Using these bags, $\TT$ is a path decomposition of width $2$ for $\PP{j}$. Since $\PP{j}$ is neither a linear forest nor a forest, we have that $\pw(\PP{j}) \geq 2$ and $\tw(\PP{j}) \geq 2$.
 \end{proof} 
 
 \begin{proof}[Proof of \cref{thm:vc_lb}]
  Suppose that \textsc{Vertex Cover} can be solved in time $\Oh^*((2 - \eps)^{|\modulator|})$ for some $\eps > 0$, where $\modulator$ is a modulator to pathwidth at most 2. We argue that we can then solve $\clss$-\HS in time $\Oh^*((2 - \eps)^n)$ for all $\clss$, thus violating \SETH by \cref{thm:seth_hitting_set}. 
  
  Fix an arbitrary integer $\clss \geq 1$ and a $\clss$-\HS instance $(\UU, \family)$ with $|\UU| = n$ and $|\family| = m$. From $(\UU, \family)$ construct the graph $G = G(\UU, \family) = (V, E)$ as described above in polynomial time. Let $\modulator = \central$, then $G - \modulator$ is the disjoint union of the $\PP{j}$, $j \in [m]$, and by \cref{thm:vc_param_bound} the pathwidth of $G - \modulator$ is 2. So, $\modulator$ is a modulator to pathwidth at most $2$ of size $|\modulator| = |\central| = n$. Running the \textsc{Vertex Cover} algorithm on $G$ with budget $\budget = \hsbudget + 2 \sum_{j \in [m]} p_j$ and modulator $\modulator$ solves $\clss$-\HS by \cref{thm:vc_reduction_correct}. As the size of $G$ is polynomial in $|\UU| = n$ and $|\family| = m$, the resulting running time is $\Oh^*((2 - \eps)^n)$ and hence \SETH must be false.  
 \end{proof}

\begin{cor}
 If \textsc{Vertex Cover} can be solved in time $\Oh^*((2 - \eps)^{\td(G)})$ for some $\eps > 0$, then \SETH is false. 
\end{cor}

\begin{proof}
  Follows from \cref{thm:vc_lb} and \cref{thm:mod_to_tw_implies_td}
\end{proof}

\subsection{Further Consequences}
\label{sec:further_consequences}

\begin{cor}
 \label{thm:mc_lb}
 If \MC can be solved in time $\Oh^*((2 - \eps)^{|\modulator|})$ for some $\eps > 0$, where $\modulator$ is a modulator to treewidth at most 2, then \SETH is false.
\end{cor}

\begin{proof}
 We use a well-known reduction from \VC to \MC by Garey et al.~\cite{GareyJS76}. Let $G = (V, E)$ be a graph. We construct another graph $G' = (V', E')$ as follows: The vertex set is given by $V' = V \cup \{x\} \cup \{e_u, e_v \sep e = \{u,v\} \in E\}$. The vertex $x$ is adjacent to all $v \in V$ and for every edge $e = \{u, v\} \in E$, we add the five edges $\{x, e_u\}$, $\{x, e_v\}$, $\{e_u, e_v\}$, $\{e_u, u\}$, and $\{e_v, v\}$ to $E'$. Now, $G$ contains a vertex cover of size at most $|V(G)| - \budget$ if and only if $G'$ contains a cut of size at least $4|E(G)| + \budget$. 
 
 Let $\modulator$ be a modulator to treewidth at most 2 for $G$, i.e., $\tw(G - \modulator) \leq 2$. We claim that $\modulator' = \modulator \cup \{x\}$ is a modulator to treewidth at most 2 for $G'$, i.e., $\tw(G' - \modulator') \leq 2$. After removing $\modulator'$ only edges of the form $\{e_u, e_v\}$, $\{e_u, u\}$, and $\{e_v, v\}$ remain. Notice that $G' - \modulator'$ can be obtained from $G - \modulator$ by subdividing each edge twice, i.e., replacing each edge by a path of length 3. It is well-known that subdividing edges does not affect the treewidth, hence we have that $\tw(G' - \modulator') = \tw(G - \modulator) \leq 2$. 
 
 If we can solve \MC in time $\Oh^*((2 - \eps)^{|\modulator|})$ for some $\eps > 0$, where $\modulator$ is a modulator to treewidth at most 2, then we can also solve \VC in time $\Oh^*((2-\eps)^{|\modulator|})$ by using the discussed polynomial-time reduction and noting that any modulator to pathwidth at most 2 is also a modulator to treewidth at most 2. Hence, \SETH must be false by \cref{thm:vc_lb}.
\end{proof}

\begin{cor}
 If \MC can be solved in time $\Oh^*((2 - \eps)^{\td(G)})$ for some $\eps > 0$, then \SETH is false.
\end{cor}

\begin{proof}
 Follows from \cref{thm:mc_lb} and \cref{thm:mod_to_tw_implies_td}.
\end{proof}

\begin{cor}
 \label{thm:clique_del_lb}
  Let $r \geq 3$. If \KFD can be solved in time $\Oh^*((2 - \eps)^{|\modulator|})$ for some $\eps > 0$, where $\modulator$ is a modulator to treewidth at most $r - 1$, then \SETH is false.
\end{cor}

\begin{proof} 
 Let $G = (V, E)$ be a \VC instance. We construct $G' = (V', E')$ by using a classical trick and replacing each edge of $G$ by a clique on $r$ vertices which is just like the deletion edges for the \DTC lower bounds. It can be seen that $G$ has a vertex cover of size at most $\budget$ if and only if there is a set $\cover \subseteq V'$, $|\cover| \leq \budget$, such that $G' - \cover$ is $K_r$-free. 
 
 If $\modulator$ is a modulator to treewidth at most 2 for $G$, i.e., $\tw(G - \modulator) \leq 2$, then we claim that $\modulator$ is also a modulator to treewidth at most $r - 1$ for $G'$. Let $\TT$ be a tree decomposition of width 2 for $G - \modulator$ with bags $\bag_t$, $t \in V(\TT)$. We construct a tree decomposition $\TT'$ of width $r - 1$ for $G' - \modulator$ as follows. We start with the tree decomposition $\TT$ of $G - \modulator$. For every edge $e \in E(G - \modulator)$ there is a $t \in V(\TT)$ such that $e \subseteq \bag_{t}$ and to $t$ we attach a degree-1 node $t^*$ and the bag of $t^*$ contains the $r$-clique that replaced edge $e$. Hence, every edge of $G' - \modulator$ is captured in some bag of $\TT'$ and $\TT'$ is a tree decomposition of width $r - 1$.
 
 Since the construction of $G'$ can be performed in polynomial time and since the size of the modulator does not change, the result follows by \cref{thm:vc_lb}.
\end{proof}

\begin{cor}
 \label{thm:clique_del_td_lb}
 Let $r \geq 3$. If \textsc{$K_r$-free Deletion} can be solved in time $\Oh^*((2 - \eps)^{\td(G)})$ for some $\eps > 0$, then \SETH is false.
\end{cor}

\begin{proof}
 Follows from \cref{thm:clique_del_lb} and \cref{thm:mod_to_tw_implies_td}.
\end{proof}

The lower bounds obtained for \VC, \MC, and \textsc{$K_r$-free Deletion} are tight in the sense that we have algorithms with matching running time by straightforward dynamic programming on tree decompositions.

\section{Lower Bound for Dominating Set}
\label{sec:ds_lb}

In the sparse setting, the optimal base, assuming \SETH, for \DS parameterized by treewidth is 3 based on results by Lokshtanov et al.~\cite{LokshtanovMS18} and van Rooij et al.~\cite{RooijBR09}. Van Geffen et al.~\cite{GeffenJKM20} have shown that base 3 is already optimal for the much more restrictive parameter cutwidth by slightly modifying the construction of Lokshtanov et al.~\cite{LokshtanovMS18}. 

For the dense setting, the optimal base for \DS parameterized by clique-width is 4 based on results by Bodlaender el al.~\cite{BodlaenderLRV10} and Katsikarelis et al.~\cite{KatsikarelisLP19}. The lower bound of Katsikarelis et al., as is usual, already applies to the more restrictive linear-clique-width. We improve upon this lower bound by showing that similar to the sparse setting, already a simple dense variant of cutwidth, i.e. twinclass-cutwidth, requires base 4. 

We provide a straightforward tight lower bound for \DS parameterized by twinclass-cutwidth by observing that solving \DS on $G$ essentially means solving \TDS on $G / \tcpartition(G)$, assuming that there are no twinclasses of size 1, and providing a lower bound for \TDS parameterized by cutwidth based on the methods of Lokshtanov et al.~\cite{LokshtanovMS18} and van Geffen et al.~\cite{GeffenJKM20}. 

For the remainder of this section assume without loss of generality that $G$ is connected and consists of at least two vertices. We define the graph $G'$ by $V(G') = V(G) \cup \{v' \sep v \in V(G)\}$ and $E(G') = E(G) \cup \{\{u,v'\},\{u',v\},\{u',v'\} \sep \{u,v\} \in E(G)\}$, i.e., we obtain $G'$ from $G$ by creating a false twin of each vertex, hence $G'$ has no twinclass of size 1.

\begin{lem}
  \label{thm:twins_tcctw}
 It holds that $\tcctw(G') \leq \ctw(G)$.
\end{lem}

\begin{proof}
  We begin by showing that $u \neq v \in V(G)$ are twins in $G'$ if and only if $u$ and $v$ are false twins in $G$. First, suppose that $\{u,v\} \in E(G)$. By construction of $G'$, this implies that $u' \in N_{G'}(v) \setminus N_{G'}(u)$ and hence $u$ and $v$ cannot be twins in $G'$ in this case. If $u$ and $v$ are false twins in $G'$, then $N_{G}(u) = N_{G'}(u) \cap V(G) = N_{G'}(v) \cap V(G) = N_{G}(v)$, so $u$ and $v$ are false twins in $G$ as well. If $u$ and $v$ are false twins in $G$, then $N_{G'}(u) = N_G(u) \cup \{w' \sep w \in N_G(u)\} = N_G(v) \cup \{w' \sep w \in N_G(v)\} = N_{G'}(v)$, so $u$ and $v$ are false twins in $G'$ as well. 
  
  By construction of $G'$, the two vertices $v$ and $v'$ are false twins in $G'$ for all $v \in V(G)$. Together with the previous argument and since the twin relation is an equivalence relation, this shows that every twinclass $U'$ of $G'$ has the form $U' = U \cup \{v' \sep v \in U\}$, where $U$ is a false twinclass of $G$. For each false twinclass $U$ of $G$, pick a vertex $v_U \in U$ and consider the map $\tcpartition(G') \ni U' = U \cup \{v' \sep v \in U\} \mapsto v_U \in U \subseteq V(G)$. If $U'$ and $W'$ are adjacent twinclasses of $G'$, then the vertices $v_U$ and $v_W$ have to be adjacent in $G$. Hence, we can view $G' / \tcpartition(G')$ as a subgraph of $G$ and see that $\tcctw(G') = \ctw(G' / \tcpartition(G')) \leq \ctw(G)$.   
\end{proof}

\begin{lem}
  \label{thm:ds_iff_tds}
 $G'$ has a dominating set of size at most $\budget$ if and only if $G$ has a total dominating set of size at most $\budget$. 
\end{lem}

\begin{proof}
  $\Rightarrow:$ Let $X'$ be a dominating set of $G'$. If $|X' \cap \{v, v'\}| = 1$ for some $v \in V(G)$, then we can assume without loss of generality that $v \in X'$. Suppose that $|X' \cap \{v,v'\}| = 2$ for some $v \in V(G)$. If $X' \cap N_{G'}(v) = X' \cap N_{G'}(v') = \emptyset$, then choose some $w \in N_{G'}(v) \cap V(G) = N_{G'}(v') \cap V(G)$, which exists due to our assumptions about $G$, and consider $X'' = (X' \setminus \{v'\}) \cup \{w\}$. Otherwise, consider $X'' = X' \setminus \{v'\}$. In either case, $|X''| \leq |X'|$ and we claim that $X''$ is still a dominating set of $G'$. Only vertices inside $N_{G'}[v']$ may not be dominated anymore. The vertex $v'$ is still dominated, because there is some $w \in X'' \cap N_{G'}(v) = X'' \cap N_{G'}(v')$. The vertices in $N_{G'}(v')$ are still dominated, because $v \in X''$ and $N_{G'}(v) = N_{G'}(v')$. Hence, $X''$ is a dominating set of $G'$.
  
  By repeating this argument, we obtain a dominating set $X$ of $G'$ with $|X| \leq |X'|$ and $X \subseteq V(G)$. For every $v \in V(G)$, we have that $v' \notin X$ and hence, there exists some $w \in X \cap N_{G'}(v') = X \cap N_{G'}(v) = X \cap N_G(v)$, where the last equality is due to $X \subseteq V(G)$. This shows that $X$ is a total dominating set of $G$.
  
  $\Leftarrow:$ Let $X$ be a total dominating set of $G$. We claim that $X$ is also a dominating set of $G'$. As $G'[V(G)] = G$, any vertex in $V(G)$ is still dominated by $X$ inside $G'$. Consider some vertex $v'$ for some $v \in V(G)$. Since $X$ is a total dominating set of $G$, there exists some $w \in X \cap N_G(v) = X \cap N_{G'}(v) = X \cap N_{G'}(v')$ and hence $v'$ is dominated by $X$ as well. So, $X$ must be a dominating set of $G'$.
\end{proof}

\begin{mythm}
  If there is an algorithm $\algo$ solving \DS in time $\Oh^*((4 - \eps)^{\tcctw(G)})$ for some $\eps > 0$, then there is an algorithm $\algo'$ solving \TDS in time $\Oh^*((4 - \eps)^{\ctw(G)})$.
\end{mythm}

\begin{proof}
  Suppose $\algo$ is such an algorithm. The algorithm $\algo'$ first constructs $G'$ in polynomial time, then runs $\algo$ on $G'$, which takes time $\Oh^*((4 - \eps)^{\tcctw(G')}) \leq \Oh^*((4 - \eps)^{\ctw(G)})$ by \cref{thm:twins_tcctw}. Finally $\algo'$ returns the same answer as $\algo$. The correctness of $\algo'$ follows from \cref{thm:ds_iff_tds} and the running time of $\algo'$ is dominated by running $\algo$ on $G'$.
\end{proof}

We remark that an $\Oh^*(4^{\tw(G)})$-time algorithm for \TDS is posed as an exercise in \cite{CyganFKLMPPS15} which implies an $\Oh^*(4^{\ctw(G)})$-time algorithm by \cref{thm:td_tw}. It remains to show that \TDS cannot be solved in time $\Oh^*((4-\eps)^{\ctw(G)})$ for some $\eps > 0$. We show this by adapting the lower bound construction for \DS given by Lokshtanov et al.~\cite{LokshtanovMS18} and the observation of van Geffen et al.~\cite{GeffenJKM20} that slightly changing the construction already yields the same lower bound parameterized by cutwidth. Since 4 is a power of two, some technicalities regarding the groups can be avoided, allowing for a simpler construction.

\begin{mythm}
  \label{thm:tds_lb}
  There is no algorithm solving \TDS in time $\Oh^*((4 - \eps)^{\ctw(G)})$ for some $\eps > 0$, unless \SETH is false.
\end{mythm}

\subparagraph*{Construction.} Let $\formula$ be a $\clss$-\SAT instance with $\nvars$ variables and $\nclss$ clauses. We begin by describing how to construct a graph $G = G(\formula)$ of cutwidth $\nvars/2 + 4\clss + \Oh(1)$. We assume without loss of generality that $\formula$ has an even number of variables and partition the variables into pairs, so that pair $i \in [\nvars/2]$ consists of variables $x_{2i - 1}$ and $x_{2i}$.

\begin{figure}[h]
 \centering
 \begin{tikzpicture}
	\begin{pgfonlayer}{nodelayer}
		\node [style=small empty] (0) at (-8, 0) {$p_{i,1}^\ell$};
		\node [style=small empty] (1) at (-5, 0) {$p_{i,2}^\ell$};
		\node [style=small empty] (2) at (-2, 0) {$p_{i,3}^\ell$};
		\node [style=small empty] (3) at (1, 0) {$p_{i,4}^\ell$};
		\node [style=small empty] (4) at (-2, 3) {$q_{i,1}^\ell$};
		\node [style=small empty] (5) at (-5, 3) {$q_{i,2}^\ell$};
		\node [style=small empty] (6) at (-5, -3) {$q_{i,4}^\ell$};
		\node [style=small empty] (7) at (-2, -3) {$q_{i,3}^\ell$};
	\end{pgfonlayer}
	\begin{pgfonlayer}{edgelayer}
		\draw [style=thick] (0) to (1);
		\draw [style=thick] (1) to (2);
		\draw [style=thick] (2) to (3);
		\draw [style=thick] (5) to (0);
		\draw [style=thick] (5) to (2);
		\draw [style=thick] (5) to (3);
		\draw [style=thick] (4) to (1);
		\draw [style=thick] (4) to (2);
		\draw [style=thick] (4) to (3);
		\draw [style=thick] (6) to (0);
		\draw [style=thick] (6) to (1);
		\draw [style=thick] (6) to (2);
		\draw [style=thick] (0) to (7);
		\draw [style=thick] (1) to (7);
		\draw [style=thick] (7) to (3);
	\end{pgfonlayer}
\end{tikzpicture}
 \caption{The path segment $P_i^\ell$ and guards $Q_i^\ell$.}
 \label{fig:tds_path}
\end{figure}
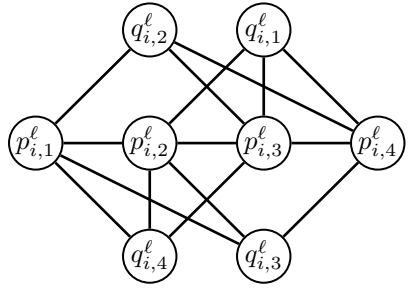

For every pair $i \in [\nvars/2]$, we create a path $P_i$ consisting of $4 \nclss (3 \frac{\nvars}{2} + 1)$ vertices and partition the path into $\nclss (3 \frac{\nvars}{2} + 1)$ segments consisting of four vertices. Segment $\ell \in [\nclss(3 \frac{\nvars}{2} + 1)]$ consists of the vertices $P_i^\ell = \{p_{i,1}^\ell, p_{i,2}^\ell, p_{i,3}^\ell, p_{i,4}^\ell\}$. To each segment $\ell$, we attach four \emph{guard vertices} $\guards_i^\ell = \{q_{i,1}^\ell, \ldots, q_{i,4}^\ell\}$ such that $N(q_{i,j}^\ell) = P_i^\ell \setminus \{p_{i,j}^\ell\}$ for all $j \in [4]$, i.e., for every vertex on $P_i^\ell$ there is one guard that avoids this vertex and is only adjacent to all other vertices on $P_i^\ell$, see \cref{fig:tds_path}. 

\begin{figure}[h]
 \centering
 \begin{tikzpicture}
	\begin{pgfonlayer}{nodelayer}
		\node [style=tiny filled] (0) at (-7, 0) {};
		\node [style=tiny filled] (1) at (-5, 0) {};
		\node [style=tiny empty] (2) at (-3, 0) {};
		\node [style=tiny empty] (3) at (-1, 0) {};
		\node [style=tiny empty] (4) at (-3, 2) {};
		\node [style=tiny empty] (5) at (-5, 2) {};
		\node [style=tiny empty] (6) at (-5, -2) {};
		\node [style=tiny empty] (7) at (-3, -2) {};
		\node [style=tiny filled] (8) at (0, 0) {};
		\node [style=tiny empty] (9) at (2, 0) {};
		\node [style=tiny empty] (10) at (4, 0) {};
		\node [style=tiny filled] (11) at (6, 0) {};
		\node [style=tiny empty] (12) at (4, 2) {};
		\node [style=tiny empty] (13) at (2, 2) {};
		\node [style=tiny empty] (14) at (2, -2) {};
		\node [style=tiny empty] (15) at (4, -2) {};
		\node [style=tiny empty] (16) at (7, 0) {};
		\node [style=tiny filled] (17) at (9, 0) {};
		\node [style=tiny filled] (18) at (11, 0) {};
		\node [style=tiny empty] (19) at (13, 0) {};
		\node [style=tiny empty] (20) at (11, 2) {};
		\node [style=tiny empty] (21) at (9, 2) {};
		\node [style=tiny empty] (22) at (9, -2) {};
		\node [style=tiny empty] (23) at (11, -2) {};
		\node [style=tiny empty] (24) at (14, 0) {};
		\node [style=tiny empty] (25) at (16, 0) {};
		\node [style=tiny filled] (26) at (18, 0) {};
		\node [style=tiny filled] (27) at (20, 0) {};
		\node [style=tiny empty] (28) at (18, 2) {};
		\node [style=tiny empty] (29) at (16, 2) {};
		\node [style=tiny empty] (30) at (16, -2) {};
		\node [style=tiny empty] (31) at (18, -2) {};
		\node [style=none] (32) at (-7, 2) {$\mathbf{1}\colon$};
		\node [style=none] (33) at (0, 2) {$\mathbf{2}\colon$};
		\node [style=none] (34) at (7, 2) {$\mathbf{3}\colon$};
		\node [style=none] (35) at (14, 2) {$\mathbf{4}\colon$};
	\end{pgfonlayer}
	\begin{pgfonlayer}{edgelayer}
		\draw [style=thick] (0) to (1);
		\draw [style=thick] (1) to (2);
		\draw [style=thick] (2) to (3);
		\draw [style=thick] (5) to (0);
		\draw [style=thick] (5) to (2);
		\draw [style=thick] (5) to (3);
		\draw [style=thick] (4) to (1);
		\draw [style=thick] (4) to (2);
		\draw [style=thick] (4) to (3);
		\draw [style=thick] (6) to (0);
		\draw [style=thick] (6) to (1);
		\draw [style=thick] (6) to (2);
		\draw [style=thick] (0) to (7);
		\draw [style=thick] (1) to (7);
		\draw [style=thick] (7) to (3);
		\draw [style=thick] (8) to (9);
		\draw [style=thick] (9) to (10);
		\draw [style=thick] (10) to (11);
		\draw [style=thick] (13) to (8);
		\draw [style=thick] (13) to (10);
		\draw [style=thick] (13) to (11);
		\draw [style=thick] (12) to (9);
		\draw [style=thick] (12) to (10);
		\draw [style=thick] (12) to (11);
		\draw [style=thick] (14) to (8);
		\draw [style=thick] (14) to (9);
		\draw [style=thick] (14) to (10);
		\draw [style=thick] (8) to (15);
		\draw [style=thick] (9) to (15);
		\draw [style=thick] (15) to (11);
		\draw [style=thick] (16) to (17);
		\draw [style=thick] (17) to (18);
		\draw [style=thick] (18) to (19);
		\draw [style=thick] (21) to (16);
		\draw [style=thick] (21) to (18);
		\draw [style=thick] (21) to (19);
		\draw [style=thick] (20) to (17);
		\draw [style=thick] (20) to (18);
		\draw [style=thick] (20) to (19);
		\draw [style=thick] (22) to (16);
		\draw [style=thick] (22) to (17);
		\draw [style=thick] (22) to (18);
		\draw [style=thick] (16) to (23);
		\draw [style=thick] (17) to (23);
		\draw [style=thick] (23) to (19);
		\draw [style=thick] (24) to (25);
		\draw [style=thick] (25) to (26);
		\draw [style=thick] (26) to (27);
		\draw [style=thick] (29) to (24);
		\draw [style=thick] (29) to (26);
		\draw [style=thick] (29) to (27);
		\draw [style=thick] (28) to (25);
		\draw [style=thick] (28) to (26);
		\draw [style=thick] (28) to (27);
		\draw [style=thick] (30) to (24);
		\draw [style=thick] (30) to (25);
		\draw [style=thick] (30) to (26);
		\draw [style=thick] (24) to (31);
		\draw [style=thick] (25) to (31);
		\draw [style=thick] (31) to (27);
	\end{pgfonlayer}
\end{tikzpicture}
 \caption{The four states $\one_i^\ell$, $\two_i^\ell$, $\three_i^\ell$, and $\four_i^\ell$ on the path segment $P_i^\ell$. The filled vertices belong to the corresponding set.}
 \label{fig:tds_states}
\end{figure}

Let $\one_i^\ell = \{p_{i,1}^\ell, p_{i,2}^\ell\}$, $\two_i^\ell = \{p_{i,1}^\ell, p_{i,4}^\ell\}$, $\three_i^\ell = \{p_{i,2}^\ell, p_{i,3}^\ell\}$, $\four_i^\ell = \{p_{i,3}^\ell, p_{i,4}^\ell\}$ and $\goodsets_i^\ell = \{\one_i^\ell, \two_i^\ell, \three_i^\ell, \four_i^\ell\}$, see \cref{fig:tds_states} for a depiction of these sets. For each segment $P_i^\ell$, we introduce four \emph{connector vertices} $\hat{z}_{i,S}^\ell$, $S \in \goodsets_i^\ell$ and four \emph{clique vertices} $z_{i,S}^\ell$, $S \in \goodsets_i^\ell$, together with two \emph{clique guards} $y_{i,1}^\ell$ and $y_{i,2}^\ell$. The connector vertex $\hat{z}_{i,S}^\ell$ is adjacent to $P_i^\ell \setminus S$ and the clique vertex $z_{i,S}^\ell$. The clique guards $y_{i,1}^\ell$ and $y_{i,2}^\ell$ are adjacent and the clique vertices $z_{i,S}^\ell$, $S \in \goodsets_i^\ell$, together with $y_{i,1}^\ell$ form a clique.

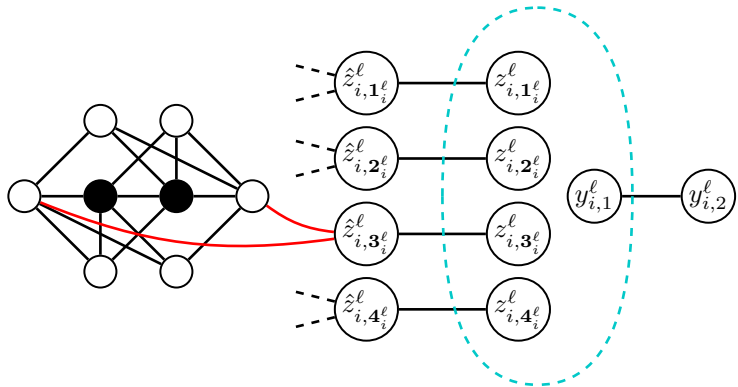
\begin{figure}[h]
 \centering
 \begin{tikzpicture}
	\begin{pgfonlayer}{nodelayer}
		\node [style=tiny empty] (16) at (-7, 0) {};
		\node [style=tiny filled] (17) at (-5, 0) {};
		\node [style=tiny filled] (18) at (-3, 0) {};
		\node [style=tiny empty] (19) at (-1, 0) {};
		\node [style=tiny empty] (20) at (-3, 2) {};
		\node [style=tiny empty] (21) at (-5, 2) {};
		\node [style=tiny empty] (22) at (-5, -2) {};
		\node [style=tiny empty] (23) at (-3, -2) {};
		\node [style=small empty] (24) at (2, 3) {$\hat{z}_{i, \mathbf{1}_i^\ell}^\ell$};
		\node [style=small empty] (25) at (2, 1) {$\hat{z}_{i, \mathbf{2}_i^\ell}^\ell$};
		\node [style=small empty] (26) at (2, -1) {$\hat{z}_{i, \mathbf{3}_i^\ell}^\ell$};
		\node [style=small empty] (27) at (2, -3) {$\hat{z}_{i, \mathbf{4}_i^\ell}^\ell$};
		\node [style=none] (28) at (0, 1.5) {};
		\node [style=none] (29) at (0, 0.5) {};
		\node [style=none] (30) at (0, 3.5) {};
		\node [style=none] (31) at (0, 2.5) {};
		\node [style=none] (32) at (0, -2.5) {};
		\node [style=none] (33) at (0, -3.5) {};
		\node [style=tiny empty] (34) at (6, 3) {$z_{i, \mathbf{1}_i^\ell}^\ell$};
		\node [style=tiny empty] (35) at (6, 1) {$z_{i, \mathbf{2}_i^\ell}^\ell$};
		\node [style=tiny empty] (36) at (6, -1) {$z_{i, \mathbf{3}_i^\ell}^\ell$};
		\node [style=tiny empty] (37) at (6, -3) {$z_{i, \mathbf{4}_i^\ell}^\ell$};
		\node [style=small empty] (38) at (8, 0) {$y_{i,1}^\ell$};
		\node [style=small empty] (39) at (11, 0) {$y_{i,2}^\ell$};
		\node [style=none] (40) at (6.5, 5) {};
		\node [style=none] (41) at (6.5, -5) {};
		\node [style=none] (42) at (4, 0) {};
		\node [style=none] (43) at (9, 0) {};
	\end{pgfonlayer}
	\begin{pgfonlayer}{edgelayer}
		\draw [style=thick] (16) to (17);
		\draw [style=thick] (17) to (18);
		\draw [style=thick] (18) to (19);
		\draw [style=thick] (21) to (16);
		\draw [style=thick] (21) to (18);
		\draw [style=thick] (21) to (19);
		\draw [style=thick] (20) to (17);
		\draw [style=thick] (20) to (18);
		\draw [style=thick] (20) to (19);
		\draw [style=thick] (22) to (16);
		\draw [style=thick] (22) to (17);
		\draw [style=thick] (22) to (18);
		\draw [style=thick] (16) to (23);
		\draw [style=thick] (17) to (23);
		\draw [style=thick] (23) to (19);
		\draw [style=thick, color=red, bend right=15] (19) to (26);
		\draw [style=thick, color=red, bend right=15] (16) to (26);
		\draw [style=dotted] (27) to (32.center);
		\draw [style=dotted] (27) to (33.center);
		\draw [style=dotted] (25) to (29.center);
		\draw [style=dotted] (25) to (28.center);
		\draw [style=dotted] (24) to (31.center);
		\draw [style=dotted] (24) to (30.center);
		\draw [style=thick] (24) to (34);
		\draw [style=thick] (25) to (35);
		\draw [style=thick] (26) to (36);
		\draw [style=thick] (27) to (37);
		\draw [style=thick] (38) to (39);
		\draw [style=dashed cyan thick, in=180, out=90] (42.center) to (40.center);
		\draw [style=dashed cyan thick, in=90, out=0] (40.center) to (43.center);
		\draw [style=dashed cyan thick, in=0, out=-90] (43.center) to (41.center);
		\draw [style=dashed cyan thick, in=180, out=-90] (42.center) to (41.center);
	\end{pgfonlayer}
\end{tikzpicture}
 \caption{The block gadget $B_i^\ell$. The vertices in the dashed ellipse form a clique. The vertex $\hat{z}_{i, \three_i^\ell}^\ell$ is adjacent to $P_i^\ell \setminus \three_i^\ell$.}
 \label{fig:tds_block}
\end{figure}

We define $Z_i^\ell = \{z_{i,S}^\ell, \hat{z}_{i,S}^\ell \sep S \in \goodsets_i^\ell\} \cup \{y_{i,1}^\ell, y_{i,2}^\ell\}$, $B_i^\ell = P_i^\ell \cup Q_i^\ell \cup Z_i^\ell$, and $R^\gamma = \bigcup_{i = 1}^{n/2} \bigcup_{\ell = (\gamma - 1)\nclss + 1}^{\gamma \nclss} B_i^\ell$ for $\gamma \in [3 \frac{\nvars}{2} + 1]$. The $B_i^\ell$ are called \emph{blocks}; the $R^\gamma$ are called \emph{regions}. See \cref{fig:tds_block} for a depiction of a single block. Each region can be visualized as a matrix of blocks with $\nvars / 2$ rows and $\nclss$ columns and we have $3 \frac{\nvars}{2} + 1$ regions in total. Furthermore, we say that $p^\ell_{i,1}$ is the \emph{entry vertex} of $B^\ell_i$ and $p^\ell_{i,4}$ is the \emph{exit vertex} of $B^\ell_i$. Notice that there is a matching of size $\nvars / 2$ between the exit vertices of column $\ell$ and the entry vertices of column $\ell + 1$.

Let $\states = \{\one, \two, \three, \four\}$ and fix a bijection $\embedding \colon \{0,1\}^2 \rightarrow \states$ mapping truth assignments of a pair of variables to possible states on a path segment. For each clause $C_j$, $j \in [\nclss]$, we introduce $3 \frac{\nvars}{2} + 1$ \emph{clause vertices} $\hat{c}_j^\gamma$, $\gamma \in [3 \frac{\nvars}{2} + 1]$, one per region. The vertex $\hat{c}_j^\gamma$ is adjacent to some vertices in the blocks $B_i^{(\gamma - 1) \nclss + j}$ for every $i \in [\nvars / 2]$. Namely, for every variable pair $i$ consider the truth assignments of this pair that satisfy clause $C_j$ and consider the states corresponding, via $\embedding$, to these truth assignments. Each of these states corresponds to some vertex $z_{i,S}^{(\gamma - 1) \nclss + j}$, $S \in \goodsets_i^{(\gamma - 1) \nclss + j}$, which we make adjacent to $\hat{c}_j^\gamma$.

Finally, we introduce eight \emph{endpoint guards} $h_t$, $h'_t$, $t \in [4]$, together with the edges $\{h_1, h_2\}$, $\{h_1, h_3\}$, $\{h_3, h_4\}$, $\{h'_1, h'_2\}$, $\{h'_1, h'_3\}$, $\{h'_3, h'_4\}$ and $h_1$ is made adjacent to $p_{i,1}^1$ for all $i \in [\nvars / 2]$ and $h'_1$ is made adjacent to $p_{i,4}^{\nclss (3 \frac{\nvars}{2} + 1)}$ for all $i \in [\nvars / 2]$. This concludes the construction of $G(\formula)$. See \cref{fig:tds_blockmatrix} for a depiction of the graph $G$ as a matrix of blocks.

\begin{figure}
 \centering
 \begin{tikzpicture}[scale=0.85]
	\begin{pgfonlayer}{nodelayer}
		\node [style=wide rectangle] (0) at (-8, 4) {$B_1^1$};
		\node [style=wide rectangle] (1) at (-8, 2) {$B_2^1$};
		\node [style=wide rectangle] (2) at (-8, -1) {$B_{n/2}^1$};
		\node [style=none] (3) at (-8, 0.75) {$\vdots$};
		\node [style=wide rectangle] (4) at (-5, 4) {$B_1^2$};
		\node [style=wide rectangle] (5) at (-5, 2) {$B_2^2$};
		\node [style=wide rectangle] (6) at (-5, -1) {$B_{n/2}^2$};
		\node [style=none] (7) at (-5, 0.75) {$\vdots$};
		\node [style=none] (8) at (-3, 4) {$\cdots$};
		\node [style=none] (9) at (-3, 2) {$\cdots$};
		\node [style=none] (10) at (-3, -1) {$\cdots$};
		\node [style=none] (11) at (-3, 0.75) {$\ddots$};
		\node [style=wide rectangle] (12) at (-1, 4) {$B_1^m$};
		\node [style=wide rectangle] (13) at (-1, 2) {$B_2^m$};
		\node [style=wide rectangle] (14) at (-1, -1) {$B_{n/2}^m$};
		\node [style=none] (15) at (-1, 0.75) {$\vdots$};
		\node [style=none] (16) at (-3.75, 4) {};
		\node [style=none] (17) at (-3.75, 2) {};
		\node [style=none] (18) at (-3.75, -1) {};
		\node [style=none] (19) at (-2.25, -1) {};
		\node [style=none] (20) at (-2.25, 2) {};
		\node [style=none] (21) at (-2.25, 4) {};
		\node [style=wide rectangle] (22) at (3, 4) {$B_1^{m + 1}$};
		\node [style=wide rectangle] (23) at (3, 2) {$B_2^{m + 1}$};
		\node [style=wide rectangle] (24) at (3, -1) {$B_{n/2}^{m + 1}$};
		\node [style=none] (25) at (3, 0.75) {$\vdots$};
		\node [style=wide rectangle] (26) at (6, 4) {$B_1^{m + 2}$};
		\node [style=wide rectangle] (27) at (6, 2) {$B_2^{m + 2}$};
		\node [style=wide rectangle] (28) at (6, -1) {$B_{n/2}^{m + 2}$};
		\node [style=none] (29) at (6, 0.75) {$\vdots$};
		\node [style=none] (30) at (8, 4) {$\cdots$};
		\node [style=none] (31) at (8, 2) {$\cdots$};
		\node [style=none] (32) at (8, -1) {$\cdots$};
		\node [style=none] (33) at (8, 0.75) {$\ddots$};
		\node [style=wide rectangle] (34) at (10, 4) {$B_1^{2m}$};
		\node [style=wide rectangle] (35) at (10, 2) {$B_2^{2m}$};
		\node [style=wide rectangle] (36) at (10, -1) {$B_{n/2}^{2m}$};
		\node [style=none] (37) at (10, 0.75) {$\vdots$};
		\node [style=none] (38) at (7.25, 4) {};
		\node [style=none] (39) at (7.25, 2) {};
		\node [style=none] (40) at (7.25, -1) {};
		\node [style=none] (41) at (8.75, -1) {};
		\node [style=none] (42) at (8.75, 2) {};
		\node [style=none] (43) at (8.75, 4) {};
		\node [style=none] (44) at (13, 4) {$\cdots$};
		\node [style=none] (45) at (13, 2) {$\cdots$};
		\node [style=none] (46) at (13, -1) {$\cdots$};
		\node [style=none] (47) at (13, 0.75) {$\ddots$};
		\node [style=none] (48) at (12.25, 4) {};
		\node [style=none] (49) at (12.25, 2) {};
		\node [style=none] (50) at (12.25, -1) {};
		\node [style=none] (51) at (13.75, -1) {};
		\node [style=none] (52) at (13.75, 2) {};
		\node [style=none] (53) at (13.75, 4) {};
		\node [style=wide rectangle] (54) at (16, 4) {$B_1^{n_c m}$};
		\node [style=wide rectangle] (55) at (16, 2) {$B_2^{n_c m}$};
		\node [style=wide rectangle] (56) at (16, -1) {$B_{n/2}^{n_c m}$};
		\node [style=none] (57) at (16, 0.75) {$\vdots$};
		\node [style=small empty] (58) at (-11, -3) {$h_1$};
		\node [style=small empty] (60) at (-11, -5) {$h_2$};
		\node [style=none] (62) at (-9.5, 5) {};
		\node [style=none] (63) at (0.5, 5) {};
		\node [style=none] (64) at (0.5, -2) {};
		\node [style=none] (65) at (-9.5, -2) {};
		\node [style=none] (66) at (1.5, 5) {};
		\node [style=none] (67) at (11.5, 5) {};
		\node [style=none] (68) at (11.5, -2) {};
		\node [style=none] (69) at (1.5, -2) {};
		\node [style=none] (70) at (-9, 5.5) {$R^1$};
		\node [style=none] (71) at (2, 5.5) {$R^2$};
		\node [style=small empty] (72) at (-8, 8) {$\hat{c}_1^1$};
		\node [style=small empty] (73) at (-5, 8) {$\hat{c}_2^1$};
		\node [style=small empty] (74) at (-1, 8) {$\hat{c}_m^1$};
		\node [style=small empty] (75) at (3, 8) {$\hat{c}_1^2$};
		\node [style=small empty] (76) at (6, 8) {$\hat{c}_2^2$};
		\node [style=small empty] (77) at (10, 8) {$\hat{c}_m^2$};
		\node [style=none] (78) at (13, 8) {$\cdots$};
		\node [style=small empty] (79) at (-9, -3) {$h_3$};
		\node [style=small empty] (80) at (-9, -5) {$h_4$};
		\node [style=small empty] (81) at (19, -3) {$h'_1$};
		\node [style=small empty] (82) at (19, -5) {$h'_2$};
		\node [style=small empty] (83) at (17, -3) {$h'_3$};
		\node [style=small empty] (84) at (17, -5) {$h'_4$};
		\node [style=none] (85) at (-3, 8) {$\cdots$};
		\node [style=none] (86) at (8, 8) {$\cdots$};
		\node [style=small empty] (87) at (16, 8) {$\hat{c}_m^{n_c}$};
		\node [style=none] (88) at (5, -4) {$n_c = 3\frac{n}{2} + 1$};
	\end{pgfonlayer}
	\begin{pgfonlayer}{edgelayer}
		\draw [style=thick] (0) to (4);
		\draw [style=thick] (1) to (5);
		\draw [style=thick] (2) to (6);
		\draw [style=thick] (16.center) to (4);
		\draw [style=thick] (5) to (17.center);
		\draw [style=thick] (6) to (18.center);
		\draw [style=thick] (19.center) to (14);
		\draw [style=thick] (20.center) to (13);
		\draw [style=thick] (21.center) to (12);
		\draw [style=thick] (22) to (26);
		\draw [style=thick] (23) to (27);
		\draw [style=thick] (24) to (28);
		\draw [style=thick] (38.center) to (26);
		\draw [style=thick] (27) to (39.center);
		\draw [style=thick] (28) to (40.center);
		\draw [style=thick] (41.center) to (36);
		\draw [style=thick] (42.center) to (35);
		\draw [style=thick] (43.center) to (34);
		\draw [style=thick] (12) to (22);
		\draw [style=thick] (13) to (23);
		\draw [style=thick] (14) to (24);
		\draw [style=thick] (34) to (48.center);
		\draw [style=thick] (35) to (49.center);
		\draw [style=thick] (36) to (50.center);
		\draw [style=thick] (53.center) to (54);
		\draw [style=thick, in=180, out=0] (52.center) to (55);
		\draw [style=thick] (51.center) to (56);
		\draw [style=thick] (60) to (58);
		\draw [style=dashed cyan thick] (62.center) to (63.center);
		\draw [style=dashed cyan thick] (63.center) to (64.center);
		\draw [style=dashed cyan thick] (64.center) to (65.center);
		\draw [style=dashed cyan thick] (65.center) to (62.center);
		\draw [style=dashed cyan thick] (66.center) to (67.center);
		\draw [style=dashed cyan thick] (67.center) to (68.center);
		\draw [style=dashed cyan thick] (68.center) to (69.center);
		\draw [style=dashed cyan thick] (69.center) to (66.center);
		\draw [style=thick] (72) to (0);
		\draw [style=thick, in=45, out=-60, looseness=0.75] (72) to (2);
		\draw [style=thick] (73) to (4);
		\draw [style=thick, in=45, out=-45] (73) to (5);
		\draw [style=thick, bend left=45] (74) to (13);
		\draw [style=thick, in=60, out=-45, looseness=1.25] (74) to (14);
		\draw [style=thick] (75) to (22);
		\draw [style=thick, in=45, out=-60, looseness=0.75] (75) to (24);
		\draw [style=thick] (76) to (26);
		\draw [style=thick, in=45, out=-45] (76) to (27);
		\draw [style=thick, bend left=45] (77) to (35);
		\draw [style=thick, in=45, out=-45] (77) to (36);
		\draw [style=thick] (80) to (79);
		\draw [style=thick] (58) to (79);
		\draw [style=thick, in=-180, out=90] (58) to (0);
		\draw [style=thick, in=180, out=75] (58) to (1);
		\draw [style=thick, in=180, out=60] (58) to (2);
		\draw [style=thick] (82) to (81);
		\draw [style=thick] (84) to (83);
		\draw [style=thick] (81) to (83);
		\draw [style=thick, in=90, out=0] (54) to (81);
		\draw [style=thick, in=0, out=105, looseness=0.75] (81) to (55);
		\draw [style=thick, in=0, out=120] (81) to (56);
		\draw [style=thick, bend left=45] (87) to (55);
		\draw [style=thick, bend left=45] (87) to (56);
	\end{pgfonlayer}
\end{tikzpicture}
 \caption{$G$ viewed as a matrix of block gadgets. The dashed rectangles delineate different regions.}
 \label{fig:tds_blockmatrix}
\end{figure}

\begin{lem}
 \label{thm:sat_to_tds}
 If $\formula$ is satisfiable, then $G(\formula)$ has a total dominating set of size at most $4 \nclss (3 \frac{\nvars}{2} + 1) \frac{\nvars}{2} + 4$.
\end{lem}

\begin{proof}
 Let $\tassign$ be a satisfying truth assignment of $\formula$. Let $\tassign_i$ be the restriction of $\tassign$ to the $i$-th variable pair. We construct a total dominating set $X$ of the desired size as follows. For every variable pair $i \in [\frac{\nvars}{2}]$, we take in each block $B_i^\ell$ the vertices in $\embedding(\tassign_i)_i^\ell \in \goodsets_i^\ell$, the vertex $z_{i, \embedding(\tassign_i)_i^\ell}^\ell$, and the clique guard $y_{i,1}^\ell$ for all $\ell \in [\nclss (3 \frac{\nvars}{2} + 1)]$. Additionally, we take the endpoint guards $h_1$, $h_3$, $h'_1$, $h'_3$. This concludes the description of $X$ and we see that $|X| = 4 \nclss (3 \frac{\nvars}{2} + 1) \frac{\nvars}{2} + 4$.
 
 First, observe that $X$ dominates all the endpoint guards $h_t$, $h'_t$, $t \in [4]$, and the endpoints of all $P_i$ are dominated, since $h_1, h'_1 \in X$. Next, since $X$ contains two vertices from each path segment $P_i^\ell$, we see that $X$ intersects the neighborhoods of all vertices in $Q_i^\ell$ and hence dominates them. On each path $P_i$, we are repeatedly taking two consecutive vertices with a gap of two, except for possibly the ends of $P_i$. Hence, also all internal vertices of $P_i$ are dominated.
 
 Consider the vertices in the $Z_i^\ell$ next. Since $y_{i,1}^\ell \in X$, we see that $y_{i,2}^\ell$ and all $z_{i,S}^\ell$, $S \in \goodsets_i^\ell$, are dominated. Additionally, since $z_{i,\embedding(\tassign_i)_i^\ell}^\ell \in X$, also $y_{i,1}^\ell$ and $\hat{z}_{i,\embedding(\tassign_i)_i^\ell}^\ell$ must be dominated. For $S \in \goodsets_i^\ell \setminus \{\embedding(\tassign_i)_i^\ell\}$, we see that $\hat{z}_{i,S}^\ell$ is dominated by some vertex on $P_i^\ell$, because $N(\hat{z}_{i,S}^\ell) \cap X = (P_i^\ell \setminus S) \cap \embedding(\tassign_i)_i^\ell \neq \emptyset$ by construction.
 
 Finally, consider the clause vertices $\hat{c}_j^\gamma$. Since clause $C_j$ is satisfied by $\tassign$, there is some variable pair $i$ so that $\tassign_i$ satisfies $C_j$. By construction of $G(\formula)$ and $X$, we have that $z_{i,\embedding(\tassign_i)_i^\ell}^{(\gamma - 1)\nclss + j} \in X$ and this vertex is adjacent to $\hat{c}_j^\gamma$, so $\hat{c}_j^\gamma$ is dominated. This concludes the proof.
\end{proof}

For the reverse direction, we first study the structure of total dominating sets in $G(\formula)$.

\begin{lem}
  \label{thm:tds_structure}
  Let $X$ be a total dominating set of $G(\formula)$, then for any block $B_i^\ell = P_i^\ell \cup Q_i^\ell \cup Z_i^\ell$, we have that $|X \cap P_i^\ell| \geq 2$, $|X \cap Z_i^\ell| \geq 2$, and $y_{i,1}^\ell \in X$. Moreover, if $|X \cap B_i^\ell| \leq 4$, then $X \cap P_i^\ell \in \goodsets_i^\ell$ and $X \cap \{z_{i,S}^\ell \sep S \in \goodsets_i^\ell\} = \{z_{i, X \cap P_i^\ell}^\ell\}$.
\end{lem}

\begin{proof}
  Suppose that $|X \cap P_i^\ell| \leq 1$, then there is some $j \in [4]$ so that $N(q_{i,j}^\ell) \cap X = \emptyset$, hence $X$ cannot be a total dominating set. Observe that $y_{i,2}^\ell$ is a degree-1 vertex, so $X$ must contain its neighbor $y_{i,1}^\ell$. To dominate $y_{i,1}^\ell \in Z_i^\ell$, $X$ must contain a vertex from $N(y_{i,1}^\ell) \subseteq Z_i^\ell$, recall that $Z_i^\ell$ also contains $y_{i,2}^\ell$.
  
  Now, suppose that the second condition holds. Due to the first part of this lemma, we have that $|X \cap B_i^\ell| = 4$, in particular $|X \cap P_i^\ell| = 2$, $X \cap Q_i^\ell = \emptyset$, and $X \cap \{\hat{z}_{i,S}^\ell \sep S \in \goodsets_i^\ell\} = \emptyset$, because $N(y_{i,1}^\ell) \subseteq Z_i^\ell \setminus \{\hat{z}_{i,S}^\ell \sep S \in \goodsets_i^\ell\}$. Suppose that $X \cap P_i^\ell \notin \goodsets_i^\ell$, then $X \cap P_i^\ell = \{p_{i,1}^\ell, p_{i,3}^\ell\}$ or $X \cap P_i^\ell = \{p_{i,2}^\ell, p_{i,4}^\ell\}$, but then $p_{i,3}^\ell$ or respectively $p_{i,2}^\ell$ is not dominated. By the established properties of $X$, the only neighbor of $\hat{z}_{i, X \cap P_i^\ell}^\ell$ that may be in $X$ is $z_{i, X \cap P_i^\ell}^\ell$, hence to dominate $\hat{z}_{i, X \cap P_i^\ell}^\ell$ we must have $z_{i, X \cap P_i^\ell}^\ell \in X$. This concludes the proof.
\end{proof}

A total dominating set $X$ of $G(\formula)$ does not necessarily repeat the same state on all segments $P_i^\ell$ of a path $P_i$. We need to study how $X$ transitions between different states on the same path. For this sake, we order the states according to their names, i.e.\ $\one < \two < \three < \four$, and every $\goodsets_i^\ell$ inherits this ordering. If $\mathbf{s}_i^\ell \in \goodsets_i^\ell = \{\one_i^\ell, \two_i^\ell, \three_i^\ell, \four_i^\ell\}$, then we let $\mathbf{s} \in \{\one, \two, \three, \four\}$ correspondingly.

\begin{lem}
  \label{thm:tds_state_order}
  Suppose that $X$ is a total dominating set of $G(\formula)$ with $|X \cap B_i^\ell| \leq 4$ for all $i$ and $\ell$. If $\mathbf{s}_i^\ell = X \cap P_i^\ell$ is the state of $X$ on $P_i^\ell$ and $\mathbf{\tilde{s}}_i^{\ell + 1} = X \cap P_i^{\ell + 1}$ is the state of $X$ on $P_i^{\ell + 1}$, then we must have that $\mathbf{\tilde{s}} \leq \mathbf{s}$.
\end{lem}

\begin{proof}
  Suppose that $\mathbf{\tilde{s}} > \mathbf{s}$. This implies that $\mathbf{s} \neq \four$. If $\mathbf{s} \leq \two$, then $p_{i,4}^\ell$ must be dominated by $p_{i,1}^{\ell + 1}$, so $\mathbf{\tilde{s}} \leq \two$ in this case, as the other states do not contain $p_{i,1}^{\ell + 1}$. If $(\mathbf{s, \tilde{s}}) = (\one, \two)$ or $(\mathbf{s, \tilde{s}}) = (\three, \four)$, then $p_{i,1}^{\ell + 1}$ is not dominated. This concludes the proof.
\end{proof}

\begin{lem}
 \label{thm:tds_to_sat}
 If $G(\formula)$ has a total dominating size $X$ of size at most $4 \nclss (3 \frac{\nvars}{2} + 1) \frac{\nvars}{2} + 4$, then $\formula$ is satisfiable.
\end{lem}

\begin{proof}
  Note that $X$ contains the vertices $h_1$, $h_3$, $h'_1$, $h'_3$ due to the degree-1 neighbors. There are in total $\nclss (3 \frac{\nvars}{2} + 1) \frac{\nvars}{2}$ blocks $B_i^\ell$ and $X$ has to contain at least four vertices in each block due to \cref{thm:tds_structure}. This completely uses the budget for $X$ and thus we have $|X \cap B_i^\ell| = 4$ for all $i$ and $\ell$. By \cref{thm:tds_structure}, $X$ assumes one of our desired states on each $B_i^\ell$. 
 
 Using the notation of \cref{thm:tds_state_order}, we say that a \emph{cheat} occurs if $\mathbf{\tilde{s}} < \mathbf{s}$. Due to \cref{thm:tds_state_order}, there can be at most $3$ cheats on each $P_i$. Since there are $\frac{\nvars}{2}$ paths, at most $3 \frac{\nvars}{2}$ regions can be spoiled by cheats. Hence, there is at least one $\gamma$ so that $R^\gamma$ does not contain a cheat. Fix this $\gamma$ for the remainder of the proof.
 
 We will read off a satisfying truth assignment $\tassign$ for $\formula$ from region $R^\gamma$. The truth assignment for variable pair $i$ is given by $\tassign_i = \embedding^{-1}(\mathbf{s})$, where $\mathbf{s}_i^{\gamma \nclss} = X \cap P_i^{\gamma \nclss} \in \goodsets_i^{\gamma \nclss}$. For a fixed variable pair $i \in [\nvars / 2]$, the total dominating set $X$ will assume the same state on all blocks $B_i^{(\gamma - 1)\nclss + j}$, $j \in [\nclss]$, since $R^\gamma$ does not contain any cheats. The desired truth assignment is given by $\tassign = \tassign_1 \cup \cdots \cup \tassign_{\nvars/2}$.
 
 Consider some clause $C_j$ and its associated clause vertex $\hat{c}_j^\gamma$ in region $R^\gamma$. Let $\ell = (\gamma - 1)\nclss + j$. Since $X$ dominates the clause vertex, there is some $i$ so that $z_{i,X \cap P_i^\ell}^\ell \in X \cap N(\hat{c}_j^\gamma)$. By construction, this means that the state of $X$ on $P_i^\ell$ corresponds to a satisfying truth assignment of clause $C_j$. By the previous paragraph this is the same state as the state of $X$ on $P_i^{\gamma \nclss}$ which in turn corresponds to $\tassign_i$, hence $\tassign$ must satisfy clause $C_j$. Since the choice of $C_j$ was arbitrary, it follows that $\formula$ is satisfiable.
\end{proof}

\begin{lem}
 \label{thm:tds_ctw_bound}
 The cutwidth of $G(\formula)$ is at most $\nvars/2 + 4\clss + \Oh(1)$.
\end{lem}

\begin{proof}
  Consider the linear layout of $G(\formula)$ that starts with the vertices $h_4, h_3, h_2, h_1$. Then we go through $G(\formula)$ column by column and for each column $\ell \in [\nclss(3\frac{\nvars}{2} + 1)]$, the linear layout first contains the clause vertex of the $\ell$-th column, and then the vertices of $B^\ell_1$, $B^\ell_2$, $\ldots$, $B^\ell_{\nvars / 2}$ in this order. Finally, the layout ends with the vertices $h'_1, h'_2, h'_3, h'_4$. 
  
  Consider any vertex $v^* \in V(G(\formula))$ and the edges crossing the cut that comes directly after $v^*$. If $v^* \in \{h_4, h_3, h_2, h'_1, h'_2, h'_3, h'_4\}$, then at most two edges cross the cut. If $v^* = h_1$, then $\nvars / 2$ edges cross the cut, namely those from $h_1$ to the block gadgets $B^1_i$ for all $i \in [\nvars / 2]$. Now, suppose that $v^*$ is a vertex in column $\ell \in [\nclss(3\frac{\nvars}{2} + 1)]$.
  
  If $v^* = \hat{c}^\gamma_j$ is the clause vertex in column $\ell = (\gamma - 1)\nclss + j$, then the cut consists of the edges between $v^*$ and its neighbors and the $\nvars / 2$ matching edges between column $\ell - 1$ and column $\ell$ (or the $n/2$ edges coming from $h_1$ if $\ell = 1$). Since each clause has size at most $\clss$, the clause vertex $\hat{c}^\gamma_j$ is adjacent to vertices in at most $\clss$ different block gadgets. By construction, $v^* = \hat{c}^\gamma_j$ is adjacent to at most $4$ vertices in each block gadget. Hence, this cut consists of at most $\nvars / 2 + 4 \clss$ edges.
  
  If $v^* \neq \hat{c}^\gamma_j$ is not the clause vertex in column $\ell = (\gamma - 1)\nclss + j$, then $v^*$ lies in some block gadget $B^\ell_i$ with $i \in [\nvars / 2]$. By construction, every other block gadget is completely contained on one side of the cut, so the cut contains edges that are internal to at most one block gadget. Since each block gadget consists of a constant number of vertices, this will account for a constant number of edges. Furthermore, the cut contains at most $\nvars / 2$ matching edges, one for each group. More precisely, for the block gadgets of column $\ell$ that appear before $B^\ell_i$ in the layout, the edge from the exit vertex to the entry vertex of the next column is in the cut, and for the block gadgets of column $\ell$ that appear after $B^\ell_i$, the edge to the previous column is in the cut. The block gadget $B^\ell_i$ has two edges to other columns that can be in the cut. Finally, the edges from block gadgets of column $\ell$ to $\hat{c}^\gamma_j$ can be in the cut, but these are bounded by $4 \clss$ as before. In total, the cut contains at most $\nvars / 2 + 4 \clss + \Oh(1)$ edges.
\end{proof}

\begin{proof}[Proof of \cref{thm:tds_lb}]
  Suppose there is an algorithm $\algo$ that solves \TDS in time $\Oh^*((4-\eps)^{\ctw(G)})$ for some $\eps > 0$. We show how to solve $\clss$-\SAT in time $\Oh^*((2 - \eps')^\nvars)$ for some $\eps' > 0$ for all $\clss$, thus contradicting \SETH by \cref{thm:seth_hitting_set}. Given a $\clss$-\SAT instance $\formula$, we construct $G(\formula)$ in polynomial time and then run $\algo$ on $G(\formula)$ and return its answer. This is correct by \cref{thm:sat_to_tds} and \cref{thm:tds_to_sat}. Due to \cref{thm:tds_ctw_bound}, the running time is $\Oh^*((4-\eps)^{\ctw(G(\formula))}) \leq \Oh^*((4-\eps)^{\nvars/2 + 4\clss + \Oh(1)}) \leq \Oh^*((4-\eps)^{\nvars/2}) \leq \Oh^*((2-\eps')^\nvars)$ for some $\eps' > 0$, where we use $\clss \in \Oh(1)$ in the second inequality.
\end{proof}

\section{Conclusion}
\label{sec:conclusion}

Our main results are the two lower bounds for \DTC, which apply also to parameterization by treewidth resp.\ cliquewidth but use much more restrictive structure; this greatly refines what was known for \OCT, i.e.\ $\ncols=2$, and gives new tight bounds for $\ncols \geq 3$. In particular, beyond the above-mentioned examples, these are further natural problems where a small modulator to a simple graph class (of constant treewidth) is as hard as small treewidth. Surprisingly perhaps, something even stronger holds for clique-width: To get the tight lower bound, a modulator with few (true) twinclasses suffices, i.e., we need neither a sequence of disjoint separators nor complex dense structure. For \DS, only the latter was established: twinclass-cutwidth rather than cliquewidth suffices to take us from base $3$ in the running time to base $4$.

Such results bring several benefits: (1) Rather than e.g.\ getting only the isolated result of (conditional) complexity of a problem relative to treewidth, we get a much larger range of input structure that exhibits the same tight complexity. 
(2) At the same time, by aiming for maximally restricted lower bound structure, we get a much better understanding of what structure makes a given problem hard. This in turn helps to focus efforts at faster algorithms through (even) stronger structural restrictions on the input.

An immediate follow-up question is whether there are improved algorithms for \DTC when $G- \modulator$ has treewidth less than $\ncols$; so far, this is known only for \OCT, but we think such algorithms exist in general. We observe that any construction relying on $(\ncols + 1)$-critical graphs must have treewidth at least $\ncols$, hence improving upon the treewidth of our construction requires a fundamentally different idea.

Similarly, is there a meaningful restriction of (linear) clique-width, for which Lampis'~\cite{Lampis20} lower bound for $\ncols$-\COL already holds? Much more broadly, what other classes of problems exhibit the same lower bound as for treewidth already relative to deletion distance to a sparse graph class? Are there problems where this jump in complexity happens later, say, for treedepth, for some elimination distance, or only for treewidth/pathwidth? E.g., what is the complexity of \DS relative to deletion distances, and the complexity relative to treedepth may be an interesting stepping stone? Similarly, to what generality do we get the same lower bound as for clique-width already relative to, e.g., twinclass-pathwidth?

\newpage

\bibliography{modlb}

\newpage

\appendix

\input{problemdefs}

\end{document}